\documentclass[final]{siamltex}

\def\d{\delta}

\def\norm#1{\|#1\|} 
\def\wb{{\bar w}}

\makeatletter
   \def\marginparright{\@mparswitchfalse}
   \def\marginparoutside{\@mparswitchtrue}
\makeatother

\usepackage{graphicx}

\usepackage{amssymb}
\usepackage{amsmath}
\usepackage{float}
\usepackage[pdfborder={0 0 0}, colorlinks=true, linkcolor=blue, 
citecolor=red]{hyperref}
\hypersetup{pdfborder=0 0 0}
\usepackage{algorithm}
\usepackage{algorithmic}
\usepackage{mathrsfs}
\usepackage{subfigure}

\usepackage[latin1]{inputenc}
\usepackage{showkeys,cite,comment}

\newcommand{\bs}[1]{\boldsymbol{#1}}

\newcommand{\numberthis}{\addtocounter{equation}{1}\tag{\theequation}}

\usepackage{xcolor}

\begin{document}
\marginparright

\title{A Randomized Maximum A Posteriori Method for \\ Posterior
  Sampling of High Dimensional \\ Nonlinear Bayesian Inverse Problems}


\author{Kainan Wang \footnotemark[1] \and Tan Bui-Thanh \footnotemark[2]  
  \and Omar Ghattas \footnotemark[3]}

\renewcommand{\thefootnote}{\fnsymbol{footnote}}

\footnotetext[1]{previously Institute for Computational Engineering \& Sciences,
  The University of Texas at Austin, Austin, TX 78712, USA, now at Halliburton-Landmark, Houston, TX 77072, USA.}
\footnotetext[2]{Department of Aerospace Engineering and Engineering
  Mechanics, and Institute for Computational Engineering \& Sciences,
  The University of Texas at Austin, Austin, TX 78712, USA.}
\footnotetext[3]{Institute for Computational Engineering \& Sciences,
  Jackson School of Geosciences, and Department of Mechanical
  Engineering, The University of Texas at Austin, Austin, TX 78712,
  USA}

\renewcommand{\thefootnote}{\arabic{footnote}}

\bibliographystyle{siam}
\newcommand{\TODO}[1]{ \fbox{\parbox{3in}{\bf TODO: #1}}}

\newcommand{\grbf}[1] {\mbox{\boldmath${#1}$\unboldmath}}
\newcommand{\gbf}[1] {\mathbf{#1}}

\newcommand{\beq} {\begin{equation}}
\newcommand{\eeq} {\end{equation}}
\newcommand{\bdm} {\begin{displaymath}}
\newcommand{\edm} {\end{displaymath}}
\newcommand{\bit}{\begin{itemize}}
\newcommand{\eit}{\end{itemize}}
\newcommand{\bde}{\begin{description}}
\newcommand{\ede}{\end{description}}
\newcommand{\bce}{\begin{center}}
\newcommand{\ece}{\end{center}}
\newcommand{\ben} {\begin{enumerate}}
\newcommand{\een} {\end{enumerate}}
\newcommand{\bea} {\begin{eqnarray}}
\newcommand{\eea} {\end{eqnarray}}
\newcommand{\barr} {\begin{array}}
\newcommand{\earr} {\end{array}}
\newcommand{\bean} {\begin{eqnarray*}}
\newcommand{\eean} {\end{eqnarray*}}
\newcommand{\edoc} {\end{document}}
\newcommand{\On}{\hspace{0.2cm} \mbox{on} \hspace{0.2cm}}
\newcommand{\In}{\hspace{0.2cm} \mbox{in} \hspace{0.2cm}}
\newcommand{\Minimize}{\hspace{0.2cm} \mbox{minimize} \hspace{0.2cm}}
\newcommand{\Maximize}{\hspace{0.2cm} \mbox{maximize} \hspace{0.2cm}}
\newcommand{\SubjectTo}{\hspace{0.2cm} \mbox{subject} \hspace{0.15cm}
            \mbox{to} \hspace{0.2cm}}

\def\etal{{\it et al.~}}

\newcommand{\point}[1] {\ensuremath{\boldsymbol{#1}}}
\newcommand{\vvec}[1] {\ensuremath{\boldsymbol{#1}}}
\renewcommand{\vec}[1] {\ensuremath{\boldsymbol{#1}}}
\newcommand{\cvec}[1] {\ensuremath{\boldsymbol{{#1}}}}
\newcommand{\ten}[1] {\ensuremath{\boldsymbol{#1}}}
\newcommand{\tenfour}[1] {\ensuremath{\boldsymbol{\mathsf{#1}}}}
\newcommand{\comp}[1]{\begin{pmatrix}{ #1 }\end{pmatrix}}
\newcommand{\vv}{\ensuremath{\vec{v}}}
\newcommand{\vr}{\left(\ensuremath{\vec{r}}\right)}
\newcommand{\att}{\left(t\right)}
\newcommand{\tvv}{\ensuremath{\tilde{\vec{v}}}}
\newcommand{\ww}{\ensuremath{\vec{w}}}
\newcommand{\GG}{\ensuremath{\ten{G}}}
\newcommand{\HH}{\ensuremath{\ten{H}}}
\newcommand{\EE}{\ensuremath{\vec{E}}}
\newcommand{\FF}{\ensuremath{\boldsymbol{\mathfrak{F}}}}
\newcommand{\tEE}{\ensuremath{\tilde{\vec{E}}}}
\newcommand{\nn}{\ensuremath{\vec{n}}}
\renewcommand{\SS}{\ensuremath{\ten{S}}}
\newcommand{\tSS}{\ensuremath{\tilde{\ten{S}}}}
\newcommand{\cp}{\ensuremath{{c_p}}}
\newcommand{\cs}{\ensuremath{{c_s}}}
\newcommand{\tcs}{\ensuremath{{\tilde{c}_s}}}
\newcommand{\Vp}{\ensuremath{{V^e_p}}}
\newcommand{\Vs}{\ensuremath{{V^e_s}}}
\newcommand{\Sp}{\ensuremath{{S^e_p}}}
\newcommand{\Ss}{\ensuremath{{S^e_s}}}
\newcommand{\nxnx}[1] {\ensuremath{\nn\times\left(\nn\times{#1}\right)}}

\newcommand{\dd}[2] {\ensuremath{\frac{\partial {#1}}{\partial {#2}}}}
\newcommand{\DD}[2] {\ensuremath{\frac{d {#1}}{d {#2}}}}
\newcommand{\Grad} {\ensuremath{\nabla}}
\newcommand{\Gradv}[1]{\ensuremath{\nabla_{#1}}}
\newcommand{\Div} {\ensuremath{\nabla\cdot}}
\newcommand{\Divv}[1] {\ensuremath{\nabla_{#1}\cdot}}
\newcommand{\Curl} {\ensuremath{\nabla\times}}
\newcommand{\Cten}{\ensuremath{\tenfour{C}}}

\newcommand{\eval}[2][\right]{\relax
  \ifx#1\right\relax \left.\fi#2#1\rvert}

\providecommand{\abs}[1]{\left\lvert#1\right\rvert}
\providecommand{\norm}[1]{\left\lVert#1\right\rVert}

\newcommand{\Nel} {\ensuremath{{N_\text{el}}}}
\newcommand{\Ned} {\ensuremath{{N_\text{ed}}}}
\newcommand{\De} {\ensuremath{{\mathsf{D}^e}}}
\newcommand{\Dep} {\ensuremath{{\mathsf{D}^{e'}}}}
\newcommand{\Dhat} {\ensuremath{\hat{\mathsf{D}}}}
\newcommand{\Dehat} {\ensuremath{{\hat{\mathsf{D}}^e}}}
\newcommand{\half} {\ensuremath{\frac{1}{2}}}
\newcommand{\IN} {\ensuremath{{\mathcal{I}_N}}}
\newcommand{\INe} {\ensuremath{{\mathcal{I}_{N_e}}}}

\newcommand{\mc}[1]{\mathcal{#1}}
\newcommand{\mcC}[1]{\mathcal{#1}}
\newcommand{\mb}[1]{\mathbf{#1}}
\newcommand{\mbb}[1]{\mathbb{#1}}
\newcommand{\mi}[1]{\mathit{#1}}
\newcommand{\mf}[1]{\mathfrak{#1}}
\newcommand{\Oi}[1]{\int_{\mc{D}} #1 \, d\vec{x}}
\newcommand{\TOi}[1]{\int_{0}^T \int_{\mc{D}} #1 \, d\vec{x}dt}
\newcommand{\Ti}[1]{\int_{0}^T #1 \, dt}
\newcommand{\TDei}[1]{\int_{0}^T\int_{\De} #1 \, d\vec{x}dt}
\newcommand{\TDeid}[1]{\int_{0}^T\int_{\De,N} #1 \, d\vec{x}dt}
\newcommand{\Dei}[1]{\int_{\De} #1 \, d\vec{x}}
\newcommand{\DeI}[1]{\int_{D} #1 \, d\vec{x}}
\newcommand{\DeiM}[1]{\int_{\Dhat} #1 \, d\vec{r}}
\newcommand{\DeMd}[1]{\int_{\De,N_e} #1 \, d\vec{x}}
\newcommand{\DeiMd}[1]{\int_{\Dhat,N_e} #1 \, d\vec{r}}
\newcommand{\DeMdN}[1]{\int_{D,N} #1 \, d\vec{x}}
\newcommand{\DeiMdN}[1]{\int_{\Dhat,N} #1 \, d\vec{r}}
\newcommand{\DeiMdNC}[2]{\int_{\Dhat,N_{#2}} #1 \, d\vec{r}}
\newcommand{\pDei}[1]{\int_{\partial \De} #1 \, d\vec{x}}
\newcommand{\pDeiM}[1]{\int_{\partial \Dhat} #1 \, d\vec{r}}
\newcommand{\pDeiMd}[1]{\int_{\partial \Dhat,N_e} #1 \, d\vec{r}}
\newcommand{\pDeiMdN}[1]{\int_{\partial \Dhat,N} #1 \, d\vec{r}}
\newcommand{\pDeiMdNC}[2]{\int_{\partial \Dhat,N_{#2}} #1 \, d\vec{r}}
\newcommand{\pMortar}[2]{\int_{\mc{M}_{#2}} #1 \, d\vec{r}}
\newcommand{\pMortard}[2]{\int_{\mc{M}^{#2},N_{#2}} #1 \, d\vec{r}}
\newcommand{\pMortardd}[3]{\int_{\mc{M}^{#2},N_{#3}} #1 \, d\vec{r}}
\newcommand{\pMortardh}[3]{\int_{\mc{M}_{#2},N_{#3}} #1 \, d\vec{r}}
\newcommand{\pMortardhn}[4]{\int_{{#2}_{#3},N_{#4}} #1 \, d\vec{r}}
\newcommand{\TpDei}[1]{\int_{0}^T\int_{\partial \De} #1 \, d\vec{x}}
\newcommand{\TpDeid}[1]{\int_{0}^T\int_{\partial \De,N} #1 \, d\vec{x}}
\newcommand{\Deid}[1]{\int_{\De,N_e} #1 \, d\vec{x}}
\newcommand{\DeidN}[1]{\int_{\De,N} #1 \, d\vec{x}}
\newcommand{\pDeidN}[1]{\int_{\partial \De,N} #1 \, d\vec{x}}
\newcommand{\pDeid}[1]{\int_{\partial \De,N_e} #1 \, d\vec{x}}
\newcommand{\nor}[1]{\left\| #1 \right\|}
\newcommand{\snor}[1]{\left| #1 \right|}
\newcommand{\LRp}[1]{\left( #1 \right)}
\newcommand{\LRs}[1]{\left[ #1 \right]}
\newcommand{\LRa}[1]{\left< #1 \right>}
\newcommand{\LRc}[1]{\left\{ #1 \right\}}
\newcommand{\pp}[2]{\frac{\partial #1}{\partial #2}}
\newcommand{\ppcp}[1]{\frac{\partial #1}{\partial \vec{c_p}}}
\newcommand{\ppcpj}[1]{\frac{\partial #1}{\partial c_p}}
\newcommand{\ppcpmf}[1]{\frac{\partial #1}{\partial \vec{\mf{c_p}}}}
\newcommand{\ppm}[1]{\frac{\partial #1}{\partial \vec{m}}}
\newcommand{\ppho}[3]{\frac{\partial^{#3} #1}{{\partial #2}^{#3}}}
\newcommand{\ppmi}[1]{\frac{\partial #1}{\partial m_i}}
\newcommand{\ppmj}[1]{\frac{\partial #1}{\partial m_j}}
\newcommand{\ppcpm}[1]{\frac{\partial #1}{\partial \vec{c_p}^-}}
\newcommand{\ppcpp}[1]{\frac{\partial #1}{\partial \vec{c_p}^+}}
\newcommand{\ppcs}[1]{\frac{\partial #1}{\partial \vec{c_s}}}
\newcommand{\ppcsm}[1]{\frac{\partial #1}{\partial \vec{c_s}^-}}
\newcommand{\ppcsp}[1]{\frac{\partial #1}{\partial \vec{c_s}^+}}
\newcommand{\td}[2]{\frac{d #1}{d #2}}
\newcommand{\Rvert}[1]{\left. #1 \right|}
\newcommand{\ipDed}[3]{\LRp{ #1, #2}_{\De,N}^{#3}}

\newcommand{\bigO}{\mathcal{O}}

\newcommand{\iOm}[1]{\int_{\Omega} #1 \, d\Omega}
\newcommand{\iGb}[1]{\int_{\Gamma_{b}} #1 \, ds}
\newcommand{\iGs}[1]{\int_{\Gamma_{s}} #1 \, ds}
\newcommand{\iGsy}[1]{\int_{\Gamma_{s}} #1 \, ds(\mb{y})}
\newcommand{\RE}{\mbox{{I}}\hspace{-0.5ex}\mbox{{R}}}
\newcommand{\iC}[1]{\int_{0}^{2\pi} #1 \, d\theta}
\newcommand{\iGr}[1]{\int_{\Gamma_{\infty}} #1 \, ds}

\newcommand{\jump}[1] {\ensuremath{\LRs{\![#1]\!}}}
\newcommand{\diff}[1] {\ensuremath{\LRs{#1}}}
\newcommand{\average}[1] {\ensuremath{\LRc{\!\{#1\}\!}}}

\newcounter{tablecell}

\newcommand*{\numbercell}{%
  \stepcounter{tablecell}%
  \thetablecell. 
}

\newtheorem{propo}{Proposition}
\newtheorem{coro}{Corollary}
\newtheorem{theo}{Theorem}
\newtheorem{lem}{Lemma}

\newtheorem{defi}{definition}

\newtheorem{rema}{remark}

\newcommand{\figlab}[1]{\label{fig:#1}}
\newcommand{\eqnlab}[1]{\label{eq:#1}}
\newcommand{\theolab}[1]{\label{theo:#1}}
\newcommand{\corolab}[1]{\label{coro:#1}}
\newcommand{\propolab}[1]{\label{propo:#1}}
\newcommand{\lemlab}[1]{\label{lem:#1}}
\newcommand{\defilab}[1]{\label{defi:#1}}
\newcommand{\remalab}[1]{\label{rema:#1}}
\newcommand{\tablab}[1]{\label{tab:#1}}

\newcommand{\figref}[1]{\ref{fig:#1}}
\newcommand{\theoref}[1]{\ref{theo:#1}}
\newcommand{\defiref}[1]{\ref{defi:#1}}
\newcommand{\remaref}[1]{\ref{rema:#1}}
\newcommand{\cororef}[1]{\ref{coro:#1}}
\newcommand{\proporef}[1]{\ref{propo:#1}}
\newcommand{\lemref}[1]{\ref{lem:#1}}
\newcommand{\eqnref}[1]{\eqref{eq:#1}}
\newcommand{\alglab}[1]{\label{alg:#1}}
\newcommand{\algref}[1]{\ref{alg:#1}}
\newcommand{\seclab}[1]{\label{sect:#1}}
\newcommand{\secref}[1]{\ref{sect:#1}}
\newcommand{\tabref}[1]{\ref{tab:#1}}

\newcommand{\GM}[2]{\mc{N}\left( #1, #2 \right)}

\newcommand{\w}{w}
\renewcommand{\wb}{\mb{\w}}
\newcommand{\W}{W}
\newcommand{\C}{\mc{C}}
\renewcommand{\b}{b}
\renewcommand{\u}{u}
\newcommand{\ub}{{\bs{\u}}}
\renewcommand{\v}{v}
\newcommand{\vb}{\mb{\v}}
\newcommand{\f}{f}
\newcommand{\g}{g}
\newcommand{\gb}{\mb{\g}}
\newcommand{\pOmega}{\partial\Omega}
\renewcommand{\d}{d}
\newcommand{\db}{\mb{d}}
\newcommand{\yobs}{\db}
\newcommand{\yb}{\mb{y}}
\newcommand{\etab}{\bs{\eta}}
\newcommand{\Ltwo}{L^2\LRp{\Omega}}
\newcommand{\X}{X}
\newcommand{\x}{\mb{x}}
\renewcommand{\L}{{\mb{L}}}
\newcommand{\R}{\mathbb{R}}
\newcommand{\A}{\mb{A}}
\newcommand{\B}{\mb{B}}
\newcommand{\M}{\mb{M}}
\newcommand{\K}{\mb{K}}
\newcommand{\q}{q}
\newcommand{\V}{\mb{V}}
\newcommand{\zb}{\mb{z}}
\renewcommand{\sb}{\mb{s}}
\newcommand{\rb}{\mb{r}}
\newcommand{\sigb}{{\bs{\theta}}}
\newcommand{\eps}{{{\varepsilon}}}
\newcommand{\epsb}{{\bs{\eps}}}
\newcommand{\Ex}{\mathbb{E}}
\newcommand{\cost}{\mc{J}}
\newcommand{\halft} {\ensuremath{\frac{1}{2}}}
\newcommand{\fmap}{\mc{G}}
\newcommand{\pipost}{{\pi_{\text{post}}}}
\newcommand{\piprior}{{\pi_{\text{prior}}}}
\newcommand{\pilike}{{\pi_{\text{like}}}}
\newcommand{\I}{\mb{I}}
\newcommand{\Q}{\mc{Q}}
\newcommand{\Qb}{\mb{Q}}
\newcommand{\Vb}{\mb{\V}}
\newcommand{\cb}{\mb{c}}
\newcommand{\Mc}{{\M}}
\newcommand{\post}{{\pi\LRp{\ub | \db}}}
\newcommand{\G}{G}
\newcommand{\Gb}{\mb{\G}}

\maketitle

\begin{abstract}
We present a randomized maximum {\em a posteriori} (rMAP) method for
generating approximate samples of posteriors in high dimensional
Bayesian inverse problems governed by large-scale forward problems. We
derive the rMAP approach by: 1) casting the problem of computing the
MAP point as a stochastic optimization problem; 2) interchanging
optimization and expectation; and 3) approximating the expectation
with a Monte Carlo method. For a specific randomized data and prior
mean, rMAP reduces to the maximum likelihood approach (RML). It can
also be viewed as an iterative stochastic Newton method. An analysis
of the convergence of the rMAP samples is carried out for both linear
and nonlinear inverse problems. Each rMAP sample requires solution of
a PDE-constrained optimization problem; to solve these problems, we
employ a state-of-the-art trust region inexact Newton conjugate
gradient method with sensitivity-based warm starts. An
approximate Metropolization approach is  presented to reduce the
bias in rMAP samples. Various numerical methods will be presented to
demonstrate the potential of the rMAP approach in posterior sampling 
of nonlinear Bayesian inverse problems in high dimensions.
\end{abstract}



\pagestyle{myheadings} \thispagestyle{plain} \markboth{K. Wang,
  T. Bui-Thanh and O. Ghattas}{Randomized Maximum a Posteriori Method
  for Nonlinear Bayesian Inverse Problems}

\section{Introduction}


We consider a class of inverse problems that seek to determine a
distributed parameter in a partial differential equation (PDE) model,
from indirect observations of outputs of the model.
We adopt the framework of Bayesian inference, which accounts for
uncertainties in
observations, the map from parameters to observables via solution of
the forward model, and prior information on the parameters. In
particular, we seek a statistical description of all possible (sets
of) parameters that conform to the available prior knowledge and at
the same time are consistent with the observations via the
parameter-to-observable map. The solution of a Bayesian inverse problem is
the posterior measure, which encodes the degree of confidence
on each set of parameters as the solution to the inverse problem under
consideration.

Mathematically, the posterior is a surface in high dimensional
parameter space.  Even when the prior and noise probability
distributions are Gaussian, the posterior need not be
due to the nonlinearity of the parameter-to-observable map.
For large-scale inverse problems, exploring non-Gaussian posteriors in
high dimensions (to compute statistics such as the mean, 
covariance, and/or higher moments) is extremely challenging. 
The usual method of choice for computing statistics is Markov chain
Monte Carlo (MCMC) \cite{Hastings70,
  MetropolisRosenbluthRosenbluthEtAl53, RobertCasella05,
  KaipioSomersalo05, CotterRobertsStuartEtAl13, RobertsTweedie96,
  RobertsRosenthal97}, which judiciously
samples the posterior distribution, so that sample statistics can be
used to approximate the exact distributions. The problem, however, is
that standard MCMC methods often require millions of samples for
convergence; since each sample requires an evaluation of the
parameter-to-observable map, this entails millions of expensive
forward PDE simulations\textemdash a prohibitive proposition. On one hand, with the rapid development of parallel computing, parallel MCMC methods \cite{Wang14, Byrd10, Wilkinson05, Brockwell06, Strid10} are studied to accelerate the computation. While parallelization allows MCMC algorithms to produce more samples in a shorter time with multiple processors, such accelerations typically do not improve the mixing and convergence of MCMC algorithms.
More sophisticated MCMC methods that exploit 
the gradient and higher derivatives of the log posterior (and hence
the parameter-to-observable map) \cite{DuaneKennedyPendletonEtAl87,
  Neal10, GirolamiCalderhead11, BeskosPinskiSanz-SernaEtAl11,
  Bui-ThanhGirolami14, CuiMartinMarzoukEtAl14, CuiLawMarzouk16,
  MartinWilcoxBursteddeEtAl12, Bui-ThanhGhattas12d,
  PetraMartinStadlerEtAl14} can, on the other hand, improve the mixing, acceptance rate,
and convergence of MCMC.
Several of these methods exploit local curvature in parameter space as
captured by the Hessian operator of the negative logarithm of the
posterior. This requires manipulating the Hessian of the data misfit
functional (i.e., the negative log likelihood). The Stochastic Newton
method \cite{MartinWilcoxBursteddeEtAl12, Bui-ThanhGhattas12d,
  PetraMartinStadlerEtAl14} makes these Hessian manipulations
tractable by invoking a low rank approximation, motivated by the
theoretically-established or experimentally-observed compactness of
this operator for many large-scale ill-posed inverse problems.

However, despite its successful application to
million-parameter problems governed by expensive-to-solve PDEs
\cite{Bui-ThanhBursteddeGhattasEtAl12, IsaacPetraStadlerEtAl15}, two
barriers exist that prevent further scaling of Stochastic Newton to challenging
problems. First, even computing low rank Hessian information for every
sample in parameter space can be prohibitive.
Second, when the curvature of the
negative log posterior changes rapidly, stochastic Newton's local
Gaussian approximation may not provide a good enough model for the posterior and hence the MCMC proposal may not be effective. This may
result in low acceptance rates and excessive numbers of forward
PDE solves.

In this paper, we consider an optimization boosted sampling framework, the randomized maximum a posterior (rMAP) method that is inspired by the randomized maximum likelihood (RML)\cite{Kitanidis96, OliverHeReynolds96} and the randomize-then-optimize (RTO) approaches \cite{BardsleySolonenHaarioEtAl13}. Through computing each sample by PDE-constrained optimization \cite{HinzePinnauUlbrichEtAl09,BorziSchulz12, Reyes15,BieglerGhattasHeinkenschlossEtAl03}, it can explore the parameter space more efficiently. It can also be viewed as a nonlinear stochastic Newton method that executes multiple Newton iterations in every MCMC step to generate a
better proposal and to allow an improved acceptance rate. On the other hand, solving optimization problems is expensive, and hence we discuss several improvements and extensions to make the rMAP method more applicable towards solving real problems.  

We present our discussions in the following order. Section \secref{infiniteBayes} introduces a statistical inversion setting based on the Bayesian framework in infinite
dimensions. 
The core of the paper is Section \secref{RMLnew}. In this section, we first convert the maximum a posteriori (MAP)
problem into a stochastic programming problem, which is then solved
using sample average approximation. This rMAP method rediscovers the RML method as a special case. Results for convergence of the rMAP ensemble
using stochastic programming theory is presented and the extension of
the rMAP to infinite dimensional problems is discussed at
length. We also show that rMAP is a generalization of stochastic Newton---for linear
inverse problems, they become identical. It is worth noting that rMAP samples only
approximate the posterior distribution. Hence, we also discuss an
approximate Metropolization to reduce the bias in Section \secref{metropolis}. 
We discuss in Section \secref{FEM} a finite element discretization of
the infinite dimensional Bayes inverse problem. We also describe how to solve the optimization problem efficiently at each sampling
step. In particular, we  present a sensitivity approach to obtain ``good'' initial guesses for further accelerating the optimization procedure. 
In Section \secref{numerics}, various numerical results
showing the efficiency of proposed strategies compared to
state-of-the-art alternatives are presented for 1D analytical problems as well as 2D inverse problems
governed by the Helmholtz equation. Finally, we conclude the paper in
Section \secref{conclusions}.

\section{Infinite dimensional Bayesian inverse problem setting}
\seclab{infiniteBayes}

We consider the following generic forward model
\begin{equation*}
 \mc{B}\LRp{\u, \w} = 0, \quad \text{ in } \Omega,
\end{equation*}
which, for example, can be partial differential equations (PDEs) modeling the physical problem under consideration.
The forward problem involves solving for the forward state
$\w$ given a modeling of the distributed parameter $\u$.  In the
inverse problem, the task is to reconstruct $\u$ given some available
observations of $\w$ on parts of
the domain $\Omega$. One widely accepted model for the relationship between model parameters and observations is the additive noise model:
\begin{equation}
\eqnlab{observation}
\db = \mc{G}\LRp{\u} + \etab, 
\end{equation}
with $\db = \LRs{\d_1,\hdots,\d_K}^T$ denoting all observed data, $\mc{G} := \LRs{\w\LRp{\mb{x}_1}, \hdots,\w\LRp{\mb{x}_K}}^T$
denoting the parameter-to-observable (or forward) map, i.e., the map
from the distributed parameter $\u$ to 
the observables $\w\LRp{\mb{x}_i}$ at locations $\{x_j\}, j=1,2,\hdots,K$ and noise being represented by $\etab$, a random vector
normally distributed by $\GM{0}{\L}$ with bounded covariance matrix
$\L$. For simplicity, we take
$\L = \sigma^2\I$, where $\I$ is the identity matrix of appropriate
dimension. For notational convenience, throughout the paper we use
boldface letters for vectors and matrices and Roman letters for
infinite dimensional counterparts. For example, $\u$ denotes a function
in $\Ltwo$ while $\ub$ represents its discrete counterpart.

The inverse problem can be formulated as choosing model parameters that minimize the discrepancy between model prediction and osbservations:
\begin{align}
&\min_{\u}  \Phi\LRp{\u,\db} := \frac{1}{2}\snor{\db-\mc{G}\LRp{\u}}_\L^2
\eqnlab{Cost}
\end{align}
subject to the forward problem
\begin{equation}
 \mc{B}\LRp{\u, \w} = 0,
  \eqnlab{Forwardw}
\end{equation}
where $\snor{\cdot}_\L:= \snor{\L^{-\half}\cdot}$ denotes the weighted
Euclidean norm induced by the inner product in $\R^K$. This optimization problem is however ill-posed. An intuitive reason is that
the dimension of vector of observations $\db$ is often much smaller
than that of the parameter $\u$ (typically infinite before
discretization), and hence $\db$ provides limited information about
the distributed parameter $\u$. As a result, the null space of the
Jacobian of the parameter-to-observable map $\mb{F}$ is non-empty. In
particular, for a class of inverse problems, we have shown that the
Gauss-Newton approximation of the Hessian (which is the square of the
Jacobian, and is also equal to the full Hessian of the misfit $\Phi$
with noise-free data evaluated at the optimal parameter) is a compact
operator \cite{Bui-ThanhGhattas12a, Bui-ThanhGhattas12,
  Bui-ThanhGhattas12f}, and hence its range space is effectively
finite-dimensional.

In this paper, we choose to tackle the ill-posedness using a Bayesian
framework \cite{KaipioSomersalo05, CalvettiSomersalo07, Franklin70,
  LehtinenPaivarintaSomersalo89, Lasanen02, Stuart10, Piiroinen05}. We seek
a statistical description of all possible parameter fields $\u$ that
conform to some prior knowledge and at the same time are consistent
with the observations. The Bayesian approach accomplishes this through a statistical
inference framework that incorporates uncertainties in the observations, the
forward map $\mc{G}$ and the prior information. To begin, we postulate a Gaussian measure $\mu :=
\GM{\u_0}{\mc{C}}$ with mean function $\u_0$ and covariance operator
$\mc{C}$ on $\u$ in $\Ltwo$ where
\[
\mc{C} := \alpha^{-1}\LRp{I - \Delta}^{-s} =: \alpha^{-1}\mc{A}^{-s}, \quad \alpha > 0,
\]
with the domain of definition of $\mc{A}$ defined as
\[
D\LRp{\mc{A}} := \LRc{u \in H^2\LRp{\Omega}: \pp{u}{\mb{n}} = 0 \text{ on } 
\pOmega}.
\]
Here, $H^2\LRp{\Omega}$ is the usual Sobolev
space. Assume that the mean function $\u_0$ resides in the
Cameron-Martin space of $\mc{C}$, then one can show (see, e.g.,
\cite{Stuart10}) that the prior measure $\mu$ is well-defined when $s
> d/2$ ($d$ is the spatial dimension), and in this case, any
realization from the prior distribution $\mu$ almost surely resides in
the H\"older space $\X := C^{0,\beta}\LRp{\Omega}$ with $0 < \beta <
s/2$. That is, $\mu\LRp{X} = 1$, and the Bayesian posterior measure
$\nu$ satisfies the Radon-Nikodym derivative
\begin{equation}
\eqnlab{RadonNikodym}
\pp{\nu}{\mu}\LRp{\u|\db} \propto \exp\LRp{-\Phi\LRp{\u,\db}},
\end{equation}
if $\mc{G}$ is a continuous map from $\X$ to $\R^K$. 

The maximum a posteriori (MAP) point (see, e.g., \cite{Stuart10,
  Dashti2013} for the definition of the MAP point in infinite
dimensional settings) is 
given by
\begin{equation}
\eqnlab{MAP}
\u^{\text{MAP}} := \arg\min_\u \mc{J}\LRp{\u; \u_0, \db} :=\half\snor{\db -
    \mc{G}\LRp{\u}}^2_\L + \half \nor{\u - \u_0}^2_{\mc{C}},
\end{equation}
where $\nor{\cdot}_{\mc{C}} := \nor{\mc{C}^{-\half}\cdot}$ denotes the
weighted $\Ltwo$ norm induced by the $\Ltwo$ inner product
$\LRa{\cdot,\cdot}$. We shall also use $\LRa{\cdot,\cdot}$ to denote
the duality pairing on $\Ltwo$.

It should be pointed out that the last term in \eqnref{MAP} can be
considered as a prior-inspired regularization; the MAP point is thus a
solution to the corresponding deterministic inverse problem. However,
the Bayesian approach goes well beyond the deterministic solution to
provide a complete statistical description of the inverse solution: the
posterior encodes the degree of confidence (probability) in the estimate 
of all possible parameter fields. 

In addition to the MAP point, it is also desired to interrogate the posterior distribution for statsitcs such as conditional mean and interval estimates. This requires sampling of the distribution where empirical statistics from produced samples can approximate those of the posterior effectively. Popular sampling methods usually suffer from problems such as curse of dimensionality. On the other hand, successful computational methods for MAP estimation are studied extensively. These facts motivate us to explore sampling methods that are facilitated by MAP estimates, which we discuss in detail below.   

\section{A randomized maximum a posteriori approach}
\seclab{RMLnew} In this section we present an approach, which we shall
call randomized maximum a posteriori (rMAP) method, to compute
approximate samples for the posterior distribution. The idea is to
first randomize the cost function to cast the MAP statement
\eqnref{MAP} into a stochastic programming problem, which is then
solved using Monte Carlo method (also known as the sample average
approximation \cite{ShapiroDentchevaRuszczynski09}). The resulting
rMAP method resembles the randomized maximum likelihood (RML)
developed in \cite{Kitanidis96, OliverHeReynolds96}
as a special case. We therefore rediscover the RML method from a
completely new, i.e. stochastic programing, view point. It is this
view that allows us to provide new theoretical results on the RML
approach for nonlinear inverse problems that are previously not
available. Indeed, the fact that RML samples are exact samples of the posterior for
linear inverse problems seems to be currently the only available
result on the RML method \cite{Kitanidis96, OliverHeReynolds96,
  BardsleySolonenHaarioEtAl13}. We shall also show that the rMAP
method (will be used interchangebly with the RML method from now on)
can be considered as a means to incorporate uncertainty into the
solution of deterministic inverse approaches.  

To begin, let us consider
finite dimensional parameter space\footnote{Finite dimensionality
  could result from a discretization of distributed parameters (see,
  e.g., \cite{Bui-Thanh2015} for a constructive finite element
  discretization).} for simplicity of the exposition, i.e., $\ub,
\ub_0 \in \R^N$. The posterior measure $\nu$ in this case has the
density $\pipost$ with respect to the Lebesgue measure:
\[
\pipost \propto \pilike \times \piprior,
\]
where the likelihood is given by $\pilike \propto \exp\LRp{-\Phi\LRp{\u,\db}} = \exp\LRp{-\half\snor{\db -
    \mc{G}\LRp{\u}}^2_\L}$ and the prior by $\piprior \propto \exp\LRp{-\half \snor{\ub - \ub_0}^2_{\mc{C}}}$.
The MAP problem \eqnref{MAP} becomes
\begin{equation}
\eqnlab{MAPfinite}
\ub^{\text{MAP}} := \arg\min_\ub \mc{J}\LRp{\ub; \ub_0, \db} :=\half\snor{\db -
    \mc{G}\LRp{\ub}}^2_\L + \half \snor{\ub - \ub_0}^2_{\mc{C}},
\end{equation}
where $\mc{C} \in \R^{N\times N}$ is the covariance matrix in this
case.  To the end of the paper, we denote by $\Ex$ the expectation.
We now randomize the cost function, and hence the MAP problem
\eqnref{MAPfinite}.

\begin{lemma}
\lemlab{SP} Let $\sigb \in \R^K$ and $\epsb \in \R^N$ be two
independent random vectors distributed by $\pi_\sigb$ and $\pi_\epsb$
with zero mean, i.e. $\Ex_{\sigb}\LRs{\sigb} = \mb{0}$ and
$\Ex_{\epsb}\LRs{\epsb} = \mb{0}$. The following result holds:
\[
\mc{J}\LRp{\ub; \ub_0, \db} = \Ex_{\sigb \times \epsb}\LRs{\mc{J}^r\LRp{\ub; \ub_0, \db, \sigb, \epsb}} -
\Ex_{\sigb}\LRs{\sigb^T\sigb} - \Ex_{\epsb}\LRs{\epsb^T\epsb},
\]
where
\[
\mc{J}^r\LRp{\ub; \ub_0, \db, \sigb, \epsb} = \half\snor{\db + \sigb-
    \mc{G}\LRp{\ub}}^2_\L + \half \snor{\ub - \ub_0-\epsb}^2_{\mc{C}},
\]
with $\Ex_{\sigb \times \epsb}$ denoting the expectation with
respect to the product measure $\pi_\sigb\times\pi_\epsb$ induced by $\LRp{\sigb,\epsb}$. Consequently,
\begin{equation}
\eqnlab{SPpre}
\ub^{\text{MAP}} := \arg\min_\ub \mc{J}\LRp{\ub; \ub_0, \db} = \arg\min_\ub \Ex_{\sigb \times \epsb}\LRs{\mc{J}^r\LRp{\ub; \ub_0, \db, \sigb, \epsb}}.
\end{equation}
\end{lemma}
\begin{proof}
Since $\sigb$ and $\epsb$ are independent we have
\begin{align*}
&\Ex_{\sigb \times \epsb}\LRs{\mc{J}^r\LRp{\ub; \ub_0, \db, \sigb, \epsb}} = \half\Ex_{\sigb}\LRs{\snor{\db +
    \sigb- \mc{G}\LRp{\ub}}^2_\L} + \half \Ex_{\epsb}\LRs{\snor{\ub -
    \ub_0-\epsb}^2_{\mc{C}}}  = \\
& \mc{J}\LRp{\ub; \ub_0, \db}+ \Ex_{\sigb}\LRs{\sigb^T\L^{-1}\LRp{\db
    - \mc{G}\LRp{\ub}}} - \Ex_{\epsb}\LRs{\epsb^T\mc{C}^{-1}\LRp{\ub
    - \ub_0}} + \Ex_{\sigb}\LRs{\sigb^T\sigb} + \Ex_{\epsb}\LRs{\epsb^T\epsb},
\end{align*}
which proves the first assertion since $\Ex_{\sigb}\LRs{\sigb} =
\mb{0}$ and $\Ex_{\epsb}\LRs{\epsb} = \mb{0}$. The second assertion is
obvious since $\Ex_{\sigb}\LRs{\sigb^T\sigb}$ and
$\Ex_{\epsb}\LRs{\epsb^T\epsb}$ are constant independent of $\ub$.
\end{proof}

Lemma \lemref{SP}, particularly identity \eqnref{SPpre}, shows that the MAP point
can be considered as the solution of the following stochastic programming problem
\begin{equation}
\eqnlab{SP}
 \min_\ub\Ex_{\sigb \times \epsb}\LRs{\mc{J}^r\LRp{\ub; \ub_0, \db, \sigb, \epsb}} = 
\Ex_{\sigb \times \epsb}\LRs{\min_\ub \mc{J}^r\LRp{\ub; \ub_0, \db, \sigb, \epsb}},
\end{equation}
where we have interchanged the order of minimization and
expectation.\footnote{The conditions under which the interchange is
  valid can be consulted in \cite[Theorem 14.60]{RockafellarWetts98}.}
Our next step is to approximate the expectation on the right hand side
of \eqnref{SP} using the Monte Carlo approach (also known as the
sample average approximation \cite{ShapiroDentchevaRuszczynski09}). In
particular, with $n$ independent and identically distributed (i.i.d.)
samples $\LRp{\sigb_j,\epsb_j}$ from the product measure
$\pi_\sigb\times\pi_\epsb$ we have
\begin{equation}
\eqnlab{SPMC}
\min_\ub\Ex_{\sigb \times \epsb}\LRs{\mc{J}^r\LRp{\ub; \ub_0, \db,
    \sigb, \epsb}} \approx \frac{1}{n} \sum_{j = 1}^n\min_\ub \mc{J}^r\LRp{\ub; \ub_0, \db, \sigb_j, \epsb_j}.
\end{equation}
Let us define
\begin{equation}
\eqnlab{rMAPsample}
\ub_j := \arg\min_\ub\mc{J}^r\LRp{\ub; \ub_0, \db,
    \sigb_j, \epsb_j} =  \half\snor{\db + \sigb_j-
    \mc{G}\LRp{\ub}}^2_\L + \half \snor{\ub - \ub_0-\epsb_j}^2_{\mc{C}},
\end{equation}
and we are in the position to define the rMAP method in Algorithm
\ref{algo:rMAP}. As can be seen, the observation vector $\db$ and the prior mean
$\ub_0$ are randomized in the first two steps, which is then followed
by solving a randomized MAP problem in the third step. Finally, we take
each perturbed MAP point $\ub_j$ as an approximate sample of the
posterior $\pipost$. 
\begin{algorithm}[The rMAP algorithm]
  \begin{algorithmic}[1]
    \ENSURE Choose the sample size $n$ 
    \FOR{$j = 1,\hdots,n$}
    \STATE Draw $\epsb_j \sim \pi_\epsb$ 
    \STATE Draw $\sigb_j \sim \pi_\sigb$ 
    \STATE Compute rMAP sample $\ub_j$ via \eqnref{rMAPsample}    
    \ENDFOR
  \end{algorithmic}
\caption{The rMAP algorithm.}
\label{algo:rMAP}
\end{algorithm}

To the end of the paper, we choose the product measure as
$\pi_\sigb\times\pi_\epsb = \GM{\mb{0}}{\L}\times
\GM{\mb{0}}{\mc{C}}$, and in this case the rMAP approach becomes the
RML method \cite{Kitanidis95, OliverReynoldsLiu08, BardsleySolonenHaarioEtAl13}. That is, the RML method is a special case of our framework.
In other words, by first casting the MAP computation into a stochastic
programming problem and then solving it using the sample average
appproximation we have arrived at a constructive derivation of the RML
method. One can show that the RML samples are exactly those of the posterior
when the forward map $\mc{G}\LRp{\ub}$ is linear \cite{Kitanidis95,
  OliverReynoldsLiu08, BardsleySolonenHaarioEtAl13}. This seems to be the only
theoretical result currently available for RML. Our stochastic programming view
point shows that
the RML method is nothing more than a sample average approximation
 to the stochastic optimization problem \eqnref{SP} whose
solution is the MAP point. However, the sample average does
not converge to the MAP point, as we now show. Let us
define
\begin{equation}
\eqnlab{deterministic}
S\LRp{\ub_0, \db, \sigb, \epsb} := \arg\min_\ub \mc{J}^r\LRp{\ub;
  \ub_0, \db, \sigb, \epsb},
\end{equation}
that is, $S\LRp{\ub_0, \db, \sigb, \epsb}$ is the ``optimizer
operator''. Clearly, this operators maps a pair $\LRp{\sigb_j,
  \epsb_j}$ to an RML sample
\[
\ub_j := \arg\min_\ub \mc{J}^r\LRp{\ub;
  \ub_0, \db, \sigb_j, \epsb_j} = S\LRp{\ub_0, \db, \sigb_j, \epsb_j}.
\]
\begin{proposition}
\propolab{rMAPsamples}
Assume $S\LRp{\ub_0, \db, \sigb, \epsb}$ is measurable with respect to
the product measure $\pi_\sigb\times\pi_\epsb$, then
\[
\frac{1}{n}\sum_{j=1}^n\ub_j \stackrel{a.s.}{\to}
\Ex_{{\sigb\times\epsb}}\LRs{S\LRp{ \ub_0, \db, \sigb, \epsb}}
\]
\end{proposition}
\begin{proof}
The result is a simple consequence of the law of large numbers.
\end{proof}

Note that setting $\sigb = \mb{0}$ and $\epsb = \mb{0}$ in
\eqnref{deterministic} reveals that $S\LRp{\ub_0, \db, \mb{0},
  \mb{0}}$ is solution of a deterministic inverse problem with
prior-inspired regularization. If we view $\sigb$ and $\epsb$ as the
uncertainty in data $\db$ and the baseline (the prior mean) parameter
$\ub_0$, the rMAP method can be considered as a Monte Carlo approach to propagate
the uncertainty from $\db$ and $\ub_0$ to that of the inverse solution.

\begin{corollary}
When the forward map $\mc{G}\LRp{\ub}$ is linear, the following holds
\[
\frac{1}{n}\sum_{j=1}^n\ub_j \stackrel{a.s.}{\to} \ub^{\text{MAP}},
\]
and each rMAP sample $\ub_j$ is in fact the actual sample of the posterior.
\end{corollary}

We now extend the rMAP method to posterior distribution in function
spaces. In this case, $\C$ is a covariance operator from $\Ltwo$ to
$\Ltwo$, $\R^K \ni \sigb \sim \GM{\bs{0}}{\L}$, and $\Ltwo \ni \eps \sim
\GM{0}{\C}$. For notational convenience, let us define
\[
\hat{\db} := \db + \sigb, \quad \text{ and } \hat{\u} := \u_0 + \eps.
\]
The randomized MAP problem is now defined as
\begin{equation}
\eqnlab{umapEqn} \hat{\u}^{\text{MAP}} := \arg\min_\u \mc{J}^r\LRp{\u;
  \hat{\u}, \hat{\db}} := \half\snor{\hat{\db} - \mc{G}\LRp{\u}}^2_\L
+ \half \nor{\u}^2_{\mc{C}} + \LRa{\u,
  \hat{\u}}_{\mc{C}}.
\end{equation}
Note that the last two terms in \eqnref{umapEqn} is not the same as
the last term in \eqnref{MAP}. The reason is that the Cameron-Martin
space of $\mc{C}$ has zero measure \cite{Hairer09, PratoZabczyk92},
and hence $\hat{u}$ almost surely does not belong to this space. As a
result, the term $\half \nor{\u}^2_{\mc{C}}$ is almost surely
infinite, which should be removed as done in \eqnref{umapEqn}. On the other
hand, a solution to \eqnref{MAP} or \eqnref{umapEqn} is necessary in
the Cameron-Martin space since, otherwise, the term
$\nor{\u}^2_{\mc{C}}$ is infinite. The existence of such a solution
has been shown in \cite{Stuart10}, and hence 
\eqnref{umapEqn} is meaningful. Furthermore, the last term $\LRa{\u,
  \hat{\u}}_{\mc{C}}$ should be understood in the limit sense since
$\hat{\u} \in \Ltwo$ and the Cameron-Martin space is dense in $\Ltwo$.
Now, we are in the position to analyze the rMAP samples in function spaces.
\begin{lemma}
\lemlab{linearRML}
If the forward map $\mc{G}\LRp{\u}$ is linear in $\u$, then
$\hat{\u}^{\text{MAP}}$ is distributed by the posterior measure
\eqnref{RadonNikodym}.
\end{lemma}
\begin{proof}
To begin, assume $\mc{G}\LRp{\u}
= \B\u$. Taking the first variation of $\mc{J}\LRp{\u; \hat{\u}, \hat{\db}}$ 
with
respect to $\u$ in the direction $\tilde{\u}$ gives
\[
\LRa{\Grad\mc{J}\LRp{\u; \hat{\u}, \hat{\db}},\tilde{\u}} = \LRa{\mc{L}\u - 
\B^*\L^{-1}\hat{\db} - 
\mc{C}^{-1}\hat{\u},\tilde{\u}},
\]
where $\B^*: \R^K \to \Ltwo$ is the adjoint of $\B$ and we have defined
\[
\mc{L} := \B^*\L^{-1}\B + \mc{C}^{-1}.
\]
By definition, $\hat{\u}^{\text{MAP}}$ is a solution of
$\LRa{\Grad\mc{J}\LRp{\u; \hat{\u}, \hat{\db}},\tilde{\u}} = 0, \forall 
\tilde{\u}$. Consequently,
we have
\begin{equation}
\eqnlab{uhatMAP}
\hat{\u}^{\text{MAP}} = \mc{L}^{-1}\LRp{\B^*\L^{-1}\hat{\db} +
\mc{C}^{-1}\hat{\u}}.
\end{equation}
Since both $\hat{\u}$ and $\hat{\db}$ are Gaussian, $\hat{\u}^{\text{MAP}}$
is also a Gaussian random functions. Assume that $\hat{\db}$ and
$\hat{\u}$ are independent, after some simple algebra and manipulation
the mean of $\hat{\u}^{\text{MAP}}$ can be written as
\begin{equation}
\eqnlab{meanPostt}
\mathbb{E}\LRs{\hat{\u}^{\text{MAP}}} = \mc{L}^{-1}\LRp{\B^*\L^{-1}\db +
\mc{C}^{-1}\u_0},
\end{equation}
which is exact the MAP point in \eqnref{MAP}. Furthermore, the
covariance operator of $\hat{\u}^{\text{MAP}}$ reads
\begin{equation}
\eqnlab{covariancePostt}
\mathbb{E}\LRs{\LRp{\hat{\u}^{\text{MAP}}-\u^{\text{MAP}}}\otimes
  \LRp{\hat{\u}^{\text{MAP}}-\u^{\text{MAP}}}} = \mc{L}^{-1}.
\end{equation}

On the other hand, using conditional Gaussian measures
\cite{Stuart10}, one can show that the posterior measure $\nu$ is a
Gaussian with mean function
\begin{equation}
\eqnlab{meanPost}
\bar{\u} = \u_0 + \mc{C} \B^*\LRp{\L + \B \mc{C}\B^*}^{-1}\LRp{\db -
  \B\u_0},
\end{equation}
and covariance operator
\begin{equation}
\eqnlab{covariancePost}
\mc{C}_{post} = \mc{C} - \mc{C}\B^*\LRp{\L + \B \mc{C}\B^*}\B\mc{C}.
\end{equation}
The fact that \eqnref{meanPostt} and \eqnref{covariancePostt} are
identical to \eqnref{meanPost} and \eqnref{covariancePost},
respectively, follows directly from the ``matrix'' inversion lemma \cite{GolubVan96}.
\end{proof}

\subsection{rMAP as the stochastic Newton method for linear inverse problems}
\seclab{RMLuSN} We begin by extending the finite dimensional
stochastic Newton (SN) method in \cite{MartinWilcoxBursteddeEtAl12}
to infinite dimensions. To that end, we define the SN
proposal in function space as
\begin{equation}
\v_{SN} = \u - \LRs{\Grad^2\mc{J}\LRp{\u; \u_0, \db}}^{-1} \Grad\mc{J}\LRp{\u; 
\u_0, \db}+
\GM{0}{\LRs{\Grad^2\mc{J}\LRp{\u; \u_0, \db}}^{-1}},
\end{equation}
where, from the definition of $\mc{J}$ in \eqnref{MAP}, we define
\begin{subequations}
\begin{align}
\Grad\mc{J}\LRp{\u; \u_0, \db} &= \Grad\mc{G}^*\LRp{\u} \L^{-1} 
\LRs{\mc{G}\LRp{\u} - \db} +
\mc{C}^{-1}\LRp{\u-\u_0}, \eqnlab{gradient} \\
\Grad^2\mc{J}\LRp{u; \u_0, \db} &= \Grad\LRs{\Grad\mc{G}^*\LRp{\u}} \L^{-1} 
\LRs{\mc{G}\LRp{\u} - \db} + \Grad\mc{G}^*\LRp{\u} \L^{-1} \Grad\mc{G}\LRp{\u}+
\mc{C}^{-1}.\eqnlab{fullHessian}
\end{align}
\end{subequations}
Clearly, the infinite dimensional SN proposal reduces to that proposed
in \cite{MartinWilcoxBursteddeEtAl12} for finite dimensional
problems. Here comes the relation between rMAP and stochastic Newton methods.

\begin{lemma}
The rMAP approach is identical to the SN method for linear inverse problems.
\end{lemma}
\begin{proof}
Since the forward map is linear, i.e. $\mc{G}\LRp{\u} = \B\u$, the
posterior is a Gaussian measure as discussed above. A simple
manipulation gives
\begin{align*}
\Grad\mc{J}\LRp{\u; \u_0, \db} = \mc{L}\u - \B^*\L^{-1}{\db} - 
\mc{C}^{-1}{\u_0}, \text{ and }
\Grad^2\mc{J}\LRp{u; \u_0, \db} = \B^* \L^{-1} \B+
\mc{C}^{-1}.
\end{align*}
Consequently,
\begin{equation}
\v_{SN} = \u^{\text{MAP}} + \GM{0}{\mc{L}^{-1}},
\end{equation}
where $\u^{\text{MAP}} = \mc{L}^{-1}\LRp{\B^*\L^{-1}\db + \mc{C}^{-1}\u_0}$
as in the proof of Lemma \lemref{linearRML}. Due to the linearity of
$\mc{G}$, we only need to use one Newton iteration to obtain
$\hat{\u}^{\text{MAP}}$ and it is exactly given by \eqnref{uhatMAP}.

In order to show the equivalence between rMAP and SN, we need to prove
that $\v_{SN}$ and $\hat{\u}^{\text{MAP}}$ come from the same
distribution. But this is obvious by inspection: the mean
function and the covariance function of $\v_{SN}$ are exactly given by
\eqnref{meanPostt} and \eqnref{covariancePostt}, i.e., the mean and the
covariance of $\hat{\u}^{\text{MAP}}$.
\end{proof}

\subsection{rMAP as an iterative stochastic Newton method for nonlinear inverse 
problems}\seclab{rmapSNconnection} For nonlinear forward map, rMAP is no longer the same as the
stochastic Newton SN method. Instead, as we now show, it can be
considered as an iterative SN method (iSN) when the full Hessian is
approximated by the Gauss-Newton Hessian. To begin, we note that the
rMAP sample $\hat{\u}^{\text{MAP}}$ is a solution of the following equation
\begin{equation}
\eqnlab{RMLequation}
\Grad\mc{J}\LRp{\u; \hat{\u}, \hat{\db}} = 0,
\end{equation}
which can be solved using Newton method. Each Newton
iteration reads
\[
\u^{k+1} = \u^k - \LRs{\Grad^{2}\mc{J}\LRp{\u^k; \hat{\u}, \hat{\db}}}^{-1}
\Grad\mc{J}\LRp{\u^k; \hat{\u}, \hat{\db}}, \quad k = 1, \hdots.
\]
Now, the Gauss-Newton part of the full
Hessian \eqnref{fullHessian} is given by
\[
\Grad^2\mc{J}_{g}\LRp{u} = \Grad\mc{G}^*\LRp{\u} \L^{-1} \Grad\mc{G}\LRp{\u}+
\mc{C}^{-1},
\]
which is independent of $\u_0$ and $\db$.
The SN proposal in this case can be written as
\[
\v_{SN} = \u - \LRs{\Grad^2\mc{J}_g\LRp{\u}}^{-1} \Grad\mc{J}\LRp{\u; \u_0, 
\db}+
\GM{0}{\LRs{\Grad^2\mc{J}_g\LRp{\u}}^{-1}},
\]
with $\u$ denoting the current state of the SN Markov chain under consideration.
On the other hand, the rMAP method with Gauss-Newton Hessian and initial guess $\u^1
= \u$ can be written as
\[
\u^{k+1} = \u^k - \LRs{\Grad^{2}{\mc{J}}_g\LRp{\u^k}}^{-1}
\Grad{\mc{J}}\LRp{\u^k; \hat{\u}, \hat{\db}}, \quad k = 1, \hdots.
\]
In particular,
\begin{equation}
\eqnlab{u2}
\u^{2} = \u - \LRs{\Grad^{2}{\mc{J}}_g\LRp{\u}}^{-1}
\Grad{\mc{J}}\LRp{\u; \hat{\u}, \hat{\db}}.
\end{equation}

Now, by definition of $\hat{\u}$ and $\hat{\db}$, there exist
$\tilde{\u}$ and $\tilde{\db}$ such that
\[
\hat{\u} = \u_0 + \tilde{\u}, \text{ and } \hat{\db} = \db + \tilde{\db},
\]
where 
\[
\tilde{\u} \sim \GM{0}{\mc{C}}, \text{ and } \tilde{\db} \sim
\GM{\mb{0}}{\L}.
\]
Consequently, by linearity of $\Grad{\mc{J}}\LRp{\u; \cdot, \cdot}$ with respect to the last two arguments (see \eqnref{gradient}) we have
\[
\Grad{\mc{J}}\LRp{\u; \hat{\u}, \hat{\db}} = \Grad{\mc{J}}\LRp{\u; \u_0, {\db}} 
- \Grad\mc{G}^*\LRp{\u} \L^{-1} 
\tilde{\db} -
\mc{C}^{-1}\tilde{\u}
\]
and \eqnref{u2} becomes
\[
\u^{2} = \u - \LRs{\Grad^{2}{\mc{J}}_g\LRp{\u}}^{-1}
\Grad{\mc{J}}\LRp{\u; {\u}_0, {\db}} - 
\underbrace{\LRs{\Grad^{2}{\mc{J}}_g\LRp{\u}}^{-1}
\LRp{\Grad\mc{G}^*\LRp{\u} \L^{-1} 
\tilde{\db} +
\mc{C}^{-1}\tilde{\u}}}_{\u^\dagger}. 
\]

Next, the proof of Lemma \lemref{linearRML} shows that $\u^\dagger$ is
distributed by
$\GM{0}{\LRs{\Grad^2\mc{J}_g\LRp{\u}}^{-1}}$.
Therefore, $\u^2$ and
$\v_{SN}$ are identically distributed. The difference between the
rMAP method and SN is now clear: the SN method uses $\u^2$
as the MCMC proposal while the rMAP first continues to iterate until
\eqnref{RMLequation} is (approximately) satisfied and then takes the
last $\u_k$ as the proposal. In this sense, rMAP can be viewed as an
iterative SN method.

\subsection{Relation between rMAP and the randomize-then-optimize approach}
\seclab{rMAPandRTO} This section draws a connection between the rMAP
method and the randomize-then-optimize (RTO) approach
\cite{BardsleySolonenHaarioEtAl13}. We shall show that they are
identical for linear forward map (linear inverse problems), but they
are different if the forward map is nonlinear. We also propose a
modification for the RTO method. 

The difference between RML and RTO is best demonstrated for finite
dimensional parameter space. In this case, the $j$th
rMAP can be computed as
\begin{equation}
\eqnlab{rMAPsampleN}
\ub_j^{rMAP} :=\arg\min_\ub \half\snor{\L^{-\half}\LRp{\db + \sigb_j-
    \mc{G}\LRp{\ub}}}^2 + \half \snor{\mc{C}^{-\half}\LRp{\ub - \ub_0-\epsb_j}}^2,
\end{equation}
while the $j$th RTO sample \cite{BardsleySolonenHaarioEtAl13} can be written as
\begin{equation}
\eqnlab{RTOsample}
\ub_j^{RTO} := \arg\min_\ub  \half\nor{
\Qb^T\LRs{\begin{array}{l}
\L^{-\half}\LRp{
    \B\ub -\db - \sigb_j} \\
\mc{C}^{-\half}\LRp{\ub - \ub_0-\epsb_j}
\end{array}}
}^2,
\end{equation}
where $\Qb$ is the first factor in the ``thin''
QR factorization of
\begin{equation}
\eqnlab{QR} \overline{\Gb} := \Gb\LRp{\ub^{\text{MAP}}} :=
\LRs{\L^{-\half}\Grad\mc{G}\LRp{\ub^{\text{MAP}}},\mc{C}^{-\half}}^T = \Qb
\mb{R}
\end{equation}
evaluated at the MAP point. Due to the presence of $\mc{C}^{-1}$,
$\Gb$ has full column rank, and hence $\mb{R}$ is
invertible. {\em Clearly, rMAP samples $\ub_j^{rMAP}$ are not the same
  as RTO ones $\ub_j^{RTO}$ since they are extrema of different cost
  functions in general}.

Now, let us assume that the forward is linear, i.e. $\mc{G}\LRp{\ub} =
\B\ub$. Setting the derivative, with respect to $\ub$, of the cost
function in \eqnref{rMAPsampleN} to zero yields equation for the $j$th
rMAP sample $\ub_j^{rMAP}$:
\[
\overline{\Gb}^T \LRs{
\begin{array}{l}
\L^{-\half}\LRp{
    \B\ub -\db - \sigb_j} \\
\mc{C}^{-\half}\LRp{\ub - \ub_0-\epsb_j}
\end{array}
} = 0.
\]
Using \eqnref{QR} and the fact that $\Qb$ is orthonormal, we arrive at
\[
\overline{\Gb}^T\Qb\Qb^T \LRs{
\begin{array}{l}
\L^{-\half}\LRp{
    \B\ub -\db - \sigb_j} \\
\mc{C}^{-\half}\LRp{\ub - \ub_0-\epsb_j}
\end{array}
} = 0
\]
which is exact the equation for the $j$th RTO sample $\ub_j^{RTO}$ if
one sets the derivative, with respect to $\ub$, of the cost function
in \eqnref{RTOsample} to zero. In other words, we have shown that {\em RTO
is identical to rMAP for linear inverse problems}.

Up to this point we observe that RTO method requires a QR
factorization of $\overline{\Gb}$ which could be computationally
intractable for large-scale inverse problems in high dimensional parameter
spaces. We propose to use $\overline{\Gb}$ in place of $\Qb$. For
general forward map, the modified RTO problem reads (compared to \eqnref{RTOsample})
\begin{equation}
\eqnlab{mRTOsample}
\ub_j^{RTO} := \arg\min_\ub  \half\nor{
\overline{\Gb}^T\LRs{\begin{array}{l}
\L^{-\half}\LRp{
    \B\ub -\db - \sigb_j} \\
\mc{C}^{-\half}\LRp{\ub - \ub_0-\epsb_j}
\end{array}}
}^2,
\end{equation}
and hence
RTO samples now satisfy the following equation
\begin{equation}
\eqnlab{mRTO}
\Gb^T\LRp{\ub}\overline{\Gb}\,\overline{\Gb}^T \LRs{
\begin{array}{l}
\L^{-\half}\LRp{
    \B\ub -\db - \sigb_j} \\
\mc{C}^{-\half}\LRp{\ub - \ub_0-\epsb_j}
\end{array}
} = 0.
\end{equation}

The modified approach has a couple of advantages: 1) QR-factorization of
(possibly large-scale) $\overline{\Gb}$ is no longer needed; and 2) There
is no need to construct $\overline{\Gb}$ since all we need is its
action, which can be computed efficiently using adjoint technique.
The determinant of $\overline{\Gb}$ is necessary if the RTO density is
needed, but this is already available in the MAP computation.

\section{Metropolis-adjusted rMAP method}
\seclab{metropolis}
Recall from Lemma \lemref{linearRML} that, for linear inverse
problems, rMAP sample is exactly distributed by the posterior measure
$\nu$. When the forward map is nonlinear, Proposition
\proporef{rMAPsamples} shows that this is no longer true. In this
case, rMAP samples have bias which should be removed via, for example,
the standard Metropolization \cite{RobertCasella05}. The work in
\cite{OliverHeReynolds96} shows that, for some nonlinear test
problems, the acceptance rate is above $90\%$ and the authors proposed
to accept all rMAP samples. This simple strategy has been shown to
work well in many cases (see, e.g., \cite{OliverReynoldsLiu08,
  IglesiasLawStuart12}), though the resulting Markov chain can
over/under-estimate the actual posterior. We shall show that this is
the case for our inverse problem, and a de-biasing procedure is
necessary. An exact Metropolization has been proposed in
\cite{OliverHeReynolds96}, but it is intractable except for problems
with (very) small parameter dimension. We therefore propose an
approximate Metropolized step, and this is best illustrated using
finite dimensional framework. To that end, we replace $\hat{u}$ by
finite dimensional vector $\bs{\u}$, e.g., vector of finite element
nodal values.

Following \cite{Oliver14}, we begin by defining
\begin{equation}
\eqnlab{RMLdeltadef}
\delta = \mc{G}\LRp{\hat{\ub}^{\text{MAP}}}-\hat{\db}.
\end{equation}
Note that $\hat{\ub}^{\text{MAP}}$ also satisfies \eqnref{RMLequation}, which
for finite dimensional setting becomes
\begin{equation}
\eqnlab{RMLequationFinite}
\Grad\mc{G}^*\LRp{\hat{\ub}^{\text{MAP}}} \L^{-1} \LRs{\mc{G}\LRp{\hat{\ub}^{\text{MAP}}} - 
\hat{\db}} +
\mc{C}^{-1}\LRp{\hat{\ub}^{\text{MAP}}-\hat{\ub}} = 0.
\end{equation}

We can view the equations \eqnref{RMLdeltadef} and
\eqnref{RMLequationFinite} as definition of a map
$\mc{T}:(\hat{\ub},\hat{\db})\to(\hat{\ub}^{\text{MAP}}, \delta)$, and we
assume that this map needs to be locally invertible. This allows us to
explicitly write $\mc{T}^{-1}$ by
\begin{equation}
\eqnlab{RMLinversetransform} 
 \LRs{\begin{matrix}
  \hat{\ub}\\
  \hat{\db}
 \end{matrix}}
 = \mc{T}^{-1}(\hat{\ub}^{\text{MAP}}, \delta) = 
 \LRs{\begin{matrix}
       \hat{\ub}^{\text{MAP}} + \mc{C}\Grad\mc{G}^*\LRp{\hat{\ub}^{\text{MAP}}} \L^{-1}\delta \\
       \mc{G}\LRp{\hat{\ub}^{\text{MAP}}}-\delta
      \end{matrix}
 }.
\end{equation}
After dropping higher order terms, the corresponding Jacobian matrix $J$ is then
\begin{equation*}
J:= \frac{\partial \LRp{\hat{\ub}, \hat{\db}}}{\partial\LRp{\hat{\ub}^{\text{MAP}}, \delta}}\approx 
 \LRs{\begin{matrix}
\bs{I} & \mc{C}\Grad\mc{G}^*\LRp{\hat{\ub}^{\text{MAP}}}\L^{-1} \\
\Grad\mc{G}\LRp{\hat{\ub}^{\text{MAP}}} & -\bs{I}
\end{matrix}},
\end{equation*}
whose determinant can be written as
\begin{align*}
\snor{J} &\approx\snor{\det\LRp{\bs{I} + 
\mc{C}\Grad\mc{G}^*\LRp{\hat{\ub}^{\text{MAP}}}\L^{-1}\Grad\mc{G}\LRp{\hat{\ub}^{\text{MAP}}}}}\\
& = \snor{\det\LRp{\bs{I} + 
\mc{C}^{\half}\Grad\mc{G}^*\LRp{\hat{\ub}^{\text{MAP}}}\L^{-1}\Grad\mc{G}\LRp{\hat{\ub}^{\text{MAP}}}\mc{C}^{\half}}} \numberthis \label{JacobianDeterminant},
\end{align*}
Note that the Gauss-Newton approximation in the last equation can be readily computed by adjoint methods. Let us denote by $h\LRp{\hat{\ub}^{\text{MAP}},\delta}$ the density of proposing the pair $\LRp{\hat{\ub}^{\text{MAP}},\delta}$ with the above algorithm. It is then also the density for the probability $\mc{T}_{\#}\pi(\hat{\ub}, \hat{\db})$, i.e., the push-forward of the probability to propose the pair $\LRp{\hat{\ub}, \hat{\db}}$. By the measure preservation property and the change of variables formula we have
\begin{equation*}
 h\LRp{\hat{\ub}^{\text{MAP}},\delta} = f\LRp{\mc{T}^{-1}\LRp{\hat{\ub}^{\text{MAP}},\delta}}\snor{J},
\end{equation*}
where $f$ is defined as 
\begin{equation}
\eqnlab{RMLjointprobability}
f\LRp{\hat{\ub},\hat{\db}} \sim \exp\LRs{-\half\LRp{\hat{\ub} - 
\ub_0}^T\mc{C}^{-1}\LRp{\hat{\ub} - \ub_0} -\half \LRp{\hat{\db} - 
\db}^T\L^{-1}\LRp{\hat{\db} - \db}}.
\end{equation}
 With equations \eqnref{RMLinversetransform}, it is not hard to see that 
\begin{align*}
 f\LRp{\mc{T}^{-1}\LRp{\hat{\ub}^{\text{MAP}},\delta}} &= p\LRp{\hat{\ub}^{\text{MAP}}}\zeta\LRp{\delta}\eta\LRp{\hat{\ub}^{\text{MAP}}},
\end{align*}
where
\begin{equation*}
 p\LRp{\hat{\ub}^{\text{MAP}}} = \exp\LRp{-\half\snor{\hat{\ub}^{\text{MAP}}-\ub_0}_{\mc{C}}-\half\snor{\mc{G}\LRp{\hat{\ub}^{\text{MAP}}}-\db_0}_{\L}}
\end{equation*}
is proportional to the posterior distribution, 
\begin{equation*}
 \zeta\LRp{\delta} = \exp\LRp{-\half\LRp{\delta-\mc{H}\mc{K}}^T\mc{H}^{-1}\LRp{\delta-\mc{H}\mc{K}}},
\end{equation*}
and 
\begin{equation*}
 \eta\LRp{\hat{\ub}^{\text{MAP}}} = \exp\LRp{\half\mc{K}^T\mc{H}\mc{K}},
\end{equation*}
where $\mc{H}$ and $\mc{K}$ are given by
\begin{equation*}
 \mc{H}^{-1} = \L^{-1} + \L^{-1}\Grad\mc{G}\LRp{\hat{\ub}^{\text{MAP}}}\mc{C}\Grad\mc{G}^*(\hat{\ub}^{\text{MAP}})\L^{-1}
\end{equation*}
and
\begin{equation*}
 \mc{K} = \L^{-1}\LRp{\LRp{\mc{G}\LRp{\hat{\ub}^{\text{MAP}}}-\db_0} +\Grad\mc{G}\LRp{\hat{\ub}^{\text{MAP}}}\LRp{\hat{\ub}^{\text{MAP}}-\ub_0}}.
\end{equation*}

Since the terms including $\delta$ consititute a Gaussian kernel, such a decomposition allows us to marginalize $\delta$ and obtain the probability of proposing $\hat{\ub}^{\text{MAP}}$:
\begin{equation*}
 q\LRp{\hat{\ub}^{\text{MAP}}} =\int h\LRp{\hat{\ub}^{\text{MAP}},\delta} d\delta
 = p\LRp{\hat{\ub}^{\text{MAP}}}\eta\LRp{\hat{\ub}^{\text{MAP}}}\omega\LRp{\hat{\ub}^{\text{MAP}}}\snor{J},
\end{equation*}
where $\omega\LRp{\hat{\ub}^{\text{MAP}}}$ is from integrating with respect to $\delta$ and it possesses an explicit form:
\begin{equation*}
 \omega\LRp{\hat{\ub}^{\text{MAP}}}\propto\snor{\mc{H}}^{\half}=\snor{\L}^{-\half}\snor{J}^{-\half}.
\end{equation*}
Substituting these formulas into the decomposition of $q\LRp{\hat{\ub}^{\text{MAP}}}$, we obtain the ratio of posterior distribution over proposal distribution to be:
\begin{equation*}
 \theta\LRp{\hat{\ub}^{\text{MAP}}} = \frac{p\LRp{\hat{\ub}^{\text{MAP}}}}{q\LRp{\hat{\ub}^{\text{MAP}}}}\propto \exp\LRp{-\half\mc{K}^T\mc{H}\mc{K}}\snor{\L}^{\half}\snor{J}^{-\half}.
\end{equation*}
With this ratio, we are able to compute the acceptance ratio between a newly proposed state $\hat{\ub}^{\text{MAP}}_*$ and a current state $\hat{\ub}^{\text{MAP}}_k$. One computational consideration in practice would be that directly computing the gradient of the forward map, $\Grad\mc{G}\LRp{\hat{\ub}^{\text{MAP}}}$, can be expensive when the number of measurements is high. A further practical simplification would be approximating $\alpha$ with only the $\snor{\L}^{\half}\snor{J}^{-\half}$. Thus, the acceptance ratio we adopt has the form
\begin{equation*}
 \tilde{\alpha}\LRp{\hat{\ub}^{\text{MAP}}_*, \hat{\ub}^{\text{MAP}}_k} = \frac{\theta\LRp{\hat{\ub}^{\text{MAP}}_*}}{\theta\LRp{\hat{\ub}^{\text{MAP}}_k}}\\
 \approx\frac{\snor{J(\hat{\ub}^{\text{MAP}}_k)}^{\half}}{\snor{J(\hat{\ub}^{\text{MAP}}_*)}^{\half}} \numberthis 
\end{equation*}
This simplification appears to be reasonable as shown in the numerical results.

It should be pointed out that we have recently shown that the misfit
(Gauss-Newton) Hessian is a compact operator
\cite{Bui-ThanhGhattas12,Bui-ThanhGhattas12a}. Moreover,
$\mc{C}^\half$ is also a compact operator by definition of Gaussian
measure. It follows that
$\mc{C}^\half\Grad^2\Phi_g\LRp{\hat{\ub}^{\text{MAP}},\hat{\db}}\mc{C}^\half$
is compact and admits low rank approximation. This is in fact one of
the key points that is exploited to construct scalable and
mesh-independent method in our previous work on extreme scale Bayesian
inversion \cite{Bui-ThanhBursteddeGhattasEtAl12,
  Bui-ThanhGhattasMartinEtAl13}. Thus, computing $\snor{J}$ can be
done in a scalable manner independent of the mesh size using the
randomized SVD technique \cite{HalkoMartinssonTropp11}, for example.


\section{Finite element discretization and optimization}
\seclab{FEM} For the practical problems we consider we assume the
spatial dimension to be at least two, therefore we choose $s > 1$ so
that the infinite dimensional framework is well-defined as discussed
in Section \secref{infiniteBayes}. As a result, evaluating the prior
and/or generating a prior sample requires to discretize and/or solve a
fractional partial differential equation.  Similar to
\cite{Bui-Thanh2015} (and references therein) we combine the finite
element method (FEM) \cite{Ciarlet78} and the matrix transfer technique (see, e.g.
\cite{IlicLiuTurnerEtAl05} to discretize (truncated) Karhunen-Lo\`eve
(KL) expansion of the prior. For the discretization of the forward
equation, and hence the likelihood, we also use the same finite element method.

Using finite element approximation, the MAP problem \eqnref{MAP}
becomes a (possibly) high dimensional and nonlinear optimization
problem. It is thus necessary to use the state-of-the-art
scalable optimization solver to minimize the cost. Here we choose the
trust region inexact Newton conjugate gradient (CG) method (TRINCG),
for which some of the main idea can be found, e.g., in
\cite{NocedalWright06, ColemanLi96, BranchColemanLi99,
  Bui-Thanh07}). The method combines the rapid locally-quadratic
convergence rate properties of Newton method, the effectiveness of
trust region globalization for treating ill-conditioned problems, and
the Eisenstat--Walker idea of preventing oversolving. In the numerical results section, we demonstrate the efficiency of this trust region method over popular Levenberg-Marquardt techniques. As we shall see that, in some difficult examples, choosing TRINCG becomes critical in controlling computation time for rMAP sampling.

\subsection{Good initial guess for the rMAP algorithm}
\seclab{goodInit}
One of the most important aspects of numerical optimization,
particularly with Newton method, is how to choose a good initial
guess. The closer the initial guess is to the basin of attraction of a
local minimum, the faster the convergence. This is clearly important
since we desire to minimize the cost of computing rMAP proposals. One
way to achieve this is through using sensitivity analysis, which we now
describe. To begin, we distinguish $\Grad$, the derivative with
respect to $\u$, with derivatives with other variables: for example,
$\Grad_{\hat{\u}_i}$ and $\Grad_{\hat{\db}_i}$ denote derivatives with
respect to $ \hat{\u}_i$ and $\hat{\db}_i$, respectively. Consider two
consecutive rMAP samples $\hat{\u}^{\text{MAP}}_i$ and $\hat{\u}^{\text{MAP}}_{i+1}$ that satisfy
\begin{align}
\eqnlab{iSample}
F\LRp{\hat{\u}^{\text{MAP}}_i; \hat{\u}_i,\hat{\db}_i} &:= 
\Grad\mc{J}\LRp{\hat{\u}^{\text{MAP}}_i; \hat{\u}_i, \hat{\db}_i} = 0, \\
\eqnlab{i1Sample}
F\LRp{\hat{\u}^{\text{MAP}}_{i+1}; \hat{\u}_{i+1},\hat{\db}_{i+1}} &:= 
\Grad\mc{J}\LRp{\hat{\u}^{\text{MAP}}_{i+1}; \hat{\u}_{i+1}, \hat{\db}_{i+1}} = 0.
\end{align}
Now, let us define
\[
\tilde{\u} = \hat{\u}_{i+1} - \hat{\u}_i, \text{ and } \tilde{\db} =
\hat{\db}_{i+1} - \hat{\db}_i.
\]
Assuming that $\hat{\u}^{\text{MAP}}_i$ is already computed from
\eqnref{iSample}, we now construct an initial guess for
 solving \eqnref{i1Sample} using Newton method:
\begin{equation}
\eqnlab{RMLinitial}
\u^{init} = \hat{\u}^{\text{MAP}}_i +
\underbrace{\LRp{\Grad_{\hat{\db}_i}\hat{\u}^{\text{MAP}}_i, \tilde{\db}} +
\LRa{\Grad_{\hat{\u}_i} \hat{\u}^{\text{MAP}}_i, \tilde{\u}}}_{T},
\end{equation}
which is simply the first order Taylor approximation of
$\hat{\u}^{\text{MAP}}_{i+1}$ around $\LRp{\hat{\u}_i, \hat{\db}_i}$.

What remains is to compute $T$ in \eqnref{RMLinitial}. To
this end, we expand the gradient in \eqnref{i1Sample} using the
first order Taylor expansion to obtain the following equation for $T$
\begin{equation}
\eqnlab{Teqn}
\Grad^2\mc{J}\LRp{\hat{\u}^{\text{MAP}}_i; \hat{\u}_i, \hat{\db}_i}T \approx
\Grad\mc{G}^*\LRp{\hat{\u}^{\text{MAP}}_i} \L^{-1}\tilde{\db} +
\mc{C}^{-1}\tilde{\u}.
\end{equation}
Solving \eqnref{Teqn} requires an adjoint solve to evaluate the right
hand side, and the inverse of $\Grad^2\mc{J}\LRp{\hat{\u}^{\text{MAP}}_i;
  \hat{\u}_i, \hat{\db}_i}$ (the Hessian evaluated at the $i$th rMAP
sample). If $\snor{\hat{\u}^{\text{MAP}}_i - \hat{\u}^{\text{MAP}}_{i+1}}$ is small,
$\u^{init}$ is a very good approximation of
$\hat{\u}^{\text{MAP}}_{i+1}$. Thus, solving \eqnref{i1Sample} with
$\u^{init}$ as the initial guess helps reduce the number of
optimization iterations (and hence the number of forward PDE solves)
substantially. In practice, we linearize around the MAP point
\eqnref{MAP} and this approach further cuts down the number of
PDE
solves since $\Grad^2\mc{J}\LRp{\u^{\text{MAP}}; \u_0, \db}$ is fixed and can
be well approximated using low rank approximation
\cite{Bui-ThanhBursteddeGhattasEtAl12,
  Bui-ThanhGhattasMartinEtAl13}.

\section{Numerical results}
\seclab{numerics}
In this section, we present sampling results using several test cases. In Section \secref{numericAnalytical}, we use two analytical functions to compare the sampling efficiency between the rMAP and the RTO method, and between the stochastic Newton method  described above. In Section \secref{numericHelmholtz}, we use the rMAP method to sample a Bayesian inverse problem on a two dimensional Helmholtz forward model. Therein, we compare the computational efficiency between the popular Levenberg-Marquardt method (see, e.g., \cite{OliverReynoldsLiu08}) and  TRINCG method for each rMAP sample, as well as the effectiveness of using a good initial guess as is discussed in Section \secref{FEM}. In order to examine statistical convergence of rMAP methods, we also compare rMAP samples with those from the delayed rejection adaptive Metropolis (DRAM) sampler \cite{HaarioLaineMiraveteEtAl06}.  
\subsection{Analytical function example}\seclab{numericAnalytical}Let us start by numerically demonstrating how rMAP
and RTO cost functions in \eqnref{rMAPsampleN} and \eqnref{RTOsample},
respectively, change the original cost function in \eqnref{MAP}. To
that end, we consider two analytical cost functions (negative log posterior)
\begin{subequations}
\eqnlab{costFunctional}
\begin{align}
\eqnlab{costFunctional1}
\mc{J}_1 &:= \frac{1}{2}\LRp{\ub - 0.8}^2 + \frac{1}{2\times0.2^2}\LRp{\ub^2-1}^2,\\
\mc{J}_2 &:= \frac{1}{2}\LRp{\ub - 1}^2 + \frac{1}{2\times0.2^2}\LRp{\ub^3-0.8}^2.
\end{align}
\end{subequations}

\subsubsection{Comparing rMAP and RTO methods}
In Figure \figref{rMAPandRTOrandomization} are the original cost
functionals $\mc{J}_1, \mc{J}_2$ and their randomization with rMAP and
RTO methods. (Note that both the original RTO and our modified version
give identical results for all analytical results, and hence we do not
distinguish them) Here, we use the same $\sigb$ and $\epsb$ for both
rMAP and RTO. As can be seen, both randomized costs preserve the
characteristics, e.g. multi-modality and skewness, of the original
one. However, they differ from the original cost function as well as from each other, which agrees with our findings
in Section \secref{rMAPandRTO}.
\begin{figure}[h!t!b!]
\begin{center}
\subfigure[$\mc{J}_1$, and its rMAP and RTO]{
\includegraphics[trim=1cm 6.0cm 2cm 7.1cm,clip=true, width=0.4\columnwidth]{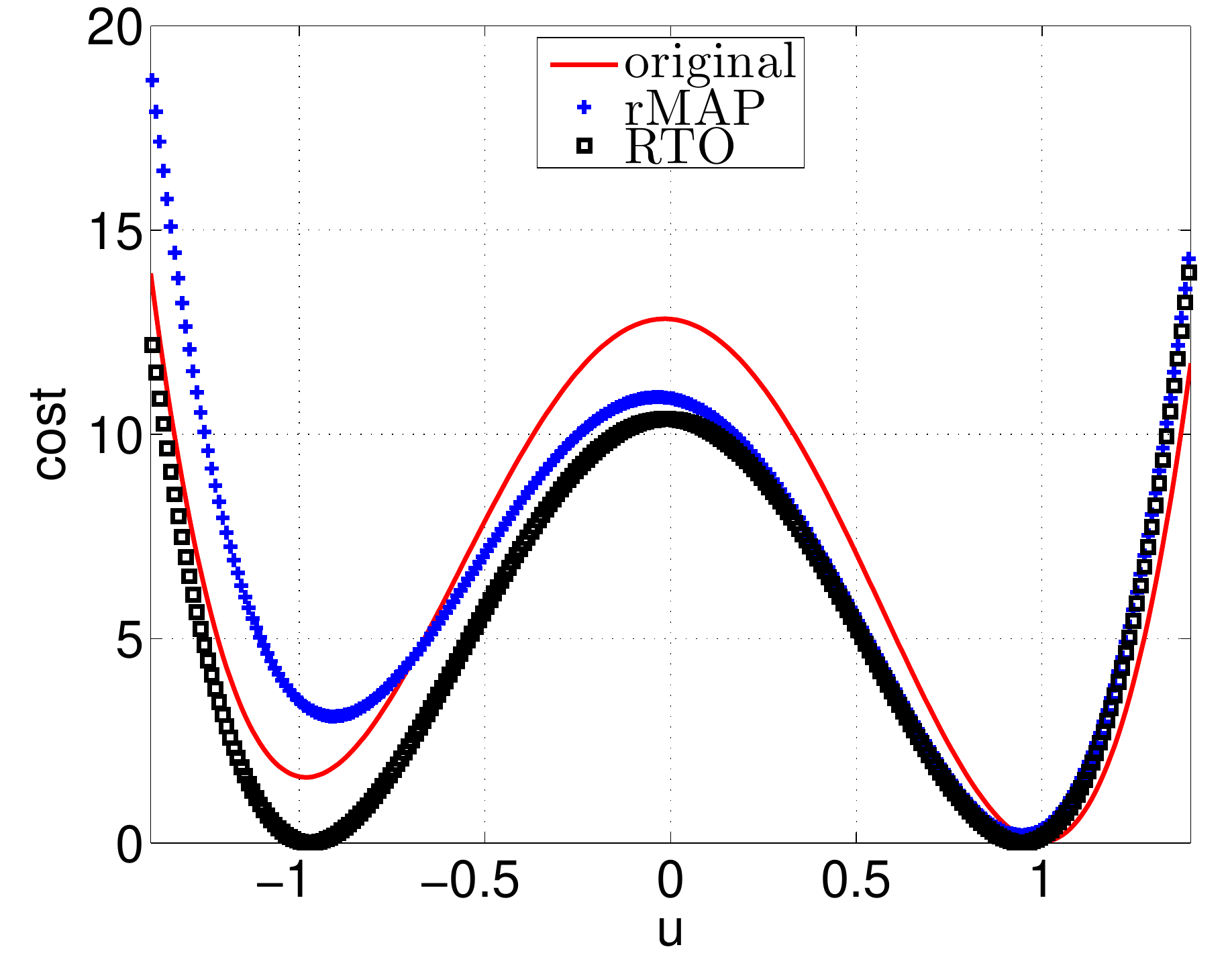}
\figlab{rMAPandRTO1}
 }
\subfigure[$\mc{J}_2$, and its rMAP and RTO]{
\includegraphics[trim=1cm 6.0cm 2cm 7.1cm,clip=true, width=0.4\columnwidth]{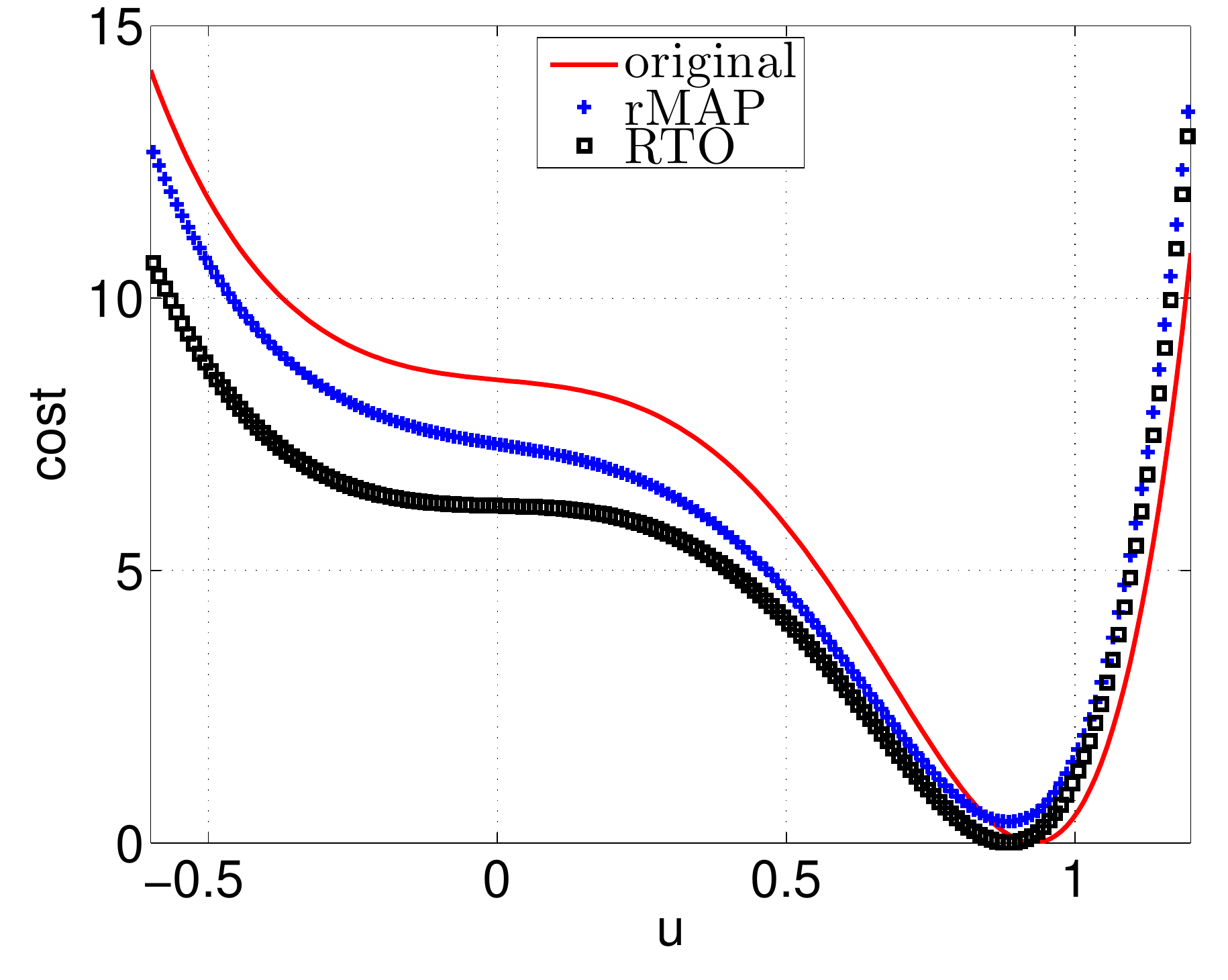}
\figlab{rMAPandRTO2}
 }

 \caption{Randomization of the cost functionals in
   \eqnref{costFunctional} with rMAP and RTO methods.}
 \figlab{rMAPandRTOrandomization}
 \end{center}
\end{figure}

We next examine the sensitivity of both rMAP and RTO with
multi-modality and optimization solver. To that end, we first use
Matlab {\em fminunc}, the unconstrained optimization solver, and use the MAP point as initial guesses to
compute rMAP and RTO samples for $\mc{J}_1$ cost functional. As can be
seen in Figures \figref{rMAPmap1} and \figref{RTOmap1}, both methods
are stuck in a mode. Instead, if we use $\hat{\ub}_j := \ub_0 +
\epsb_j$ as initial guess for computing the $j$th sample we obtain the
results in Figures \figref{rMAP1} and \figref{RTO1},
respectively. Clearly, both methods explore both modes well. Thus,
{\em for rMAP and RTO to work with local optimization solver, it is
  important that initial guesses are well distributed in the parameter
  space}. In fact, good initial guesses also help significantly reduce
the number of forward solves as we will show in the following subsection. 

As a comparison, we employ Matlab's constrained optimization solver {\em fminbnd} with prescribed bound $-100\le \ub \le
100$ to more than sufficient to cover the modes. This optimization
solver computes initial guesses using the {\em golden section
  rule}. The results for rMAP and RTO are shown in \figref{rMAP3} and
\figref{RTO3}: rMAP still works well in this case while RTO is stuck
in the left mode. Thus, rMAP seems to be more robust with optimization
solvers. From numerical experiments we observe that rMAP tends to
displace the original function more than RTO does, and this may
partially explain the robustness of the former. However, rMAP also
seems to ignite ``silent'' mode in the original function as we now
show in Figure \figref{rMAPandRTOhist2} for $\mc{J}_2$. Note that the
original cost function $\mc{J}_2$ has only one mode, but it can become
multi-modal for a range of $\epsb$ and $\sigb$. As can be observed in
Figures \figref{rMAP4} and \figref{RTO4}, rMAP puts a lot of samples
in an artificial mode that was not in the original function, while RTO
does not seem to see the same thing. With the square root Jacobian
correction in Section \secref{metropolis}, we can, in Figure
\figref{rMAP4corrected}, both remove that artificial mode and improve
the histogram for the actual mode. We can also improve the RTO samples
by first taking the RTO density as important sampling density and then
using the important weights to correct RTO samples. The result in
Figure \figref{RTO4corrected} shows that this strategy indeed provides
better histogram as well.

\begin{figure}[h!t!b!]
\begin{center}
\subfigure[rMAP: MAP initial guess]{
\includegraphics[trim=1cm 6.0cm 2cm 7.1cm,clip=true, width=0.3\columnwidth]{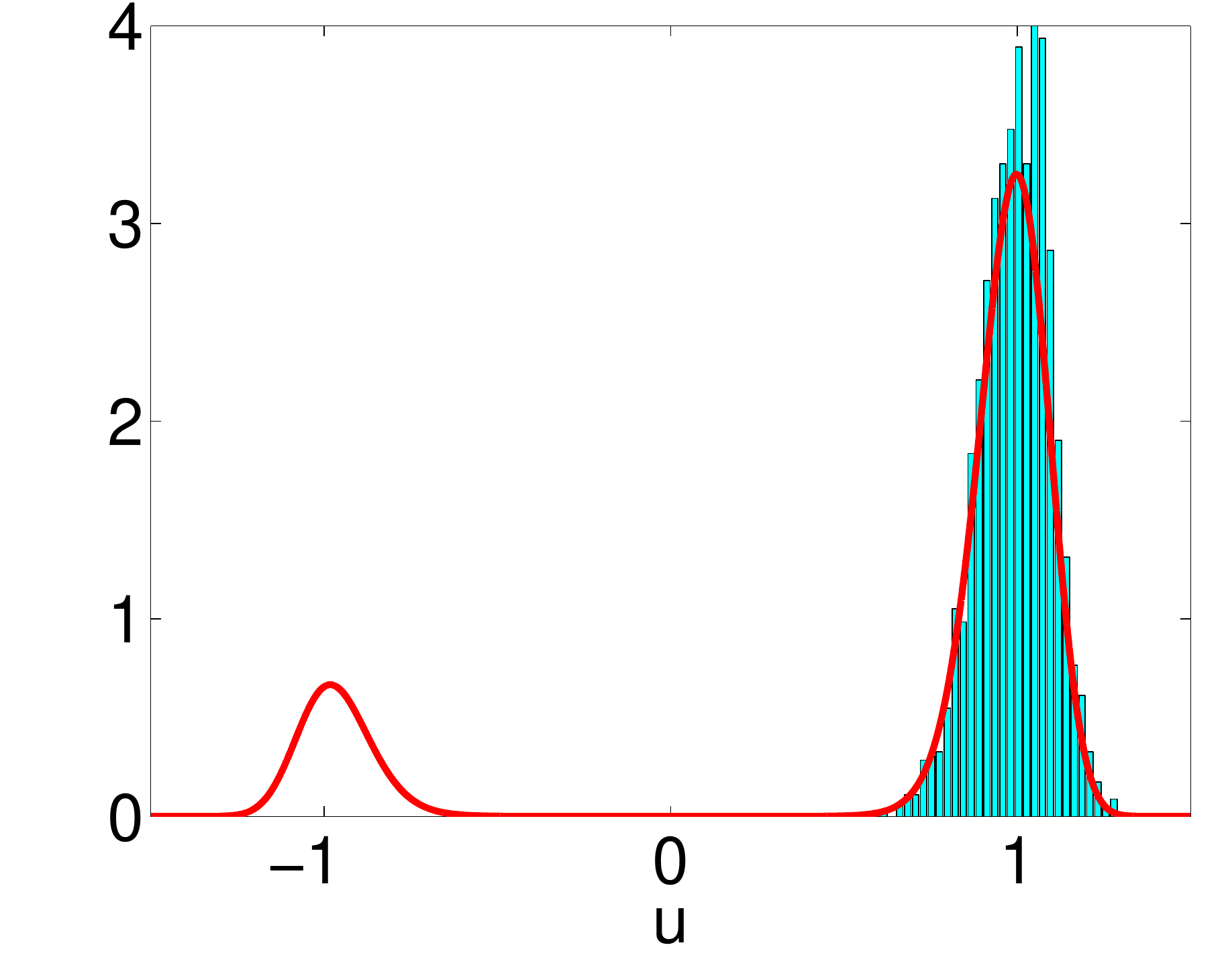}
\figlab{rMAPmap1}
 }
\subfigure[rMAP: random initial guess]{
\includegraphics[trim=1cm 6.0cm 2cm 7.1cm,clip=true, width=0.3\columnwidth]{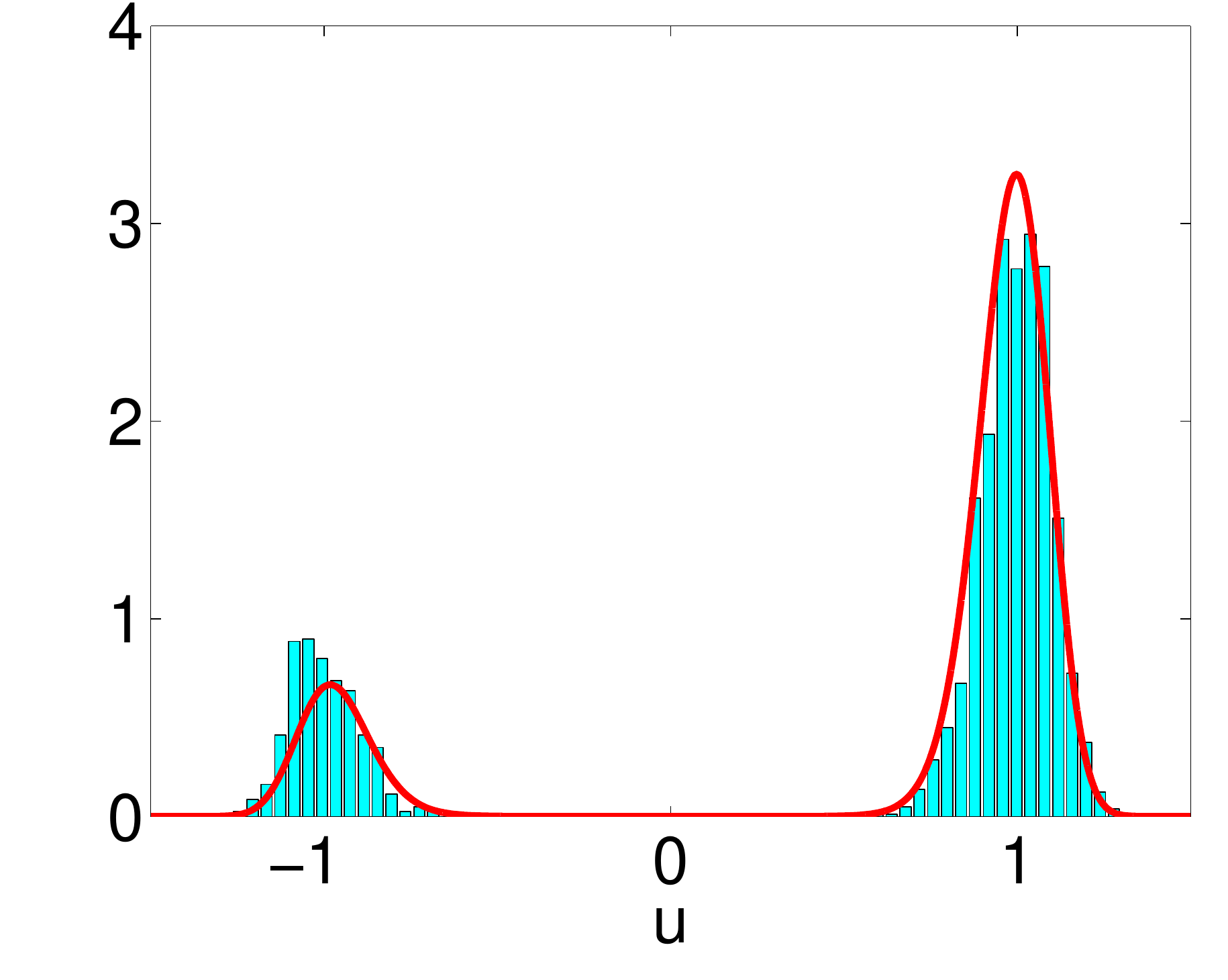}
\figlab{rMAP1}
 }
\subfigure[rMAP: ``Golden section'' initial guess]{
\includegraphics[trim=1cm 6.0cm 2cm 7.1cm,clip=true, width=0.3\columnwidth]{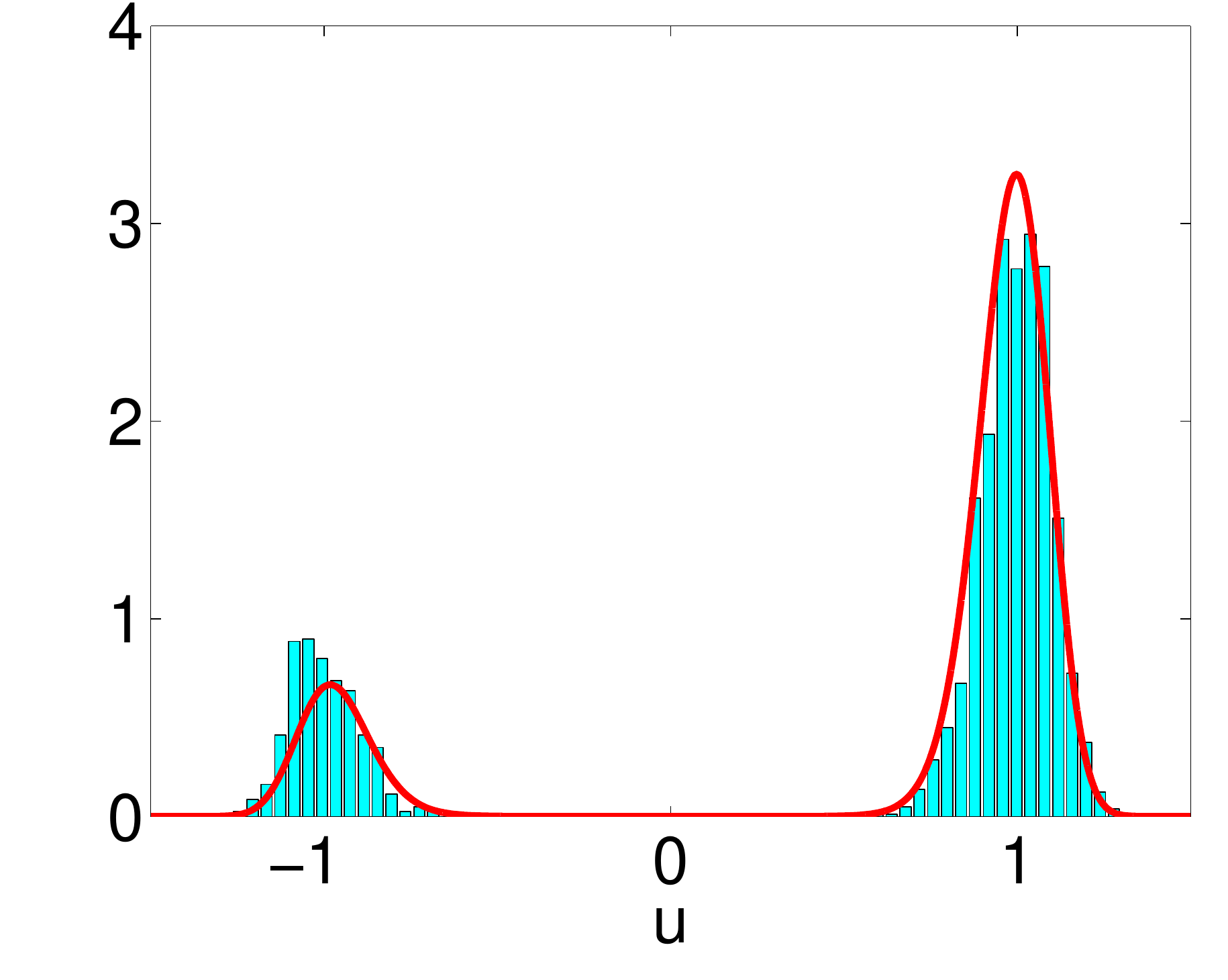}
\figlab{rMAP3}
 }
\subfigure[RTO: MAP initial guess]{
\includegraphics[trim=1cm 6.0cm 2cm 7.1cm,clip=true, width=0.3\columnwidth]{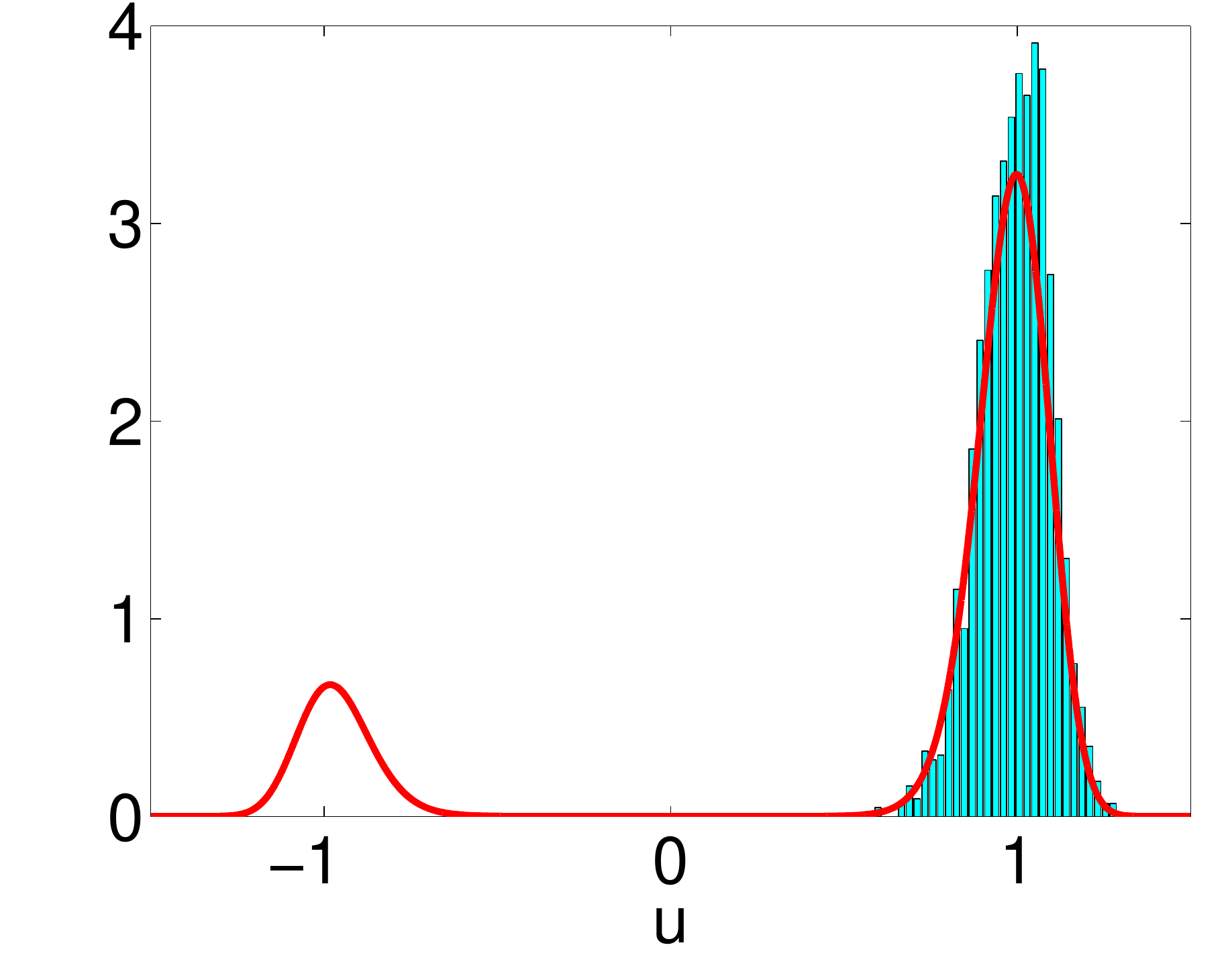}
\figlab{RTOmap1}
 }
\subfigure[RTO: random initial guess]{
\includegraphics[trim=1cm 6.0cm 2cm 7.1cm,clip=true, width=0.3\columnwidth]{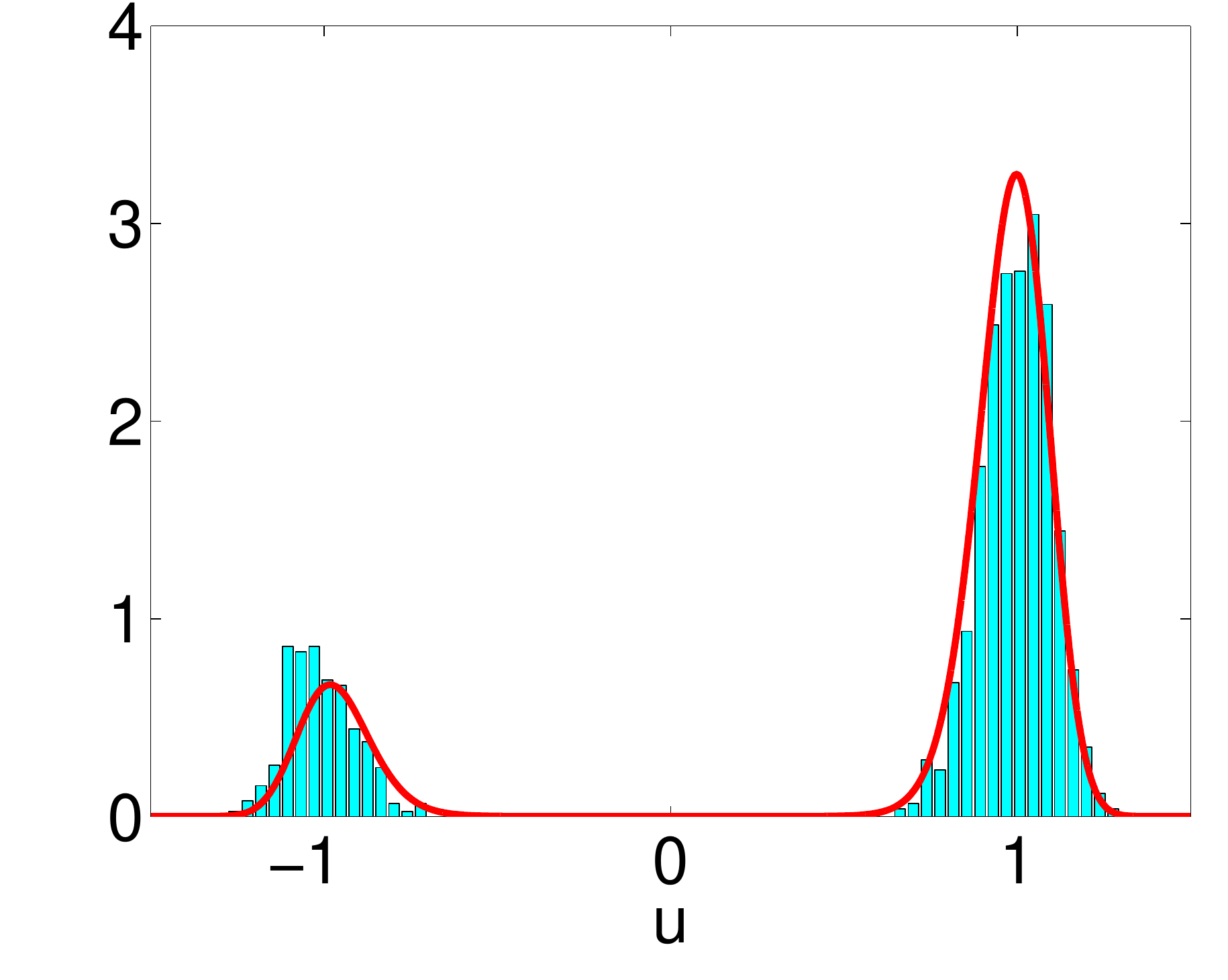}
\figlab{RTO1}
 }
\subfigure[RTO: ``Golden section'' initial guess]{
\includegraphics[trim=1cm 6.0cm 2cm 7.1cm,clip=true, width=0.3\columnwidth]{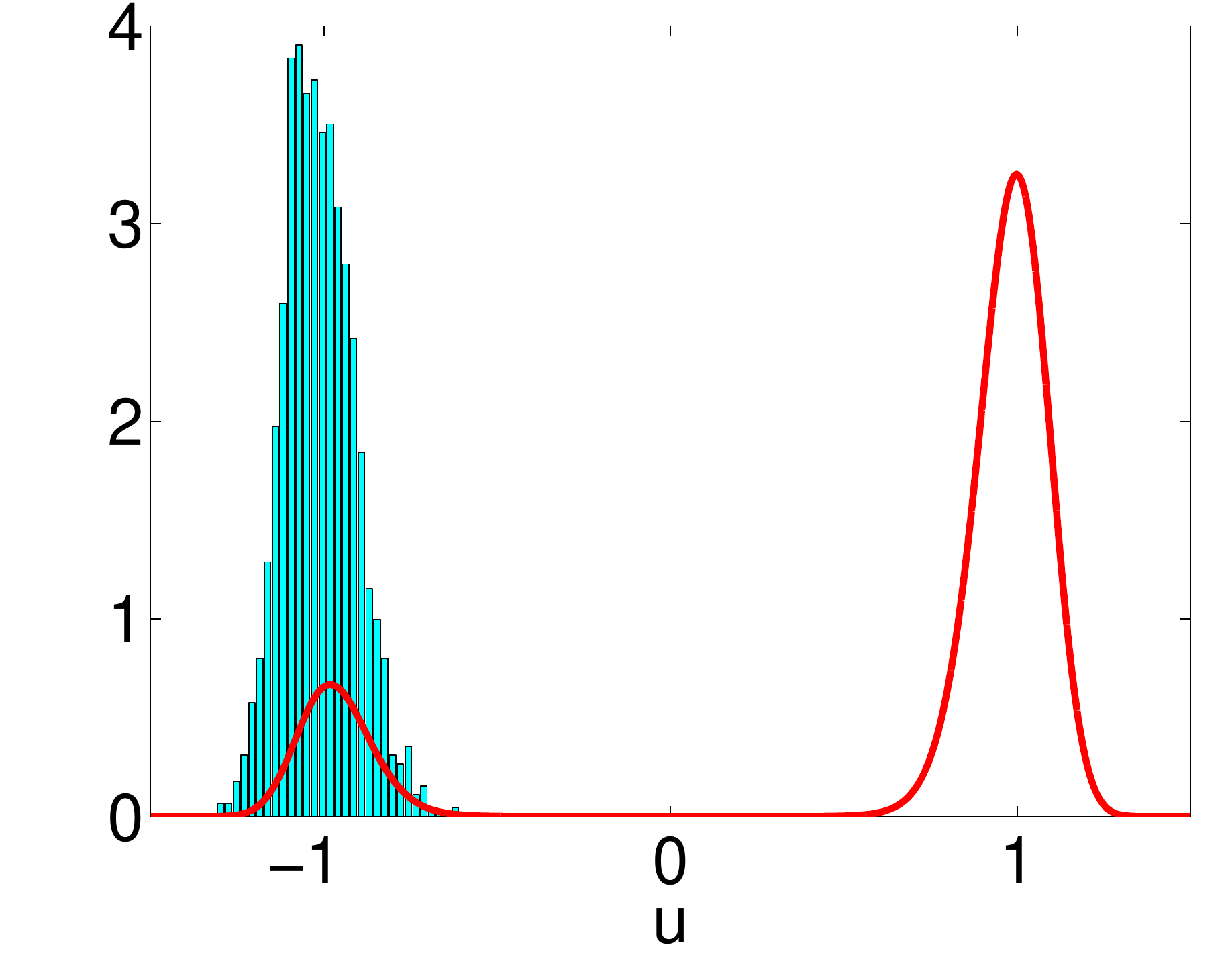}
\figlab{RTO3}
 }
\end{center}
 \caption{Sensitivity of rMAP and RTO with local optimization solvers
   and initial guesses. Figures \figref{rMAPmap1} and \figref{RTOmap1}
   are with {\em fminunc} and MAP initial guess. Figures
   \figref{rMAP1} and \figref{RTO1} are with {\em fminunc} and random
   prior means as initial guesses. Figures \figref{rMAP3} and
   \figref{RTO3} are with {\em fminbnd} and the default golden section
   rule initial guess. The cost functional that is used to conduct
   these experiments is $\mc{J}_1$.}  
 \figlab{rMAPandRTOhist1}
\end{figure}

\begin{figure}[h!t!b!]
\begin{center}
\subfigure[rMAP]{
\includegraphics[trim=1cm 6.0cm 2cm 7.1cm,clip=true, width=0.37\columnwidth]{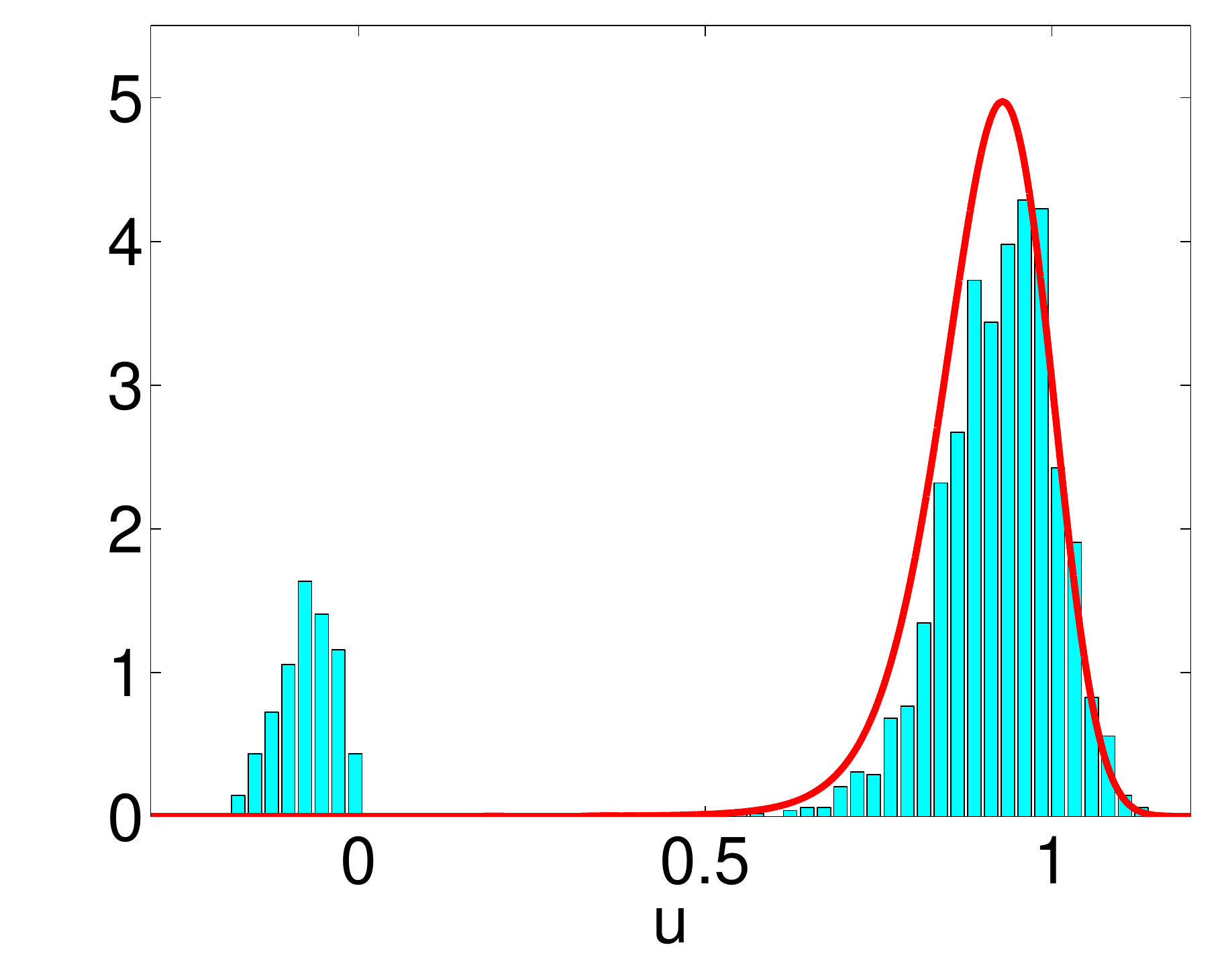}
\figlab{rMAP4}
 }
\subfigure[Corrected rMAP]{
\includegraphics[trim=1cm 6.0cm 2cm 7.1cm,clip=true, width=0.37\columnwidth]{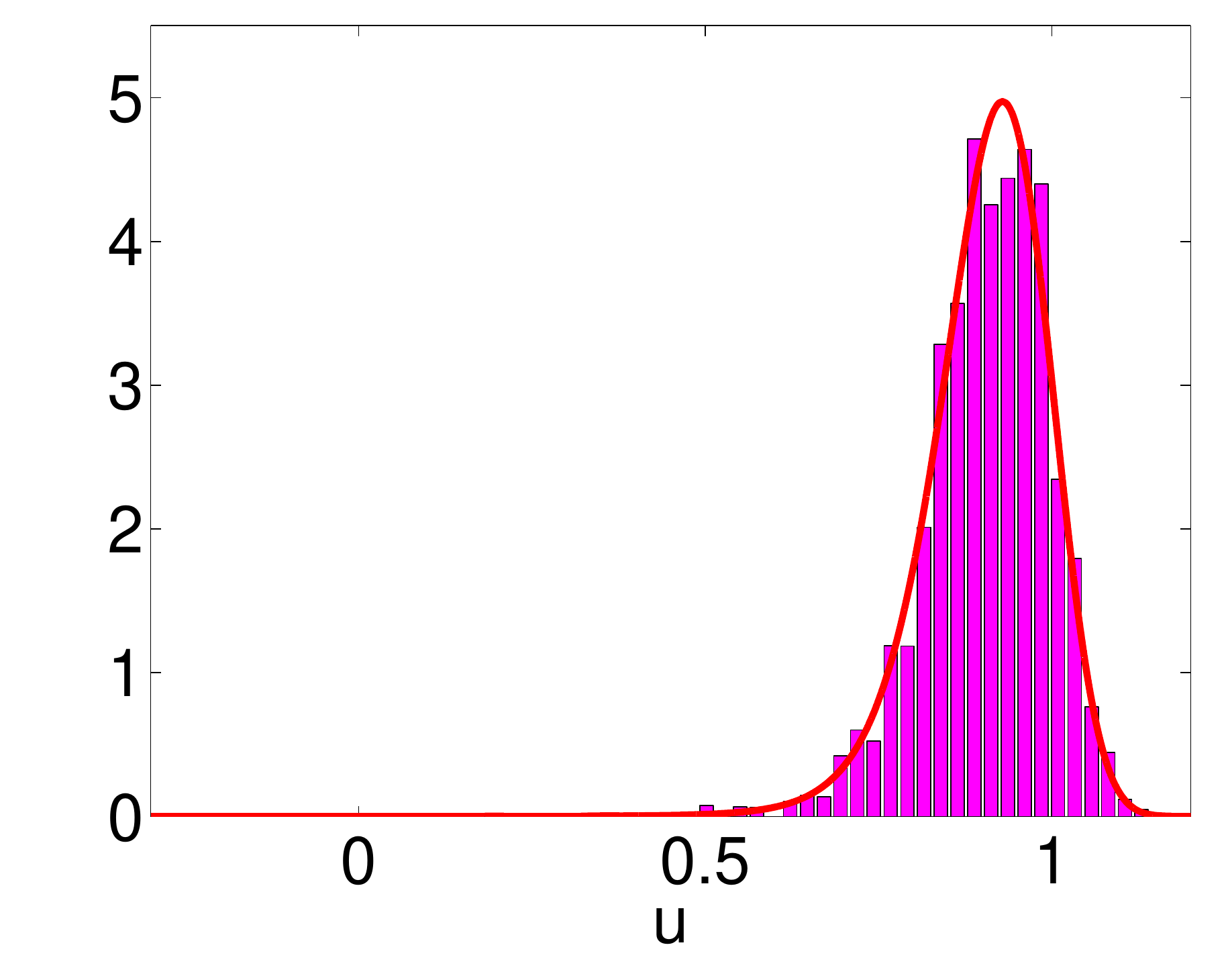}
\figlab{rMAP4corrected}
 }
\subfigure[RTO]{
\includegraphics[trim=1cm 6.0cm 2cm 7.1cm,clip=true, width=0.37\columnwidth]{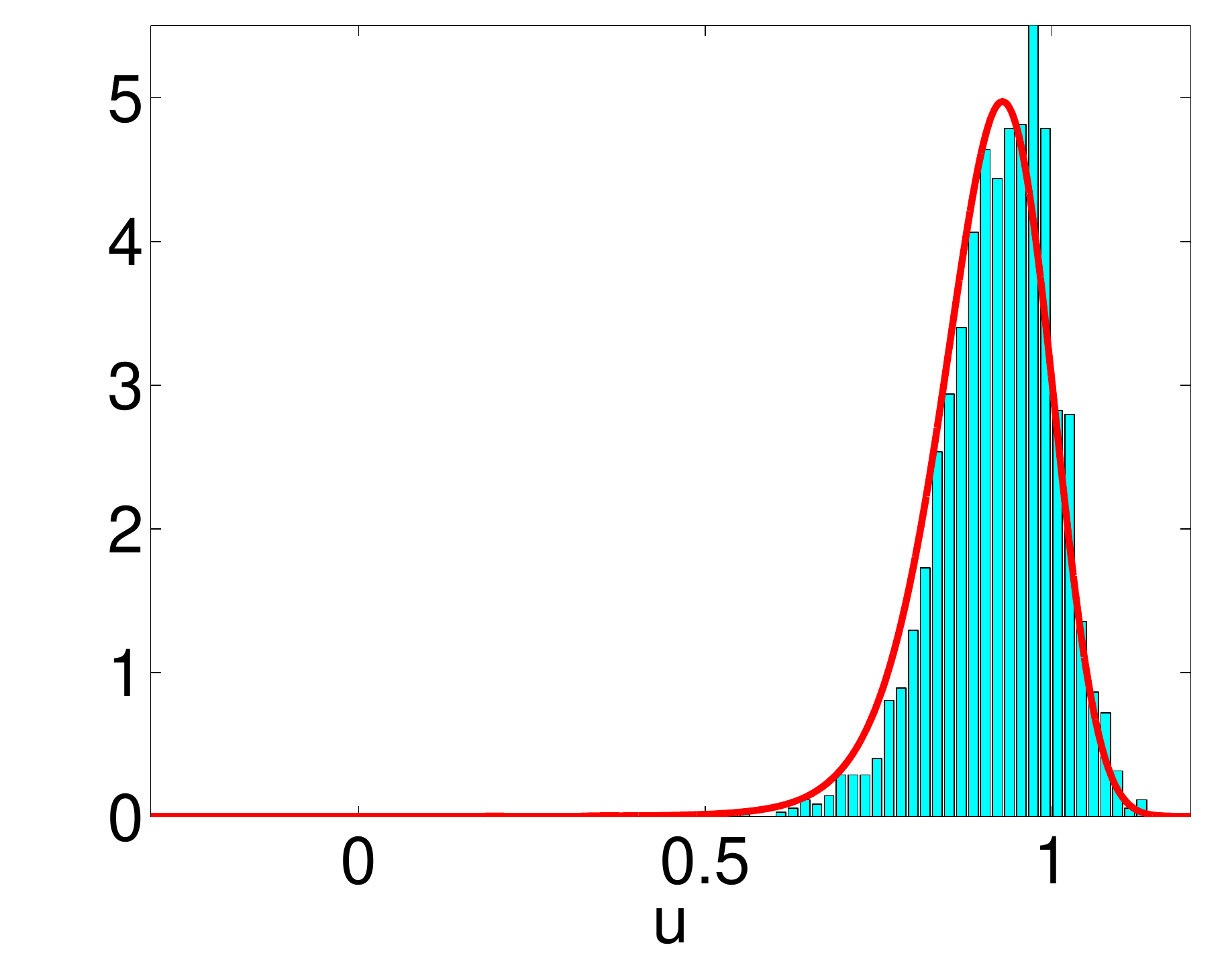}
\figlab{RTO4}
 }
\subfigure[Corrected RTO]{
\includegraphics[trim=1cm 6.0cm 2cm 7.1cm,clip=true, width=0.37\columnwidth]{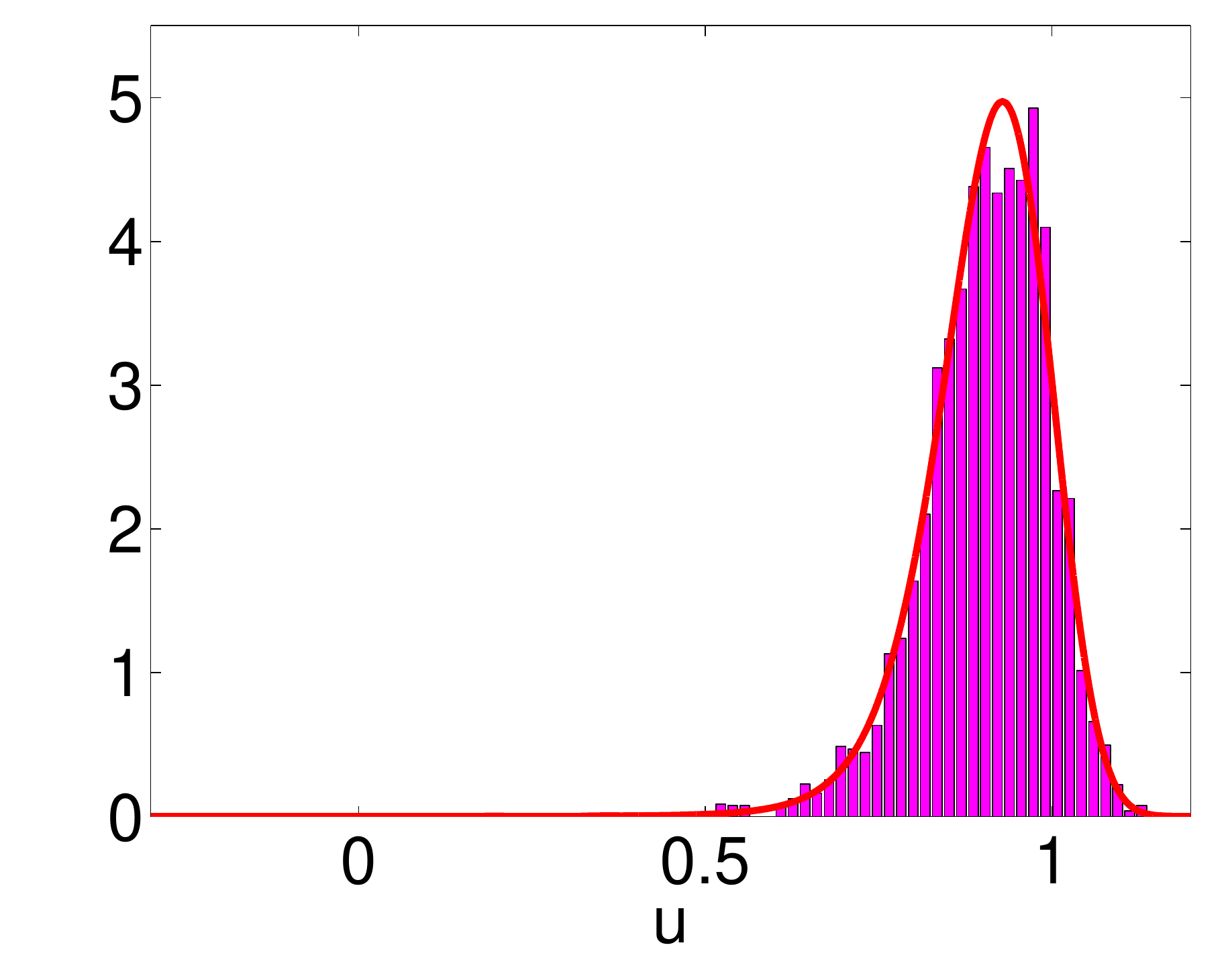}
\figlab{RTO4corrected}
 }
\end{center}
 \caption{Illustration of artificial mode created by rMAP and
   correction strategies for both rMAP and RTO. Correction for rMAP is
   by the square root of the Jacobian in Section \secref{metropolis},
   and the correction for RTO is via important sampling weights. The
   numerical experiments was done for the cost functional $\mc{J}_2$.}
 \figlab{rMAPandRTOhist2}
\end{figure}
\subsubsection{Comparing rMAP and Stochastic Newton methods}
In this section, we will numerically confirm our discussion in Section
\secref{rmapSNconnection} on the improvement of rMAP over the
stochastic Newton method.  For concreteness, we choose $\mc{J}_1$ in
\eqnref{costFunctional1}, a multi-modal function, for the
comparison. We have shown in Section \secref{rmapSNconnection} that
rMAP can be viewed as an iterative stochastic Newton method. It is
this deterministic iteration that can help rMAP explore the sample
space more rapidly.
In particular, rMAP can be
interpreted as a globalization strategy. It is in fact a move away
from the inefficiencies of random-walk/diffusion processes altogether,
toward powerful optimization methods that use derivative information
to traverse the posterior.

For numerical comparison, we compute $1,000$ samples from the
Metropolis-adjusted rMAP sampler and in this case the total number of Newton
iterations is approximately $20,000$. Since the parameter dimension is
one, the total number of (forward and adjoint) PDE solves is $40,000$.
For stochastic Newton method, we take $100,000$ samples. Three
independent chains with three different initial states, namely the origin, the left
and right modes of the posterior distribution, are computed for
both samplers. Figure \figref{rmapSNcomparison} shows the histogram of each chain together with the exact density. We observe that
rMAP chains are capable of sampling both modes and the sampling results are
independent of starting points. On the contrary, SN chains show dependency
on the starting points and they are stuck in local minima.

\begin{figure}[h!t!b!]
\begin{center}
\subfigure[rMAP]{
\includegraphics[width=0.4\columnwidth]{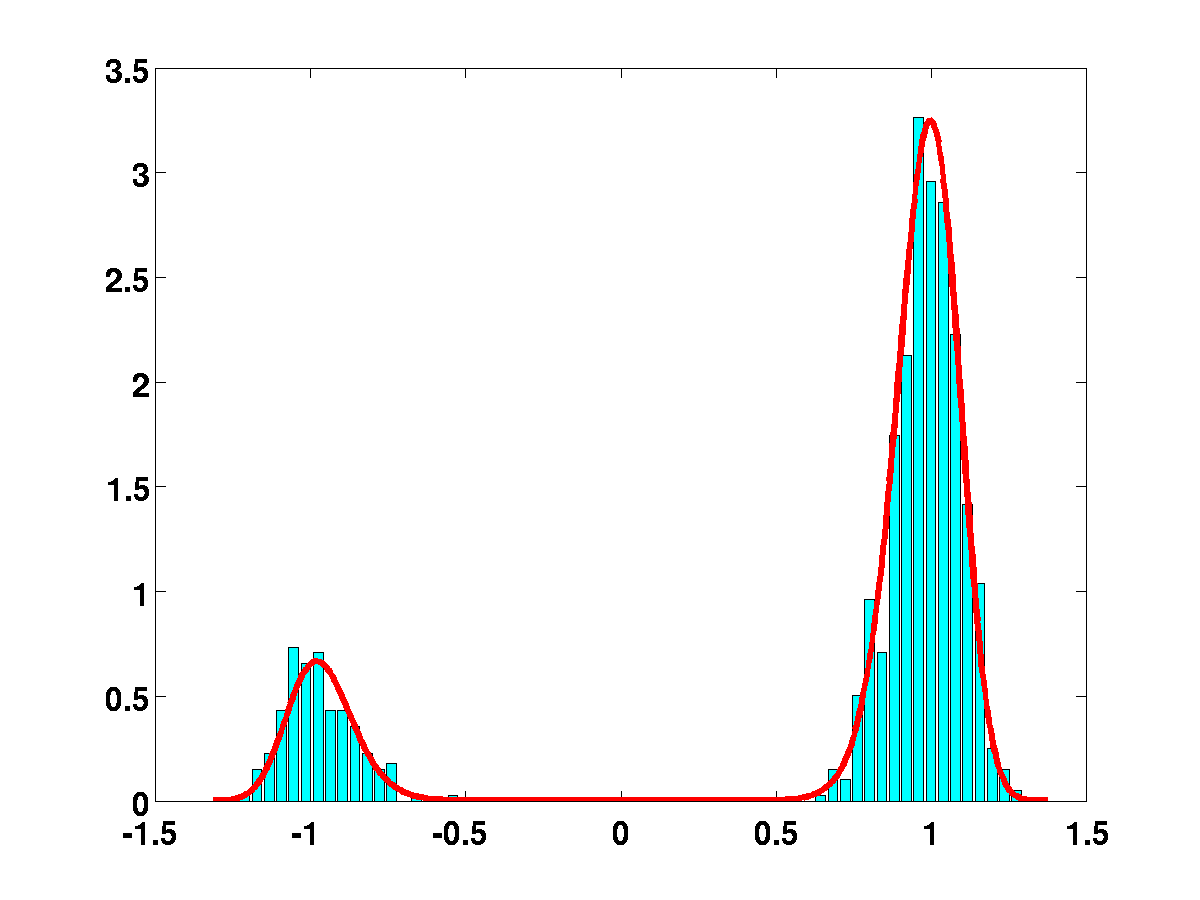}
}
\subfigure[SN (initial state at the left mode)]{
\includegraphics[width=0.4\columnwidth]{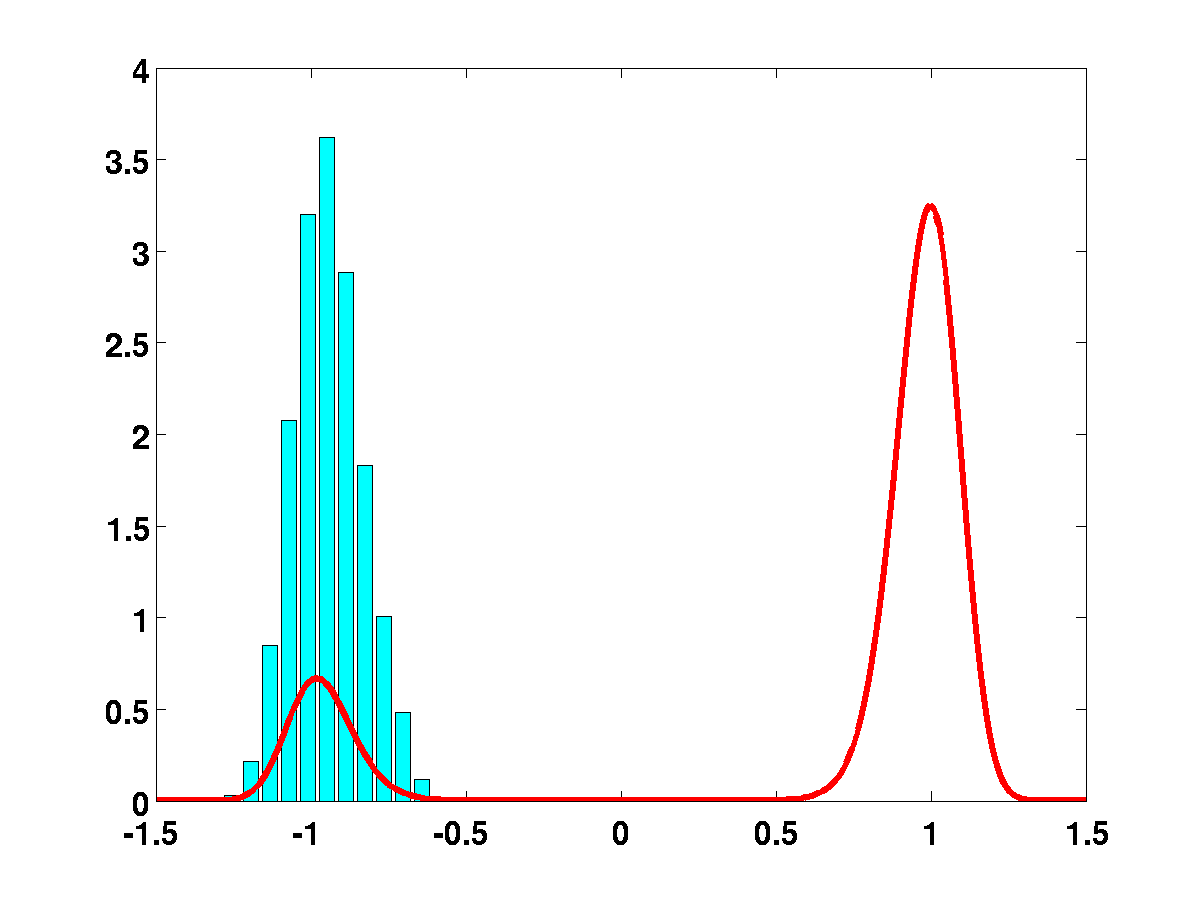}
}
\subfigure[SN (initial state at the origin)]{
\includegraphics[width=0.4\columnwidth]{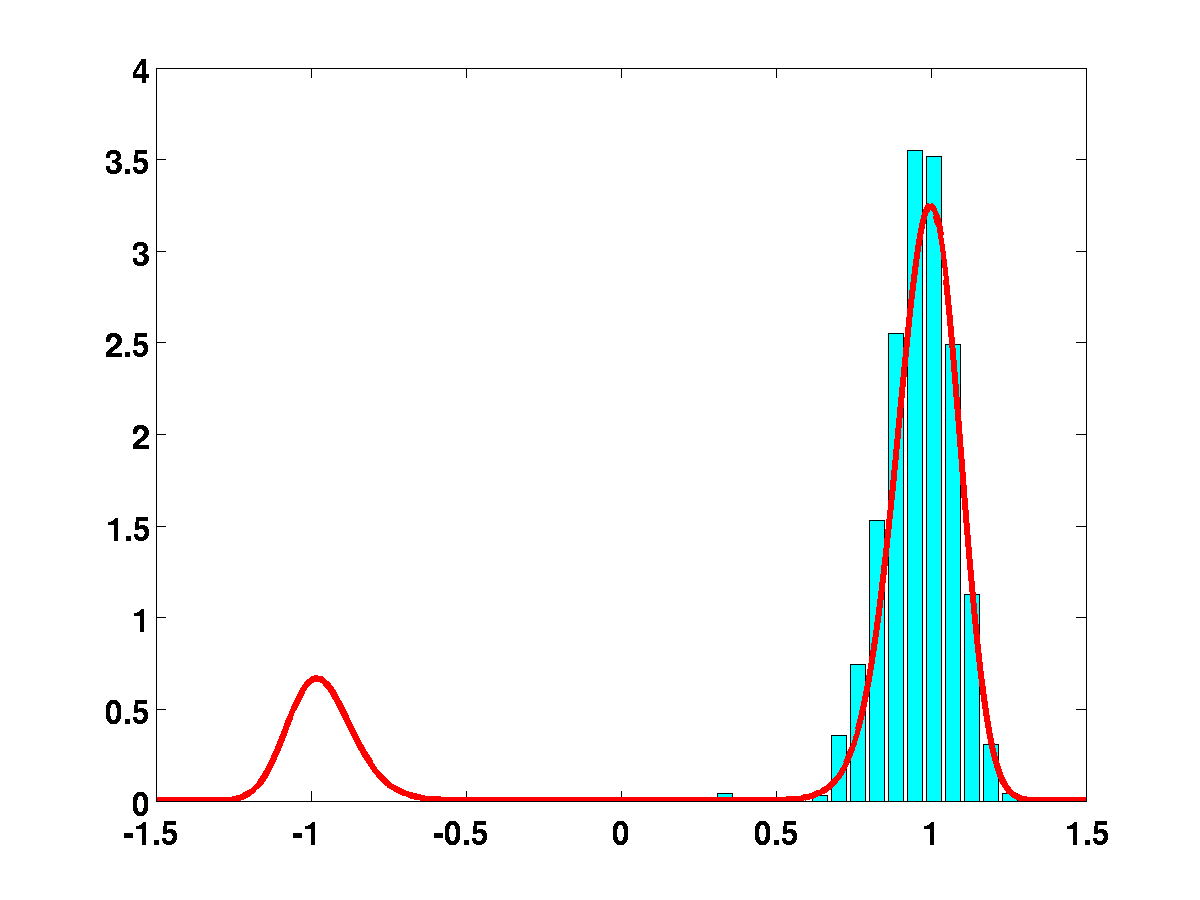}
}
\subfigure[SN (initial state at the right mode)]{
\includegraphics[width=0.4\columnwidth]{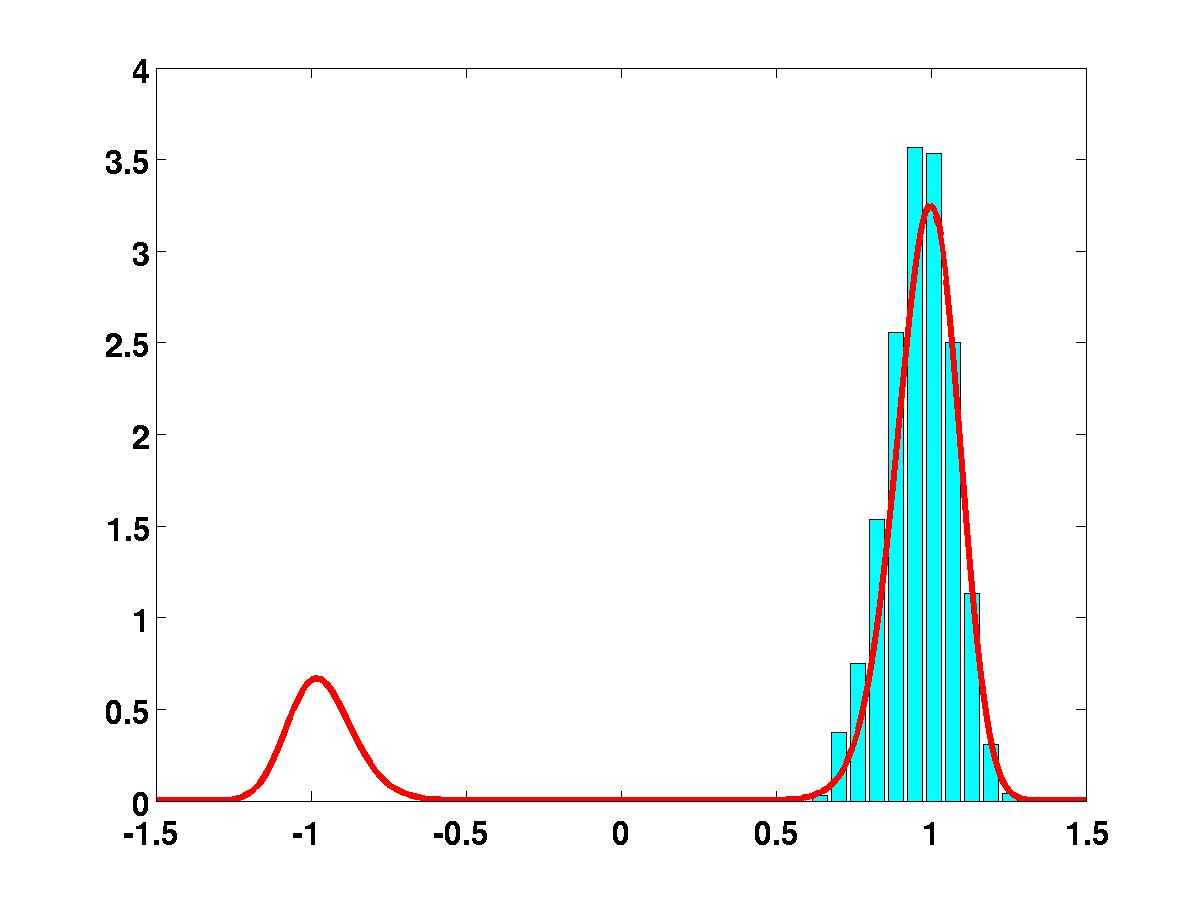}
}
\end{center}
\caption{Comparison of Metropolis-Adjusted rMAP and stochastic Newton
  (SN) MCMC methods for sampling multi-modal problem. Three starting
  points are chosen for these two samplers, namely the left mode,
  zero, and the right mode. The histograms are the same irrespective
  of the starting points for rMAP method, and hence only one plot is
  shown here. While SN chains are trapped in local minima, rMAP
  counterparts traverse the posterior very
  well.}\figlab{rmapSNcomparison}
\end{figure}

\subsubsection{Statistical Convergence of rMAP}
We also numerically examine Proposition \proporef{rMAPsamples} using
cost function $\mc{J}_1$. First, we compute the expectation
$\Ex_{{\sigb\times\epsb}}\LRs{S\LRp{ \ub_0, \db, \sigb, \epsb}}$ using
a tensor product Gauss-Hermite quadrature. Ten independent rMAP chains
are computed, each of which has one million samples. We compute the
averages $\{\frac{1}{n}\sum_{j=1}^n\ub_j\}_{n=1}^{N}, N=10^6$, over each chain and the resuts
are compared to the quadrature based expectation. In Figure
\figref{analyticalConvergence}, it is shown that the approximate mean
of rMAP samples aligns well with the limit
$\Ex_{{\sigb\times\epsb}}\LRs{S\LRp{ \ub_0, \db, \sigb, \epsb}}$, and
hence confirming our theoretical result in Proposition \proporef{rMAPsamples}.

\begin{figure}[h!t!b!]
\begin{center}
 \includegraphics[width=0.6\columnwidth, height=0.25\textheight]{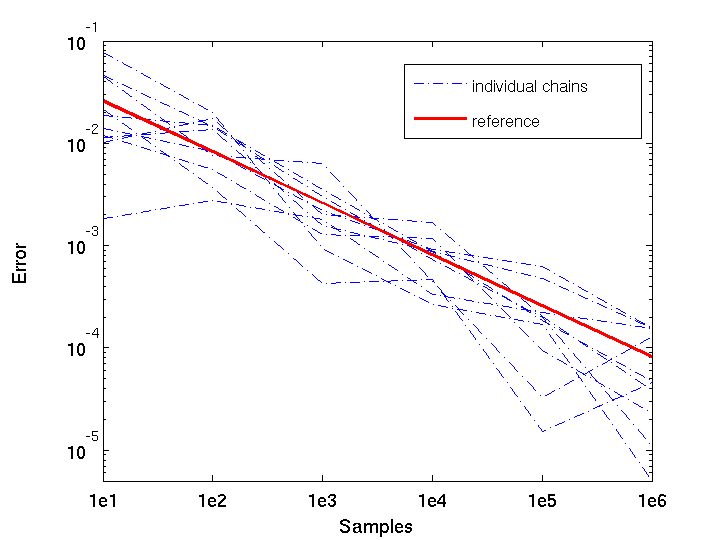}
 \caption{Convergence test of rMAP samples against a quadrature evaluated expectation value. Blue dashed lines show errors of each of ten rMAP chains, each containing one million samples. They align well with the the solid red line, which represents the theoretical, $n^{-\half}$ convergence rate from the central limit theorem.}
 \figlab{analyticalConvergence}
\end{center}
\end{figure}

\subsection{Helmholtz Problems}
\seclab{numericHelmholtz}
Although our proposed framework is valid for Bayesian inverse problems
governed by any system of forward PDEs, here we illustrate the use of the framework
on a frequency domain acoustic wave equation in the form of the
Helmholtz equation. Namely, the forward model $\mc{B}\LRp{\u, \w}$ is defined, in an open and bounded domain $\Omega$, as:
\begin{align*}
-\Grad^2\w-e^{2\u} \w &= 0, \text{ in } \Omega,\\
                \pp{\w}{\mb{n}} &= \g, \text{ on } \pOmega,
\end{align*}
where $\w$ is the acoustic field, $\u$ the logarithm of the
distributed wave number field on $\Omega$, $\mb{n}$ the unit outward
normal on $\pOmega$, and $g$ the prescribed Neumann source on the
boundary.

In the following subsection \secref{adjointHelmholtz}, we first
discuss the computation of the gradient and Hessian of the objective
function using the adjoint method. The adjoint method enables
tractable computation of the MAP estimator, which is crucial to the
rMAP algorithm. In subsection \secref{samplingHelmholtz}, we analyze
the sampling results using rMAP algorithm. Through a comparison
between different optimization settings described above, we
demonstrate the efficiency achieved by using the TRINCG solver and a
good initial guess. In addition, the rMAP samples are compared with
delayed rejection adaptive Metropolis (DRAM) samples, where we observe
that Metropolis-adjusted rMAP samples provide statistical estimates
with similar quality compared to those obtained from DRAM, while
requiring much less computation.

\subsubsection{Computation of the gradient and Hessian-vector product}
\seclab{adjointHelmholtz}
In this section, we briefly discuss about how to compute the gradient
and Hessian-vector product efficiently. Using the standard reduced
space approach, see e.g. \cite{Bui-ThanhGhattas12}, one can show that
the (reduced) gradient $\Grad\mc{J}\LRp{\u; \hat{\u}, \hat{\db}}$ acting in any
direction $\tilde{\u}$ is given by
\[
\LRa{\Grad\mc{J}\LRp{\u; \hat{\u}, \hat{\db}},\tilde{\u}} = 
-2\iOm{\tilde{\u}e^{2\u}\w\tau}
\]
where the adjoint state $\tau$ satisfies the adjoint equation
\begin{subequations}
\eqnlab{Adjointw}
\begin{align}
-\Grad^2\tau-e^{2\u}\tau &= -\frac{1}{\sigma^2}\sum_{j=1}^K\LRp{\w - 
\d_j}\delta\LRp{\mb{x}-\mb{x}_j}  \text{ in } \Omega, \eqnlab{adjoint}\\
                \pp{\tau}{\mb{n}} &= 0, \text{ on } \pOmega. \eqnlab{adjointbc}
\end{align}
\end{subequations}

On the other hand, the Hessian acting in directions $\tilde{\u}$ and
$\stackrel{\simeq}{\u}$ reads
\[
\LRa{\LRa{\Grad^2\mc{J}\LRp{\u; \hat{\u},
      \hat{\db}},\tilde{\u}},\stackrel{\simeq}{\u}} =
-4\iOm{\tilde{\u}\stackrel{\simeq}{\u}e^{2\u}\w\tau} 
-2\iOm{\tilde{\u}e^{2\u}\tilde{\w}\tau}-2\iOm{\tilde{\u}e^{2\u}\w\tilde{\tau}},
\]
where the incremental forward state $\tilde{\w}$ obeys the incremental
forward equation
\begin{subequations}
\eqnlab{IforwardT}
\begin{align}
-\Grad^2\tilde{\w}-e^{2\u}\tilde{\w} &= 2\stackrel{\simeq}{\u}e^{2\u}\w \text{ 
in } \Omega, \eqnlab{Iforward}\\
                \pp{\tilde{\w}}{\mb{n}} &= 0, \text{ on } 
\pOmega,\eqnlab{Iforwardbc}
\end{align}
\end{subequations}
and the incremental adjoint state $\tilde{\tau}$ obeys the incremental
adjoint equation
\begin{subequations}
\begin{align}
-\Grad^2\tilde{\tau}-e^{2\u}\tilde{\tau} &= 2\stackrel{\simeq}{
  \u}e^{2\u}\tau 
-\frac{1}{\sigma^2}\sum_{j=1}^K\tilde{\w}\delta\LRp{\mb{x}-\mb{x}_j} \text{ in 
} \Omega,
\eqnlab{Iadjoint}\\ \pp{\tilde{\tau}}{\mb{n}} &= 0, \text{ on }
\pOmega. \eqnlab{Iadjointbc}
\end{align}
\end{subequations}

We shall compare our TRNCG optimization solver with the popular
Levenberg-Marquardt approach (see, e.g.,
\cite{NocedalWright06, OliverReynoldsLiu08}). For that reason, we need
to compute the Gauss-Newton Hessian-vector product. It can be shown
that the Gauss-Newton Hessian acting in directions $\tilde{\u}$ and
$\stackrel{\simeq}{\u}$ reads
\[
\LRa{\LRa{\Grad^2\mc{J}_{GN}\LRp{\u; \hat{\u},
      \hat{\db}},\tilde{\u}},\stackrel{\simeq}{\u}} =
-2\iOm{\tilde{\u}e^{2\u}\w\tilde{\tau}},
\]
where the incremental forward state $\tilde{\w}$ still satisfies
\eqnref{IforwardT}, but the incremental adjoint state $\tilde{\tau}$
now obeys the following incremental adjoint equation
\begin{subequations}
\begin{align}
-\Grad^2\tilde{\tau}-e^{2\u}\tilde{\tau} &= 
-\frac{1}{\sigma^2}\sum_{j=1}^K\tilde{\w}\delta\LRp{\mb{x}-\mb{x}_j} \text{ in 
} \Omega,
\eqnlab{IadjointGN}\\ \pp{\tilde{\tau}}{\mb{n}} &= 0, \text{ on }
\pOmega. \eqnlab{IadjointbcGN}
\end{align}
\end{subequations}

\subsubsection{Sampling results}
\seclab{samplingHelmholtz}
Now we show the application of rMAP methods
to quantify the uncertainty for the inverse problem governed by the above Helmholtz
forward model. We create two experiments to  compare and test the
methods described above. Finite element discretization of the
prior results in a parameter field with $94$ parameters for both
experiments. Since the experiments aim at testing algorithms rather
than demonstrating Bayesian modeling, we conveniently fix the noise
level for both experiments to be $1\%$. On the other hand, we use different prior
coefficients $\alpha$ to control the `easiness' of sampling. We choose
$\alpha = 8.0$ for the first experiment and $\alpha = 3.0$ for the
second experiment---these numbers are chosen after trials to clearly
represent two situations: a prior dominant case and a likelihood
dominant case. When the model is prior dominant, the inverse problem
resembles a linear inverse problem for which, following Lemma
\lemref{linearRML}, rMAP should provide exact posterior samples. On
the other hand, for the likelihood dominant case, due to the
non-linearity of the forward model, the rMAP samples are no longer
exact posterior ones and Metropolization becomes necessary.

For each of these two experiments, we draw a sample from the prior distribution and solve the forward
equation \eqnref{Forwardw} to generate a set of synthetic data as shown in Figure \figref{refall}. Then, we sample the Bayesian model with four variants of the rMAP method: with trust-region inexact Newton-CG (TRINCG) or Levenberg-Marquardt (LM)
and with/without good initial guesses. As a comparison, we also sample the model with DRAM sampler of five million samples which we consider long enough to be convergent.

\begin{figure}[h!t!b!]
\begin{center}
\subfigure[Synthetic $\u$ $\LRp{\alpha=8.0}$]{
\includegraphics[trim=1cm 6.0cm 2cm 7.1cm,clip=true,width=0.35\columnwidth]{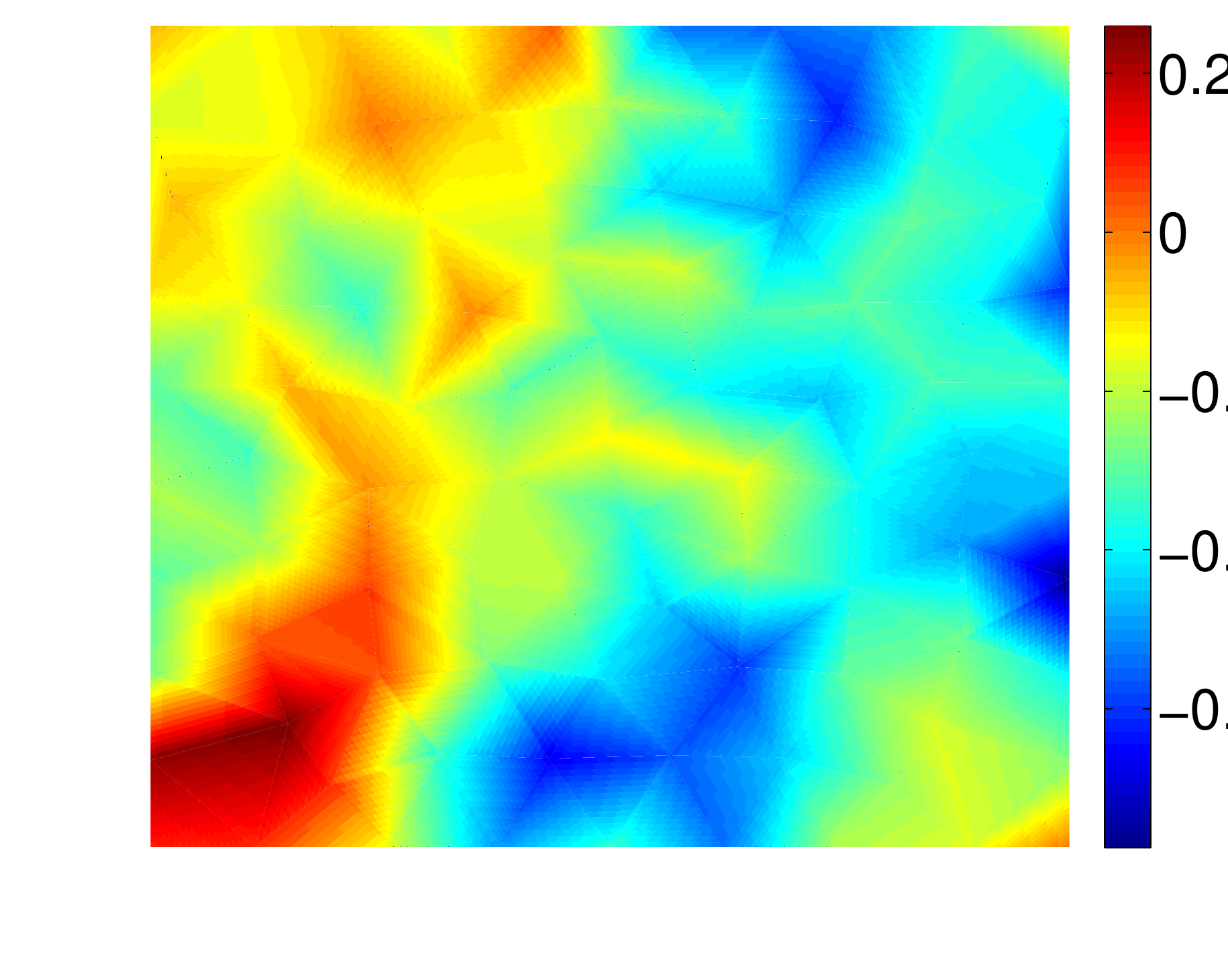}
 }\figlab{highref}
\subfigure[Synthetic $\u$ $\LRp{\alpha=3.0}$]{
\includegraphics[trim=1cm 6.0cm 2cm 7.1cm,clip=true, width=0.35\columnwidth]{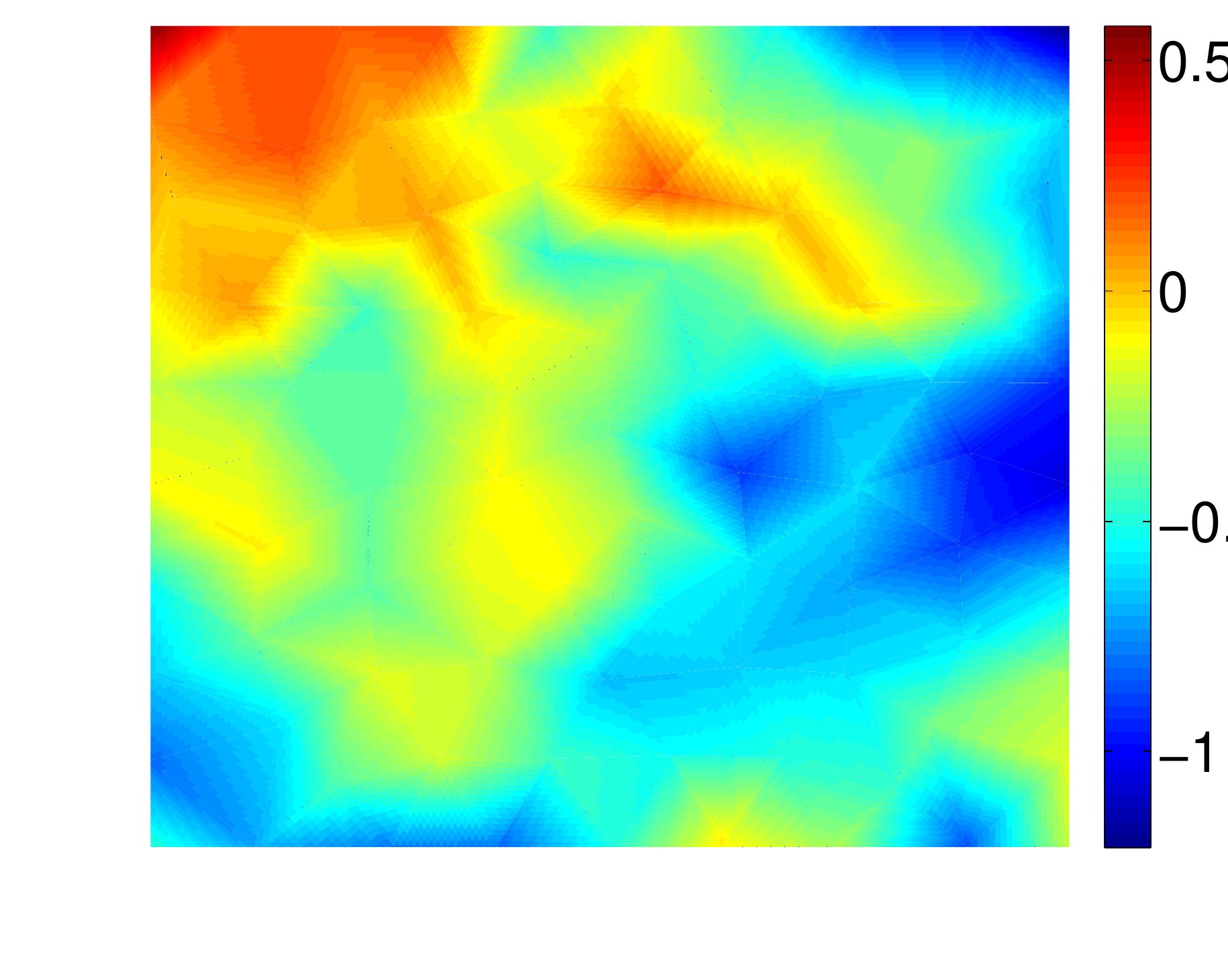}
 }\figlab{lowref}
 \caption{Synthetic parameter $\u$ for two numerical experiments.}
 \figlab{refall}
\end{center}
\end{figure}

Since rMAP samples are not exact posterior samples for nonlinear problems, it is not necessary to demand high accuracy (and hence high cost) in each optimization solution. Yet, we still hope that for these loosely approximate rMAP samples, the proposed Metropolization can effectively correct them towards the posterior distribution. To that end, we set large tolerances: $\varepsilon_F = \varepsilon_X =
\varepsilon_G = 10^{-4}$ for the first experiment and $\varepsilon_F = \varepsilon_X =
\varepsilon_G = 10^{-6}$ for the second one. For a similar reason, we limit the maximal allowed number of iterations to $150$ and $200$ respectively to further control the computational costs.

For each experiment and each variant of the method, we compute 1000 rMAP samples. Within each experiment, we use the same randomly perturbed sequences $\{\hat{\u}\}_{i=1}^{1000}$ and $\{\hat{\db}\}_{i=1}^{1000}$ for all four rMAP methods. Ideally, with this setting,
these methods should produce exactly same rMAP samples if each
optimizer had converged. In practice, the acquired samples are different among these methods due to the tolerance and iteration control. 

Figures \figref{highallcompare} and \figref{lowallcompare} show the estimated conditonal mean and variance for the high prior and the low prior cases respectively. In both cases, the plain rMAP samples have non-negligible approximation errors. These errors are successfully corrected with a Metropolization using weights described in Section \secref{metropolis}. We point out that among the four variants of rMAP methods, the one that uses TRINCG and good initial guesses has shown optimal performance. Its statistical estimates are close to that of the DRAM sampler in both experiments. This indicates the fast convergence of the TRINCG method that despite we have relaxed the convergence criteria and limited the number of iterations, TRINCG has always been able to get close to the real optimizer rapidly. Our proposed method of computing initial guesses has further ensured its efficiency. To give a closer look, we show a comparison between sampling estimates of the DRAM and the rMAP using TRINCG with good initial guesses in Figures \figref{highcomparedramrmap} and \figref{lowcomparedramrmap}.

\begin{figure}[h!t!b!]
\begin{center}
\subfigure[mean (rMAP)]{
\includegraphics[trim=1cm 6.0cm 2cm 7.1cm,clip=true,width=0.45\columnwidth, height=0.2\textheight]{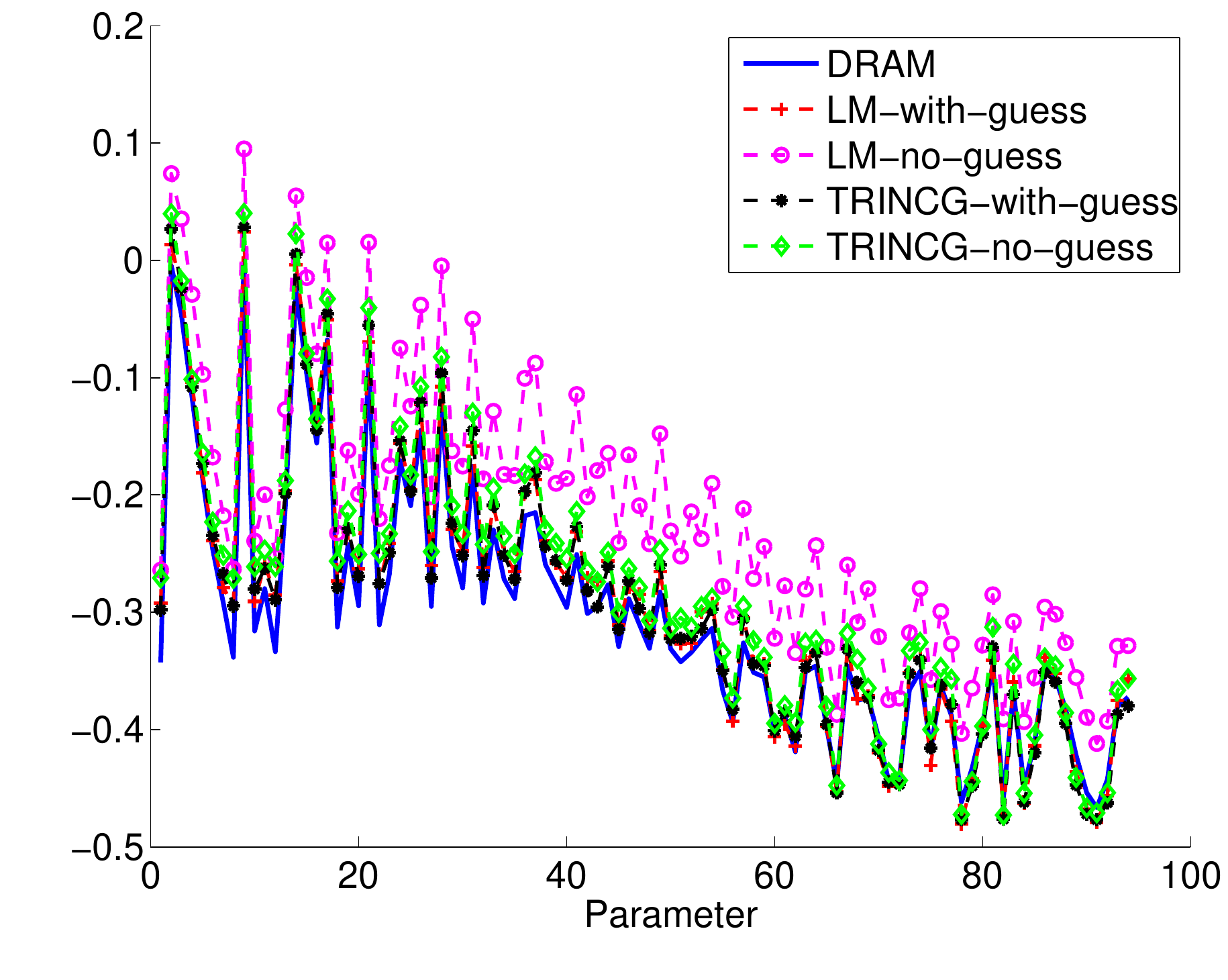}
}
\subfigure[mean (weighted-rMAP)]{
\includegraphics[trim=1cm 6.0cm 2cm 7.1cm,clip=true,width=0.45\columnwidth, height=0.2\textheight]{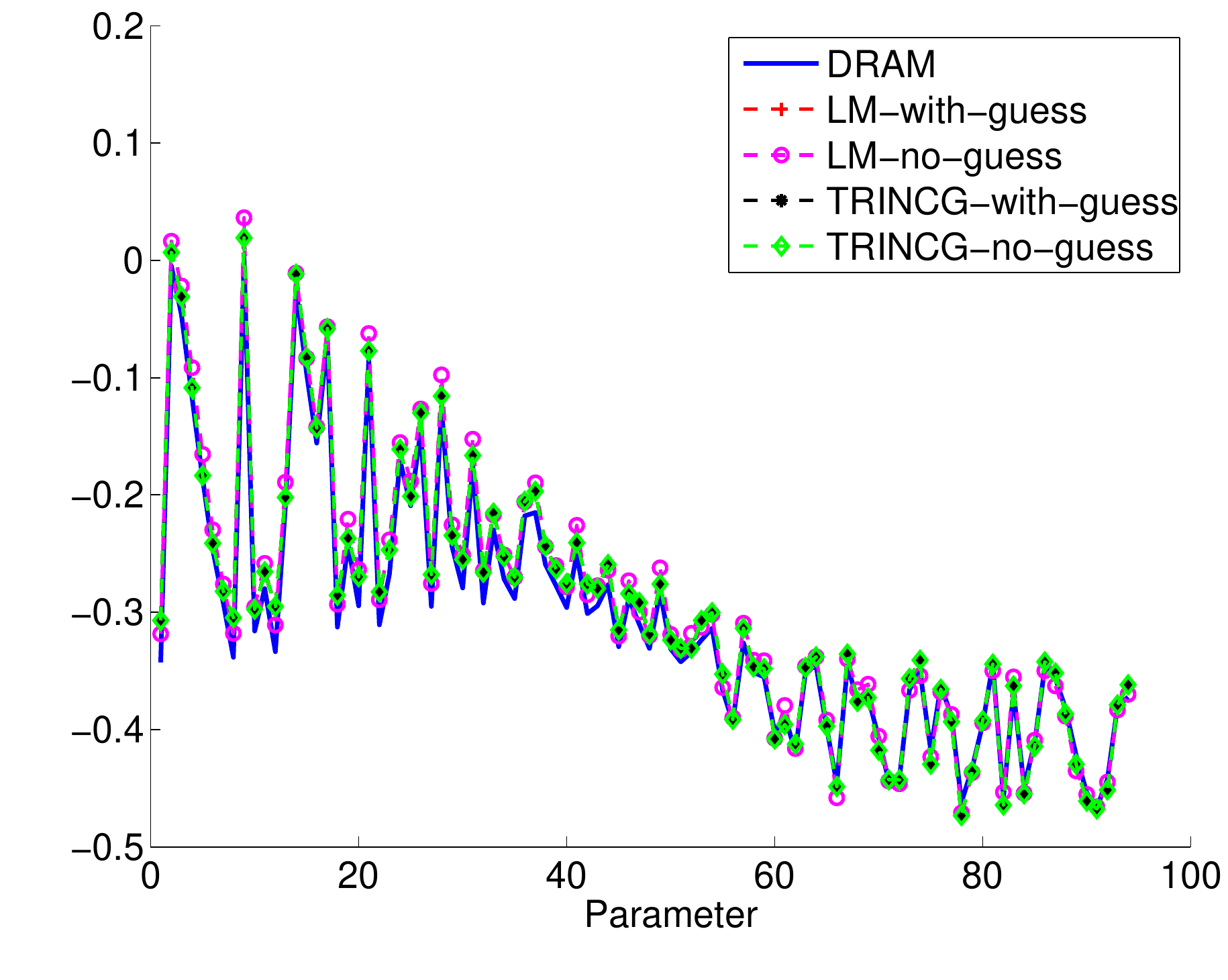}
}
\subfigure[Variance (rMAP)]{
\includegraphics[trim=1cm 6.0cm 2cm 7.1cm,clip=true,width=0.45\columnwidth,height=0.2\textheight]{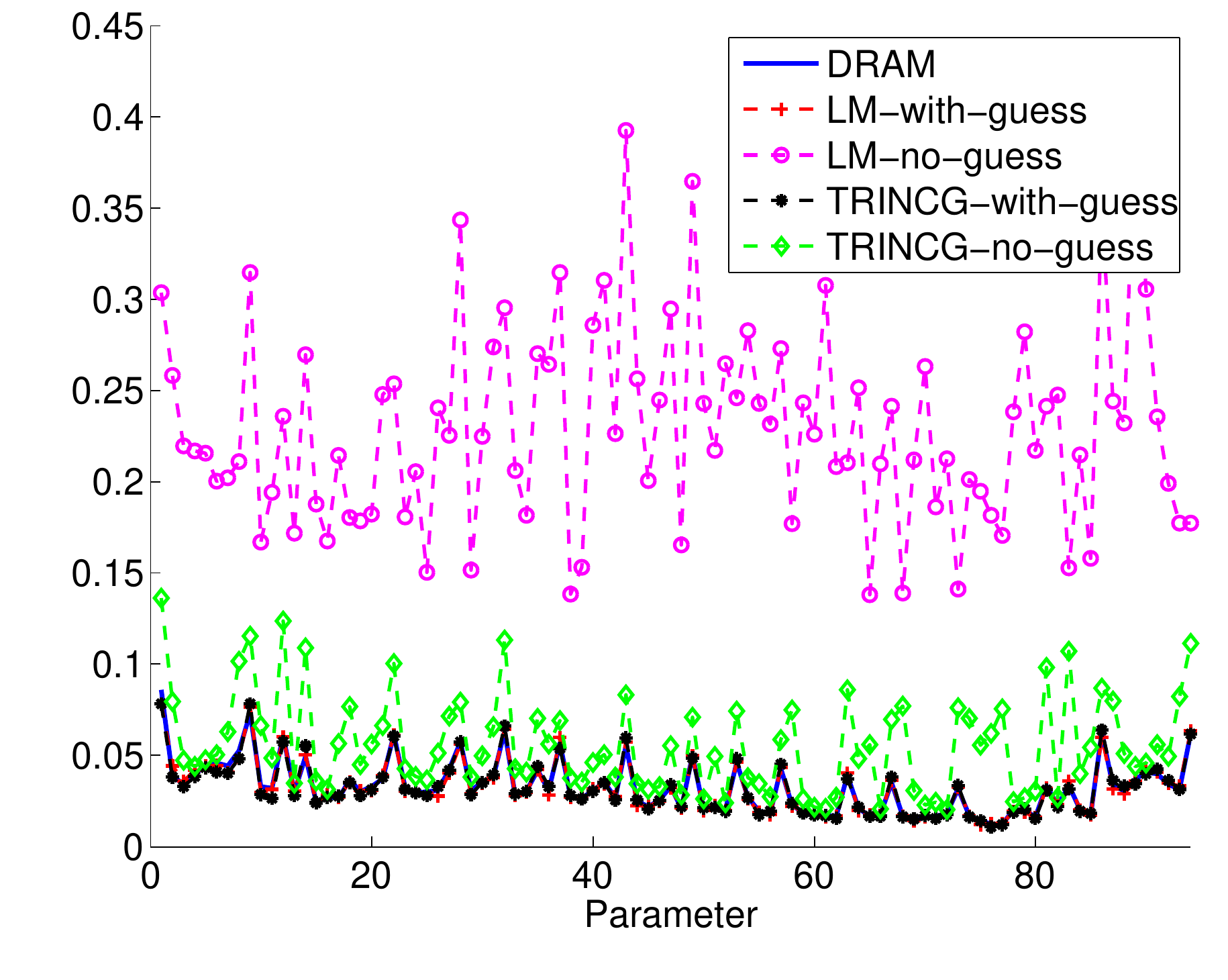}
}
\subfigure[Variance (weighted-rMAP)]{
\includegraphics[trim=1cm 6.0cm 2cm 7.1cm,clip=true,width=0.45\columnwidth,height=0.2\textheight]{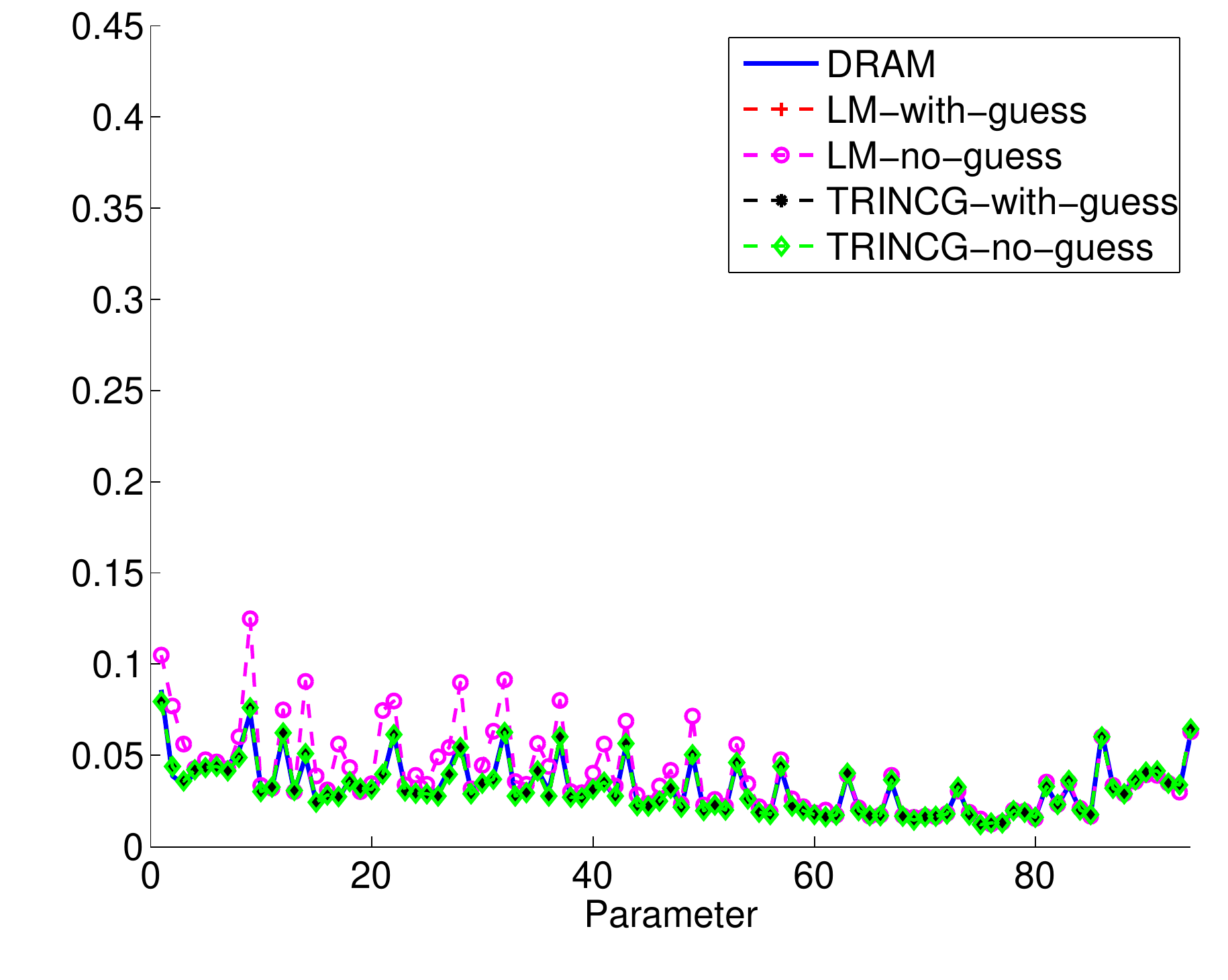}
}
\caption{ Case $\alpha = 8.0$: Comparison of estimated statistics from all samplers. Top row shows the conditonal mean estimate between (a) DRAM and rMAP samples; and (b) DRAM and Metropolized rMAP samples. Bottom row shows corresponding comparison of variance estimates.}
\figlab{highallcompare}
\end{center}
\end{figure}

\begin{figure}[h!t!b!]
\begin{center}
\subfigure[mean(rMAP)]{
\includegraphics[trim=1cm 6.0cm 2cm 7.1cm,clip=true,width=0.45\columnwidth,height=0.2\textheight]{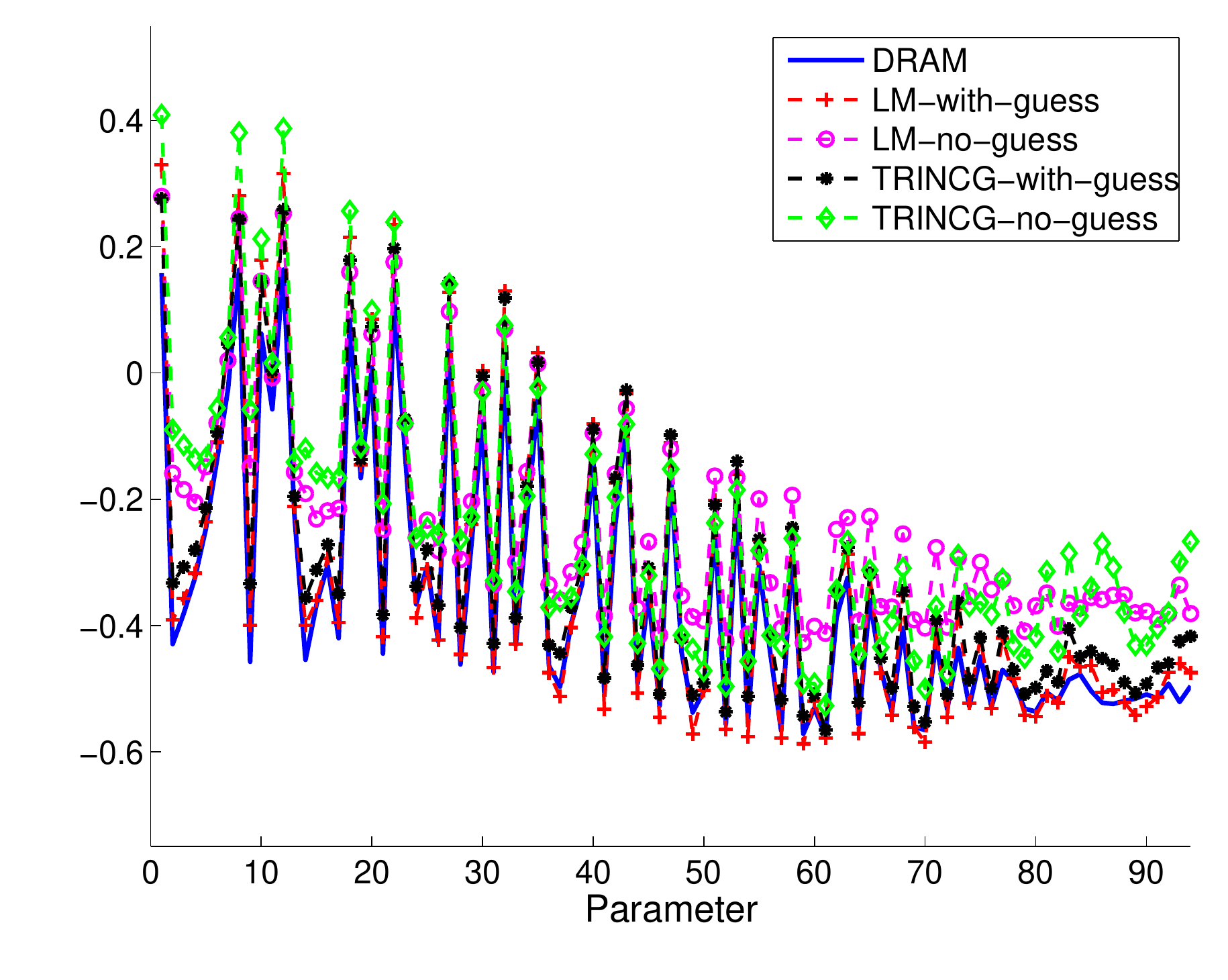}
}
\subfigure[mean(weighted-rMAP)]{
\includegraphics[trim=1cm 6.0cm 2cm 7.1cm,clip=true,width=0.45\columnwidth,height=0.2\textheight]{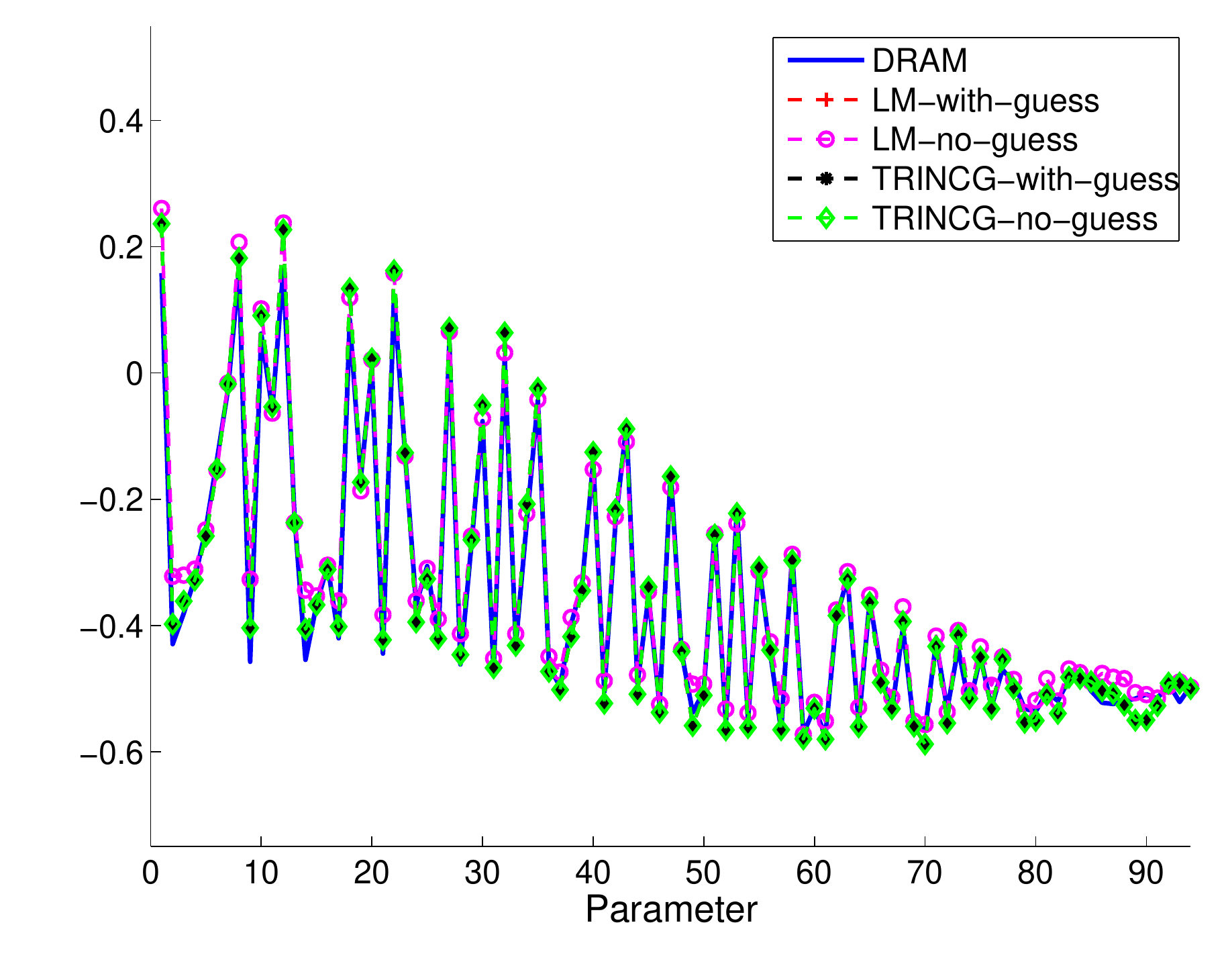}
}
\subfigure[Variance(rMAP)]{
\includegraphics[trim=1cm 6.0cm 2cm 7.1cm,clip=true,width=0.45\columnwidth,height=0.2\textheight]{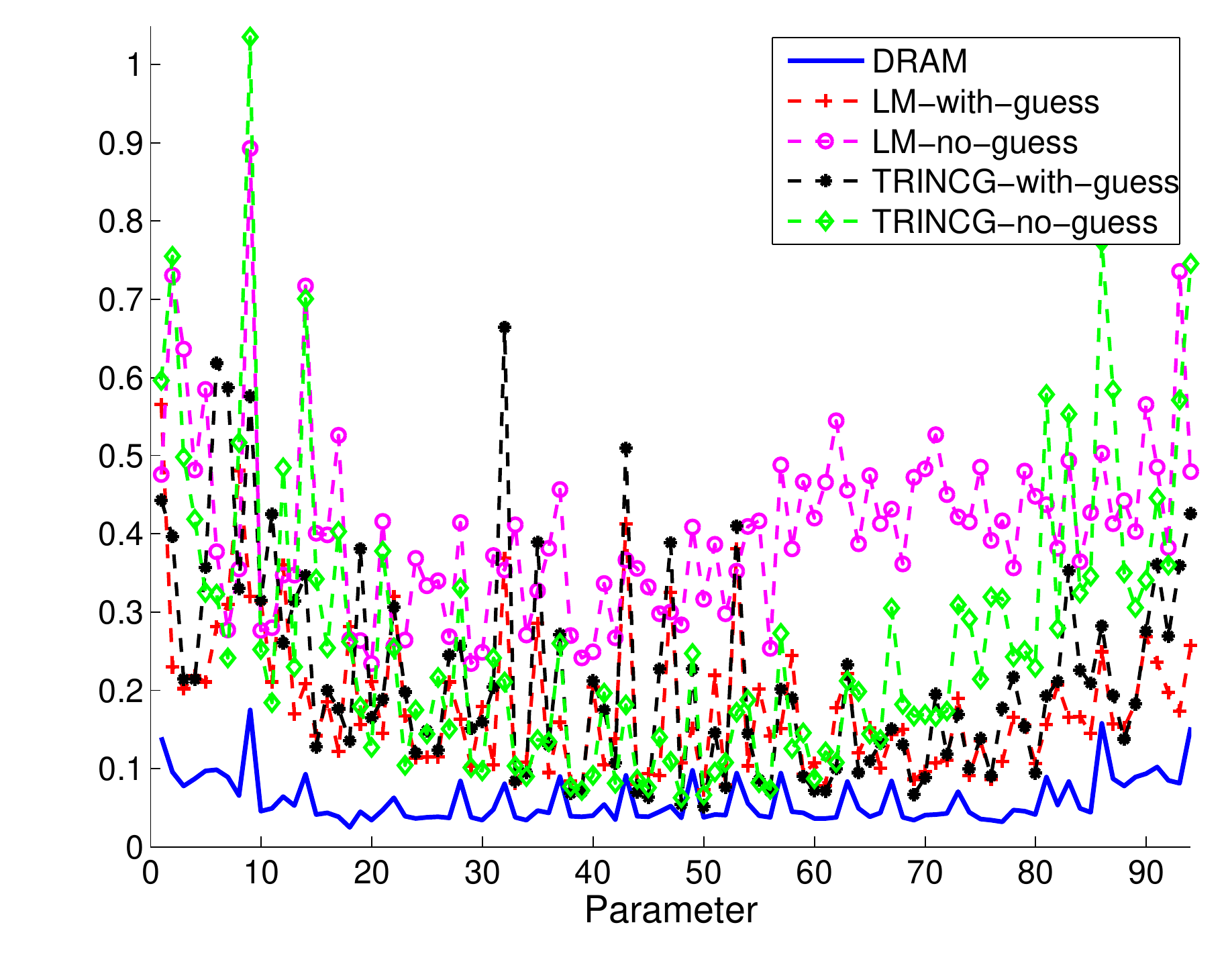}
}
\subfigure[Variance(weighted-rMAP)]{
\includegraphics[trim=1cm 6.0cm 2cm 7.1cm,clip=true,width=0.45\columnwidth,height=0.2\textheight]{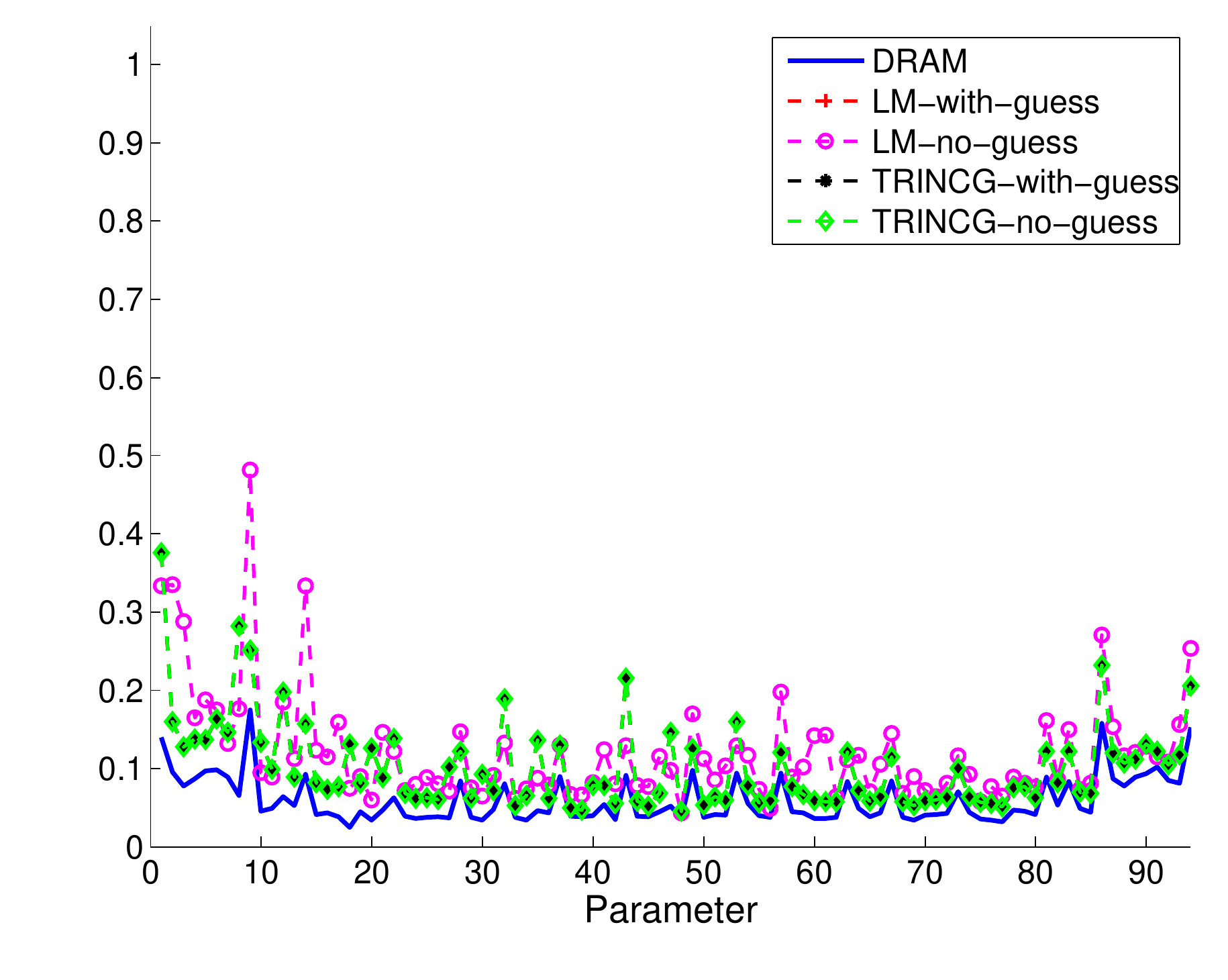}
}
\caption{($\alpha = 3.0$) Comparison of estimated statistics from all samplers. Top row shows the conditonal mean estimate between (a)DRAM and rMAP samples; and (b)DRAM and Metropolized rMAP samples. Bottom row shows corresponding comparison of variance estimates.}
\figlab{lowallcompare}
\end{center}
\end{figure}

\begin{figure}[h!t!b!]
\begin{center}
\subfigure[DRAM]{
\includegraphics[trim=1cm 6.0cm 2cm 7.1cm,clip=true,width=0.3\columnwidth]{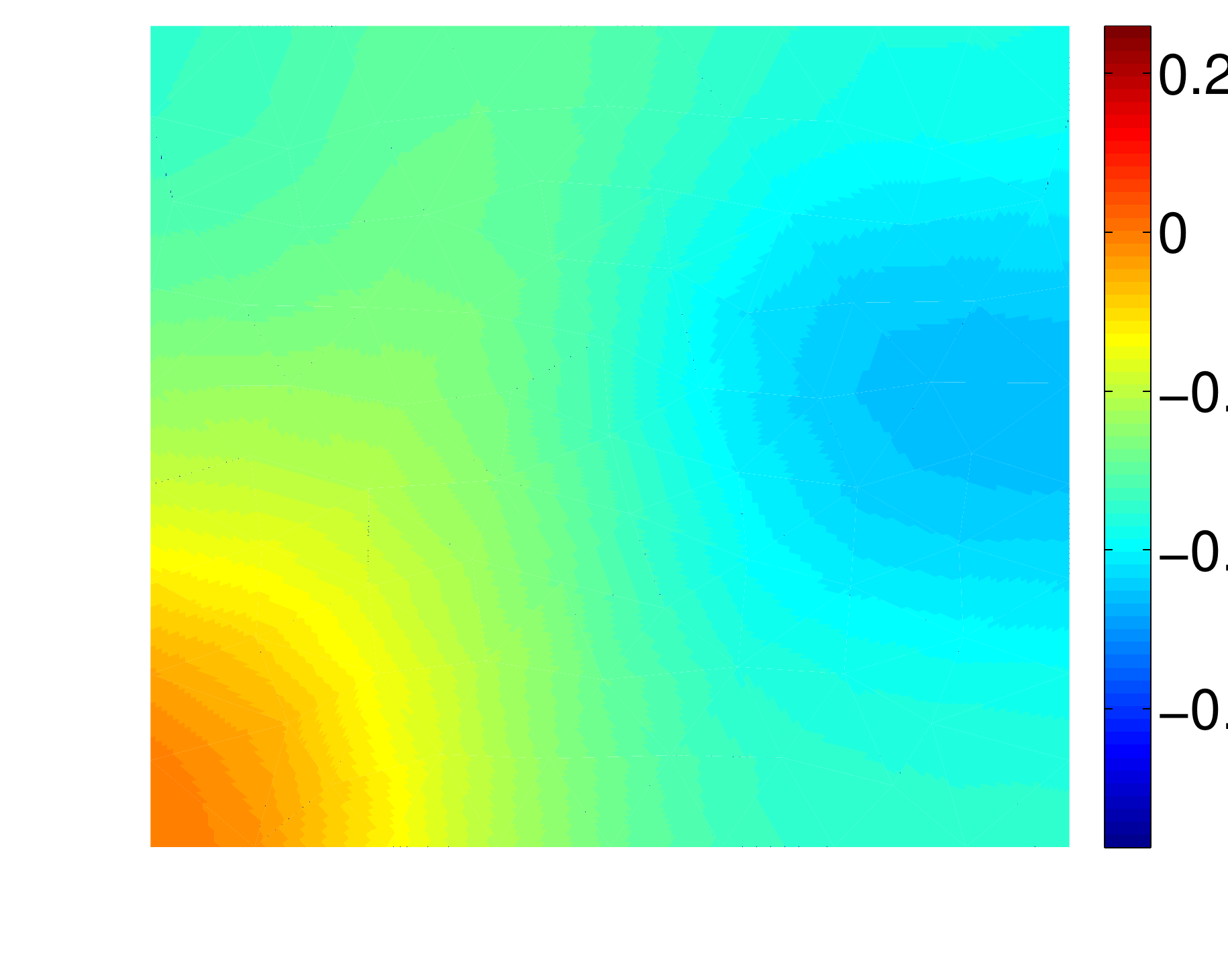}
 }\figlab{highdrammean}
 \subfigure[rMAP]{
\includegraphics[trim=1cm 6.0cm 2cm 7.1cm,clip=true,width=0.3\columnwidth]{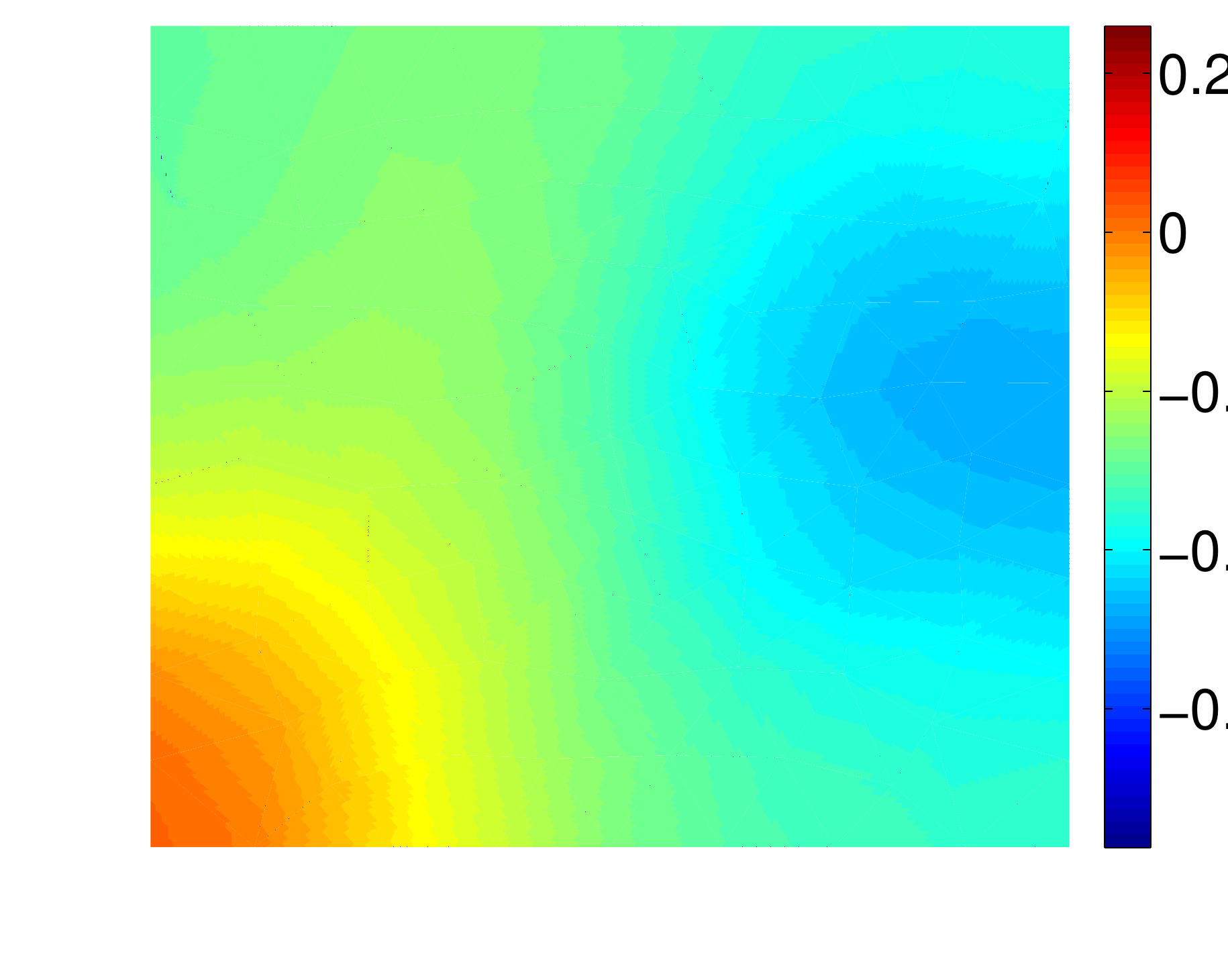}
 }\figlab{highrmlmean}
 \subfigure[Weighted-rMAP]{
\includegraphics[trim=1cm 6.0cm 2cm 7.1cm,clip=true,width=0.3\columnwidth]{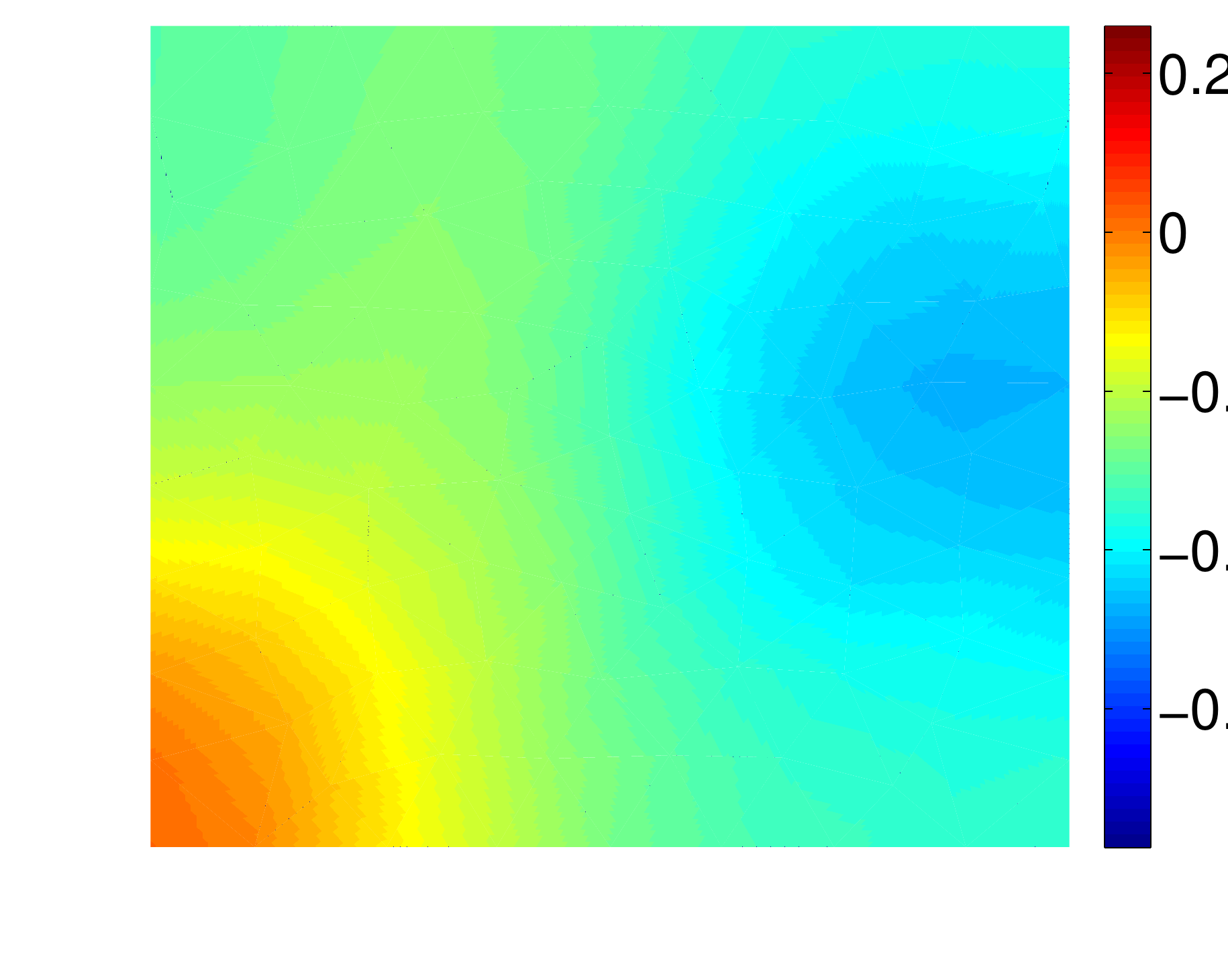}
 }\figlab{highrmlweightmean}
\subfigure[DRAM]{
\includegraphics[trim=1cm 6.0cm 2cm 7.1cm,clip=true,width=0.3\columnwidth]{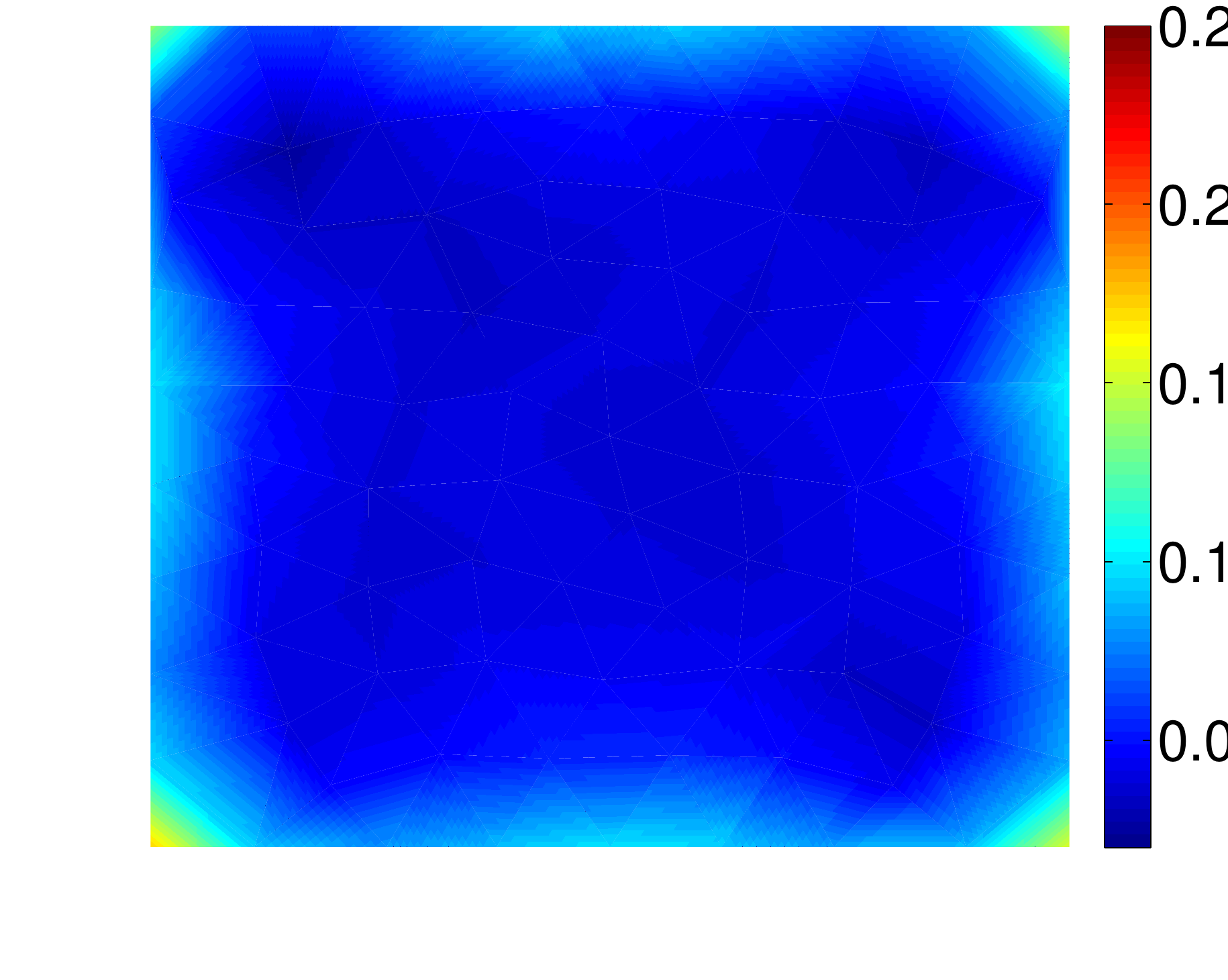}
 }\figlab{highdramvar}
 \subfigure[rMAP]{
\includegraphics[trim=1cm 6.0cm 2cm 7.1cm,clip=true,width=0.3\columnwidth]{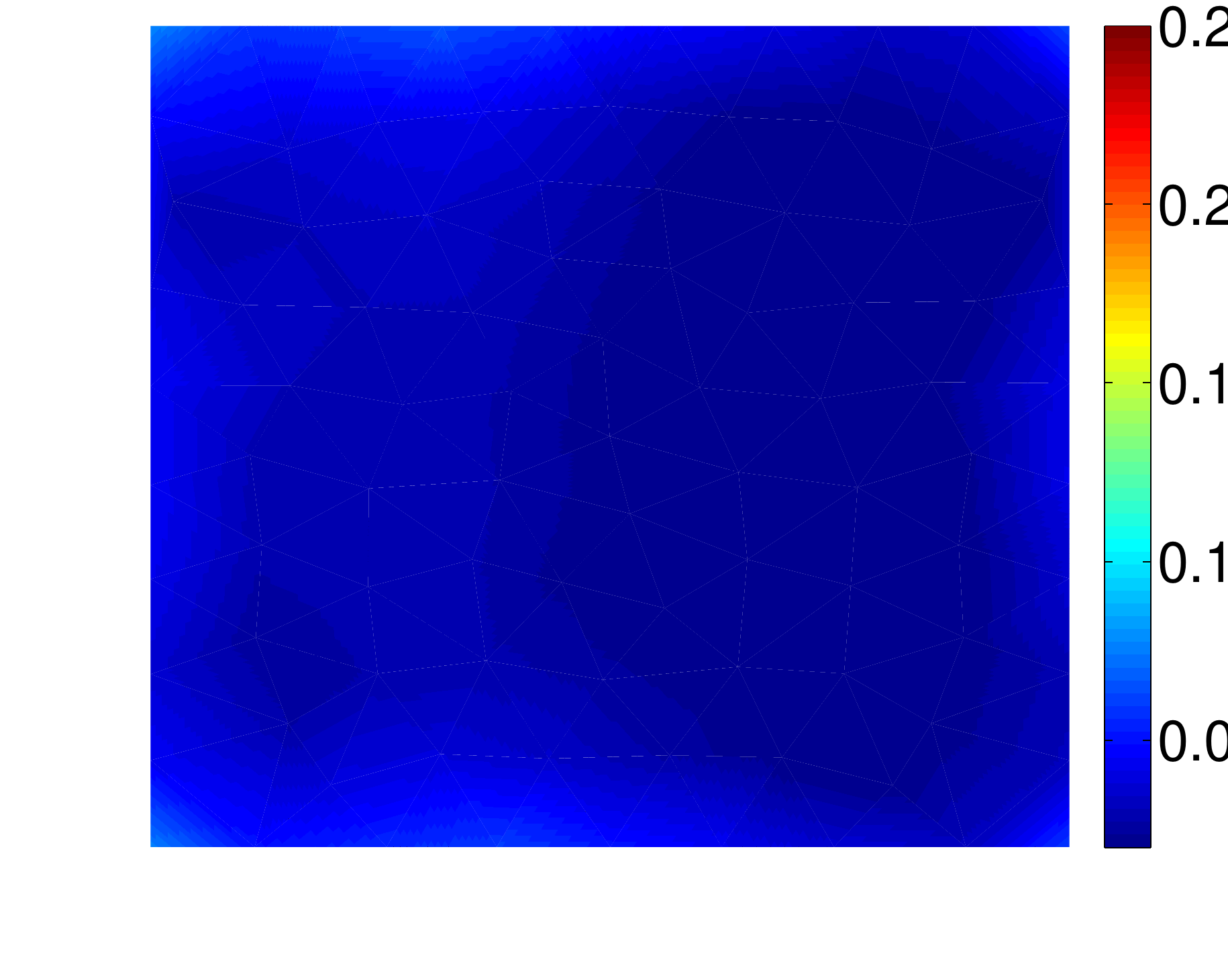}
 }\figlab{highrmlvar}
 \subfigure[Weighted-rMAP]{
\includegraphics[trim=1cm 6.0cm 2cm 7.1cm,clip=true,width=0.3\columnwidth]{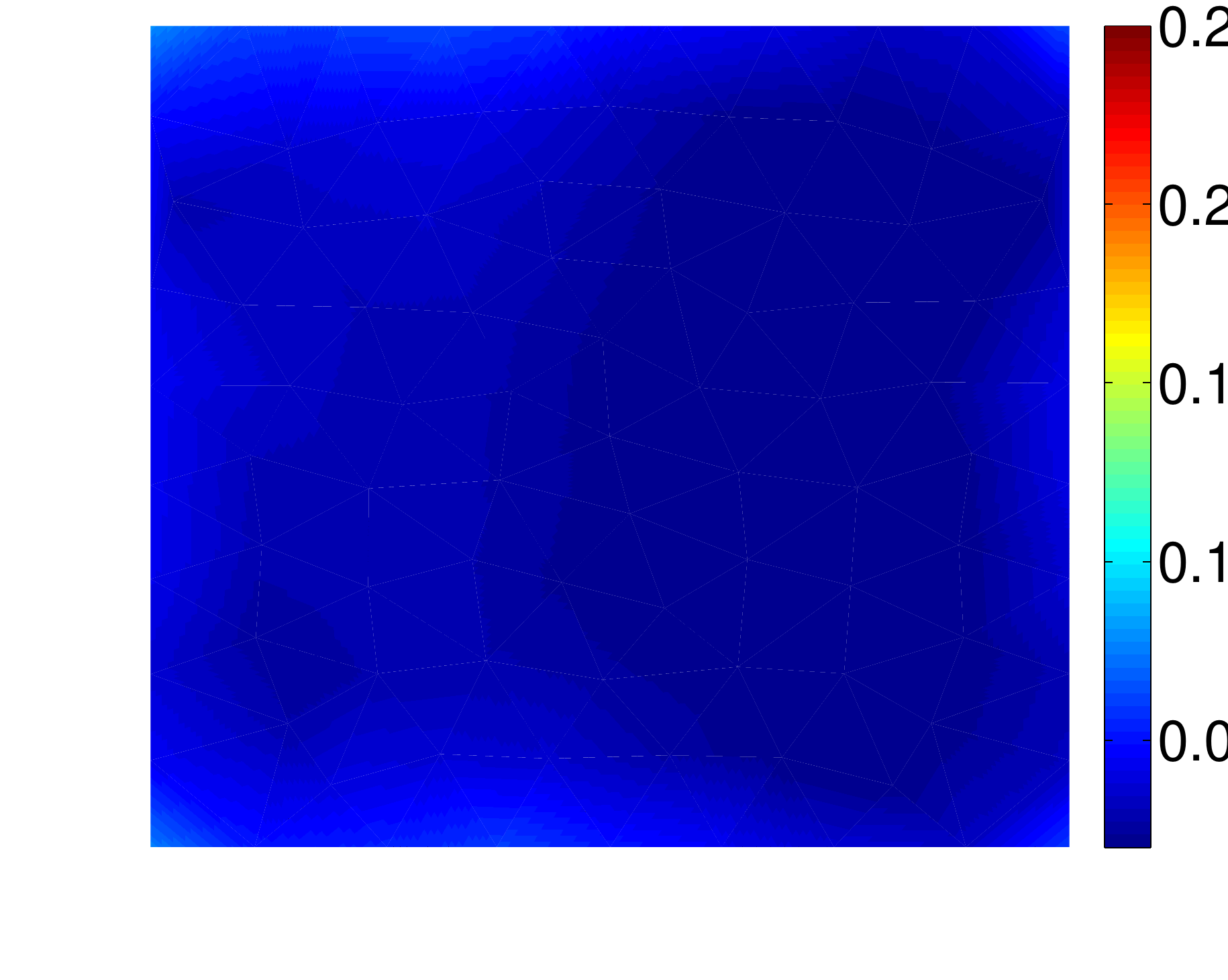}
 }\figlab{highrmlweightvar}
 \end{center}
 \caption{Comparison of estimates for the first experiment with $\alpha =
   8.0$: top row shows the conditional mean estimate. Bottom row shows comparison of variance estimation. All the rMAP samples are obtained from the TRINCG method with good initial guesses.}
 \figlab{highcomparedramrmap}
\end{figure}

\begin{figure}[h!t!b!]
\begin{center}
\subfigure[DRAM]{
\includegraphics[trim=1cm 6.0cm 2cm 7.1cm,clip=true,width=0.3\columnwidth]{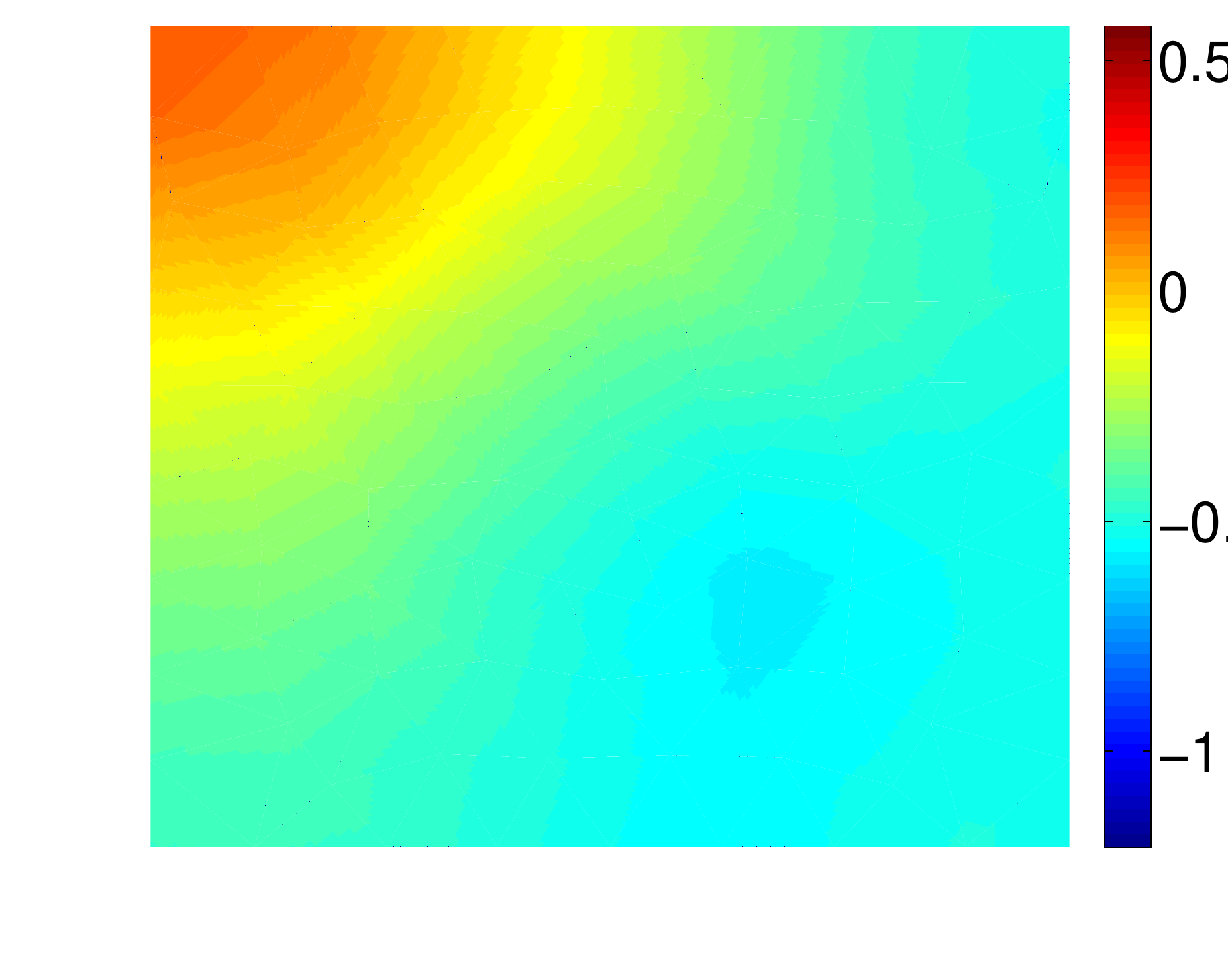}
 }\figlab{lowdrammean}
 \subfigure[rMAP]{
\includegraphics[trim=1cm 6.0cm 2cm 7.1cm,clip=true,width=0.3\columnwidth]{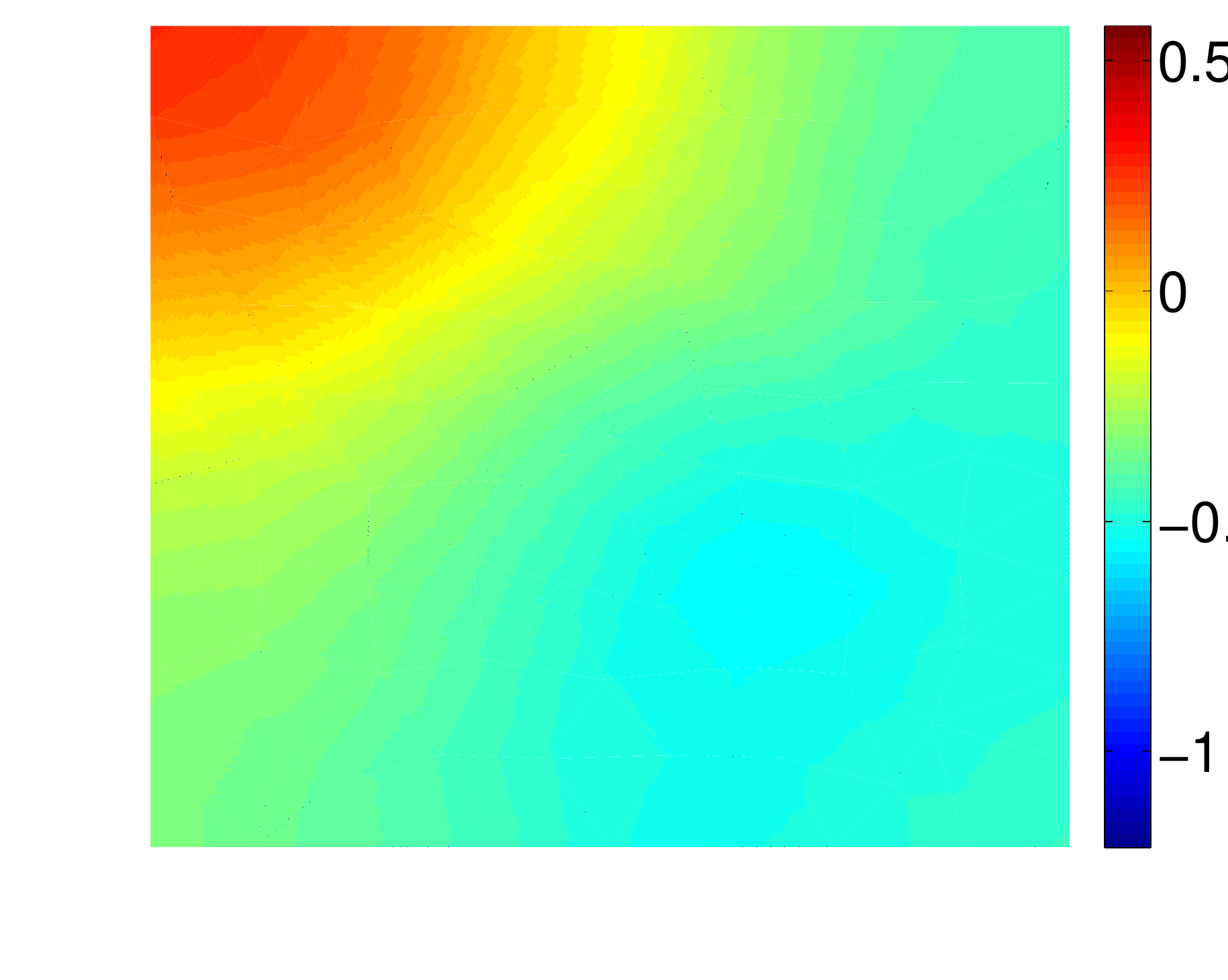}
 }\figlab{lowrmlmean}
 \subfigure[Weighted-rMAP]{
\includegraphics[trim=1cm 6.0cm 2cm 7.1cm,clip=true,width=0.3\columnwidth]{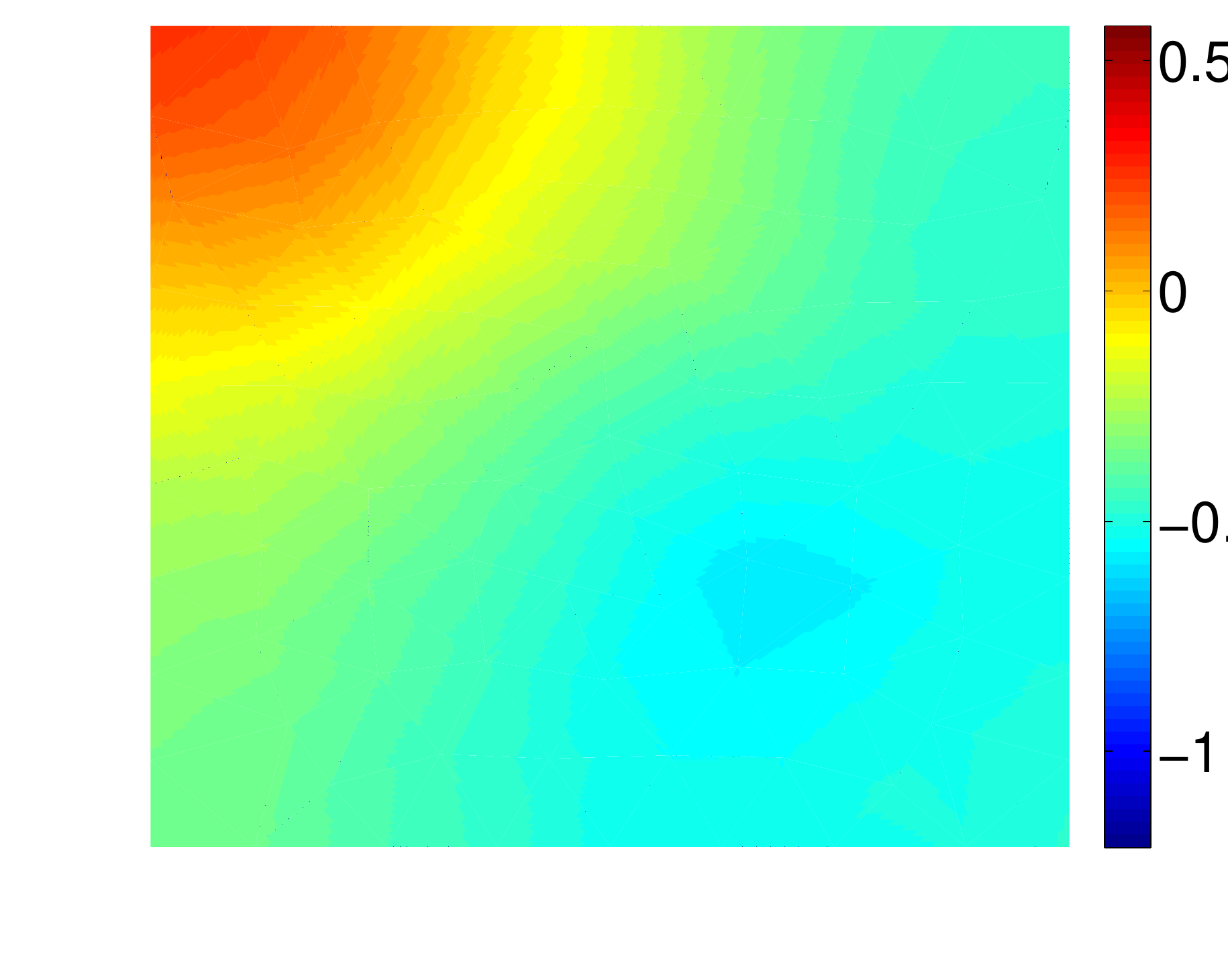}
 }\figlab{lowrmlweightmean}
\subfigure[DRAM]{
\includegraphics[trim=1cm 6.0cm 2cm 7.1cm,clip=true,width=0.3\columnwidth]{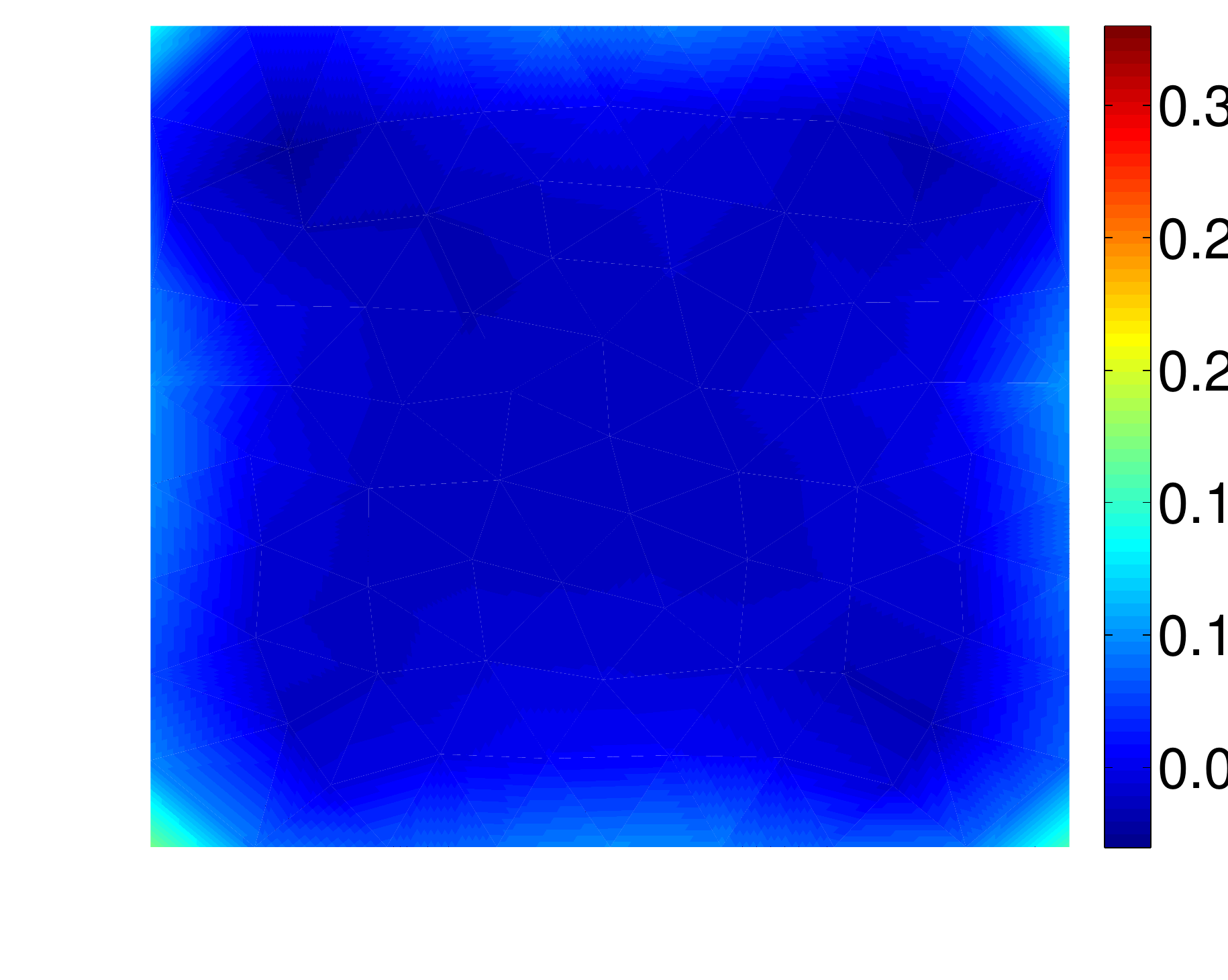}
 }\figlab{lowdramvar}
 \subfigure[rMAP]{
\includegraphics[trim=1cm 6.0cm 2cm 7.1cm,clip=true,width=0.3\columnwidth]{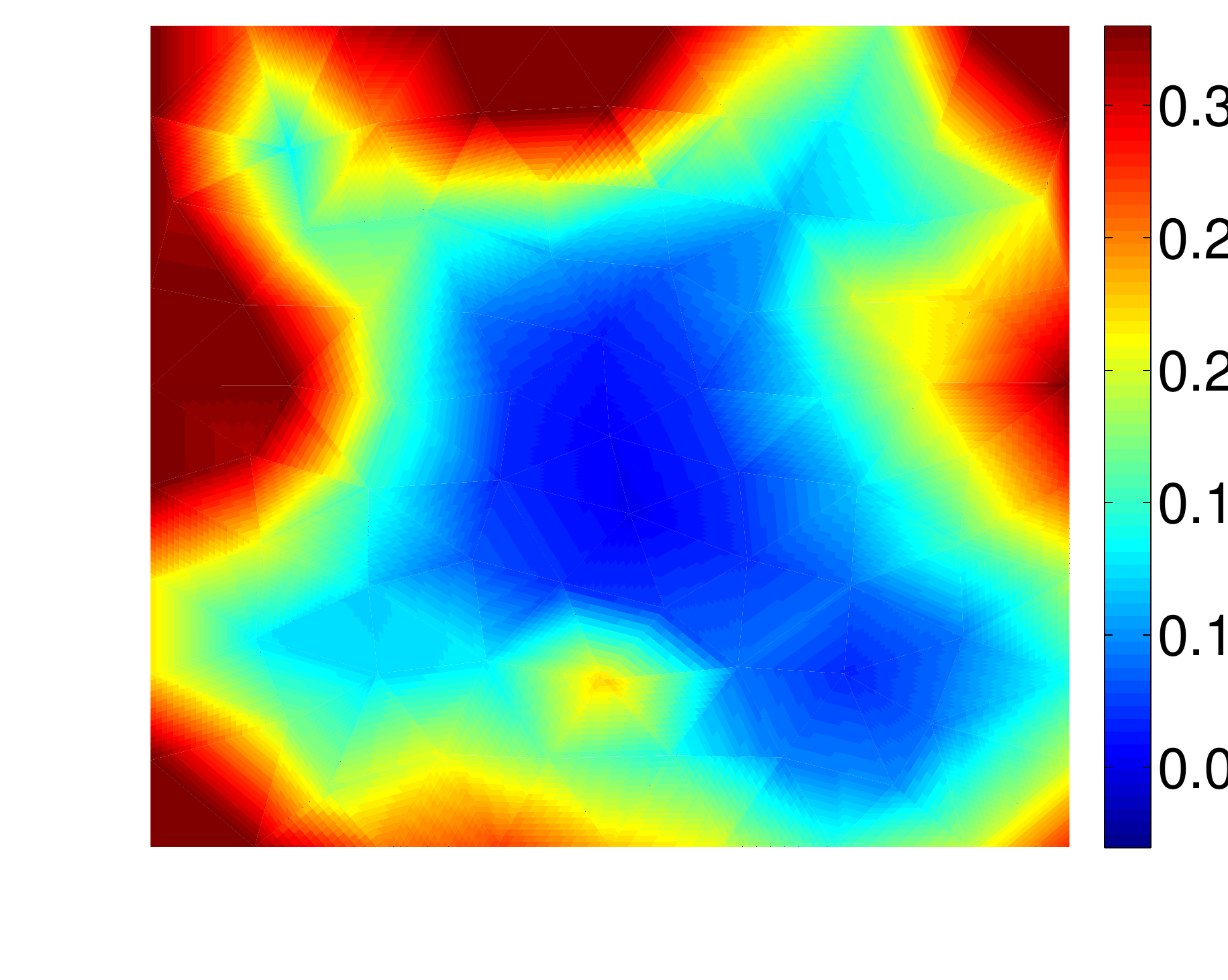}
 }\figlab{lowrmlvar}
 \subfigure[Weighted-rMAP]{
\includegraphics[trim=1cm 6.0cm 2cm 7.1cm,clip=true,width=0.3\columnwidth]{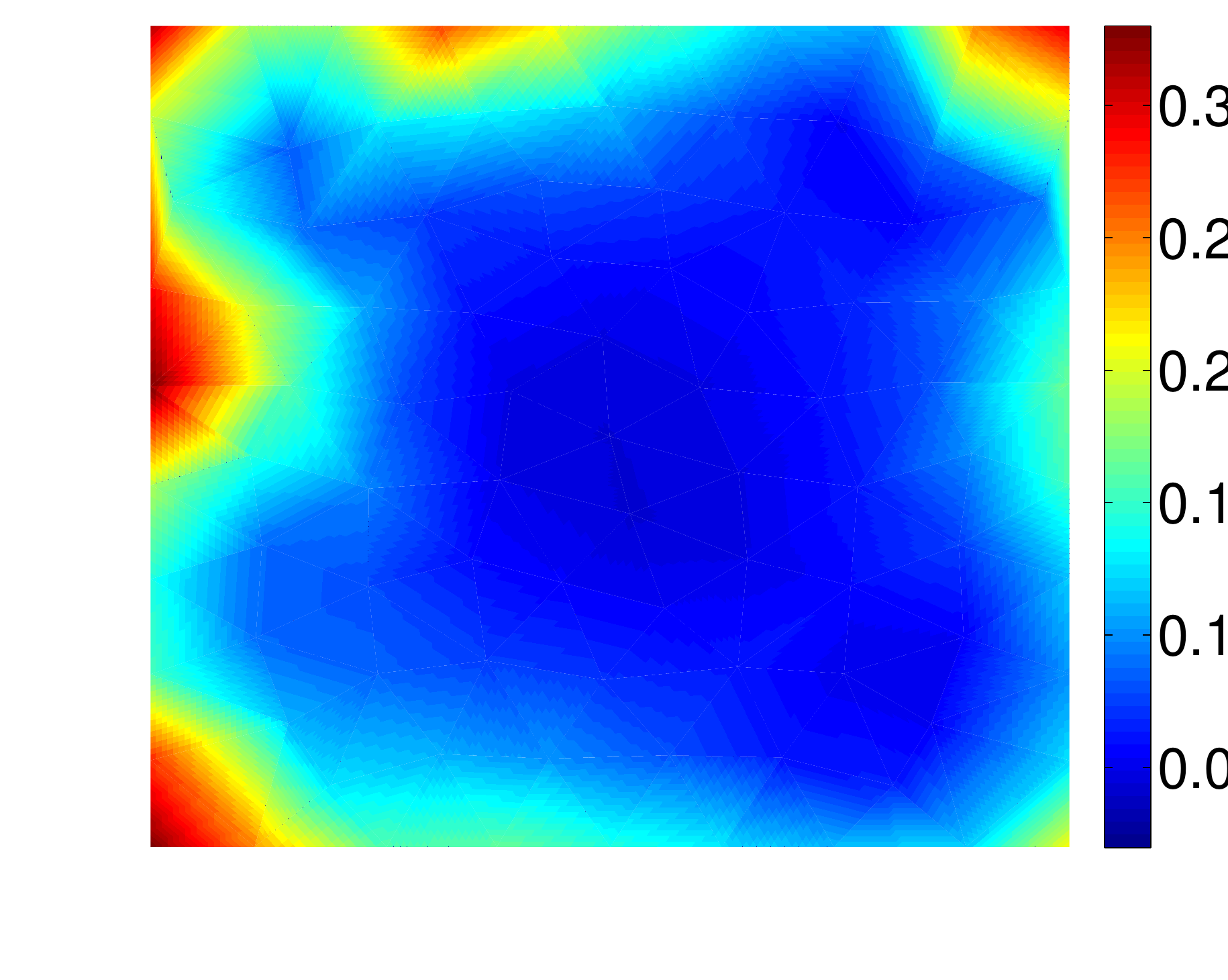}
 }\figlab{lowrmlweightvar}
 \end{center}
 \caption{Comparison of estimates for the second experiment with $\alpha =
   3.0$: top row shows the conditional mean estimate. Bottom row shows comparison of variance estimation. All the rMAP samples are obtained from the TRINCG method with good initial guesses.}
 \figlab{lowcomparedramrmap}
\end{figure}

Next, we analyze computational efficiency of the rMAP samplers. Note that the DRAM samples are highly correlated due to the large dimensionality of parameter space, meanwhile, as we will show below, rMAP samples are almost statistically independent, even for nonlinear problems. In order to compare computational performance between rMAP and DRAM as well, we utilize a concept of effective sample size (ESS) which is defined, for a sampler with a total of $n$ samples, as
\begin{equation}
 \text{ESS} = \frac{n}{\overline{\tau}},
 \eqnlab{ess}
\end{equation} 
and, for a model with $L$ parameters in total, the averaged integrated auto-correlation time (IACT) $\overline{\tau}$ is computed by
\begin{equation*}
 \overline{\tau} = \frac{1}{L}\sum_{l=1}^L\LRp{1 + 2\sum_{k=1}^{\infty}\tau(k)},
\end{equation*}
in which the auto-correlation function (ACF) $\tau(k)$ for a time series $X_t$ with mean value $\mu$ and variance $\sigma^2$ is defined as
\begin{equation*}
 \tau(k) = \frac{E\LRs{\LRp{X_t-\mu}\LRp{X_{t+k}-\mu}}}{\sigma^2}.
\end{equation*}
Since PDE solve is the most time-consuming part, we take
the total number of PDE solves (assuming the cost of solving forward,
adjoint, incremental forward, and incremental adjoint equations is
the same) as the measure of the computational cost. 

Figure \figref{iactcompare} shows the comparison of IACT for all parameters. For simplicity, we only show the IACT for rMAP samples obtained through the TRINCG together with good initial guesses. We then obtain the mean IACT's to be: $\overline{\tau}_{\text{DRAM}} = 461.90$, $\overline{\tau}_{\text{rMAP}} = 1.00$ and $\overline{\tau}_{\text{weighted-rMAP}} = 1.11$ for the first experiment, and $\overline{\tau}_{\text{DRAM}} = 564.32$, $\overline{\tau}_{\text{rMAP}} = 1.10$ and $\overline{\tau}_{\text{weighted-rMAP}} = 1.2743$ for the second experiment. Therefore, 1,000 rMAP samples are correspond to about 415,000 DRAM samples when $\alpha = 8.0$ and correspond to about 443,000 DRAM samples when $\alpha = 3.0$. As a result, for comparing computational costs in both experiments, we take into consideration 1,000 rMAP samples and 400,000 DRAM samples. 


\begin{figure}[h!t!b!]
\begin{center}
\subfigure[$\alpha = 8.0$]{
  \includegraphics[trim=1cm 6.0cm 2cm 7.1cm,clip=true,width = 0.45\columnwidth]{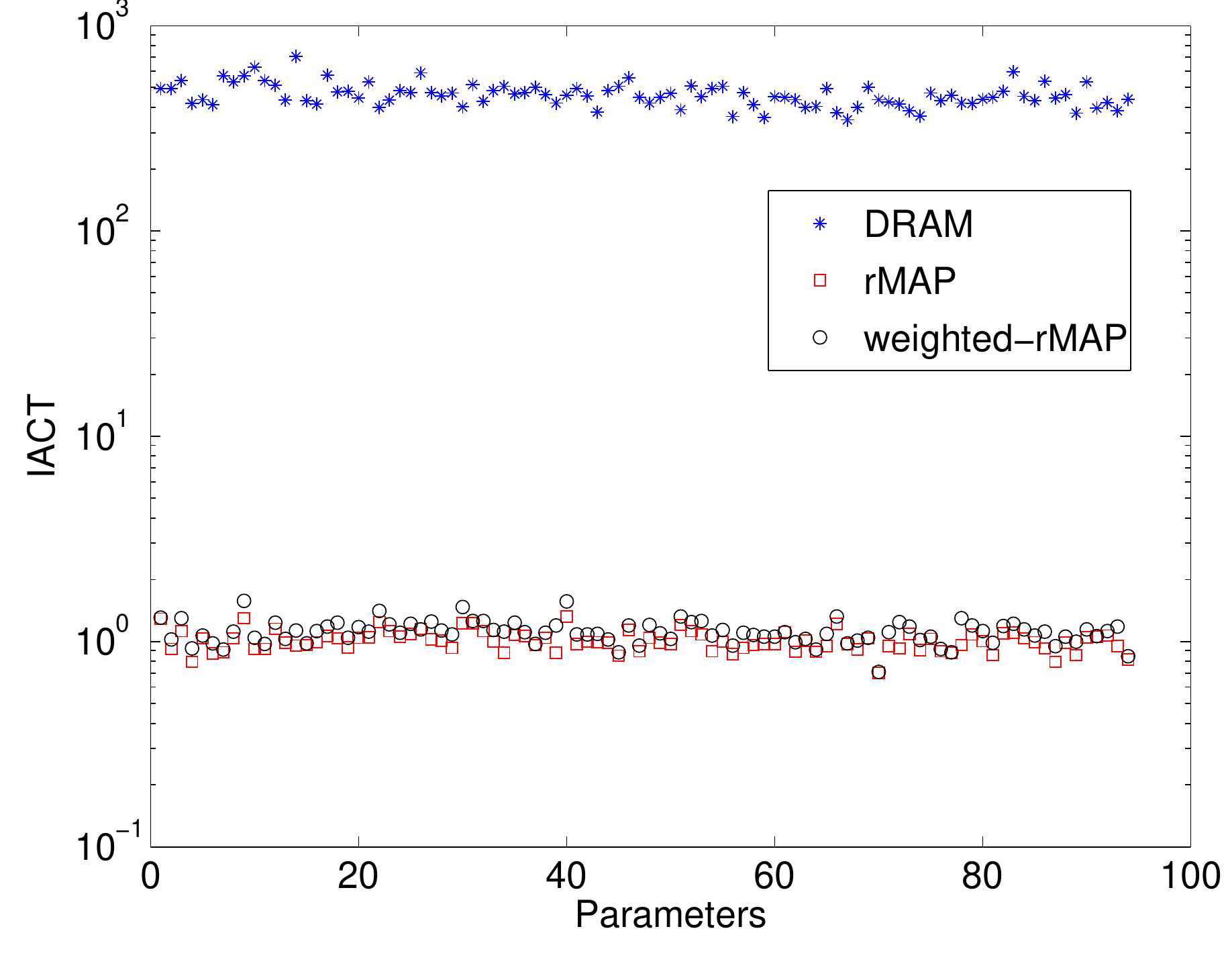}
  }
 \subfigure[$\alpha = 3.0$]{
  \includegraphics[trim=1cm 6.0cm 2cm 7.1cm,clip=true,width = 0.45\columnwidth]{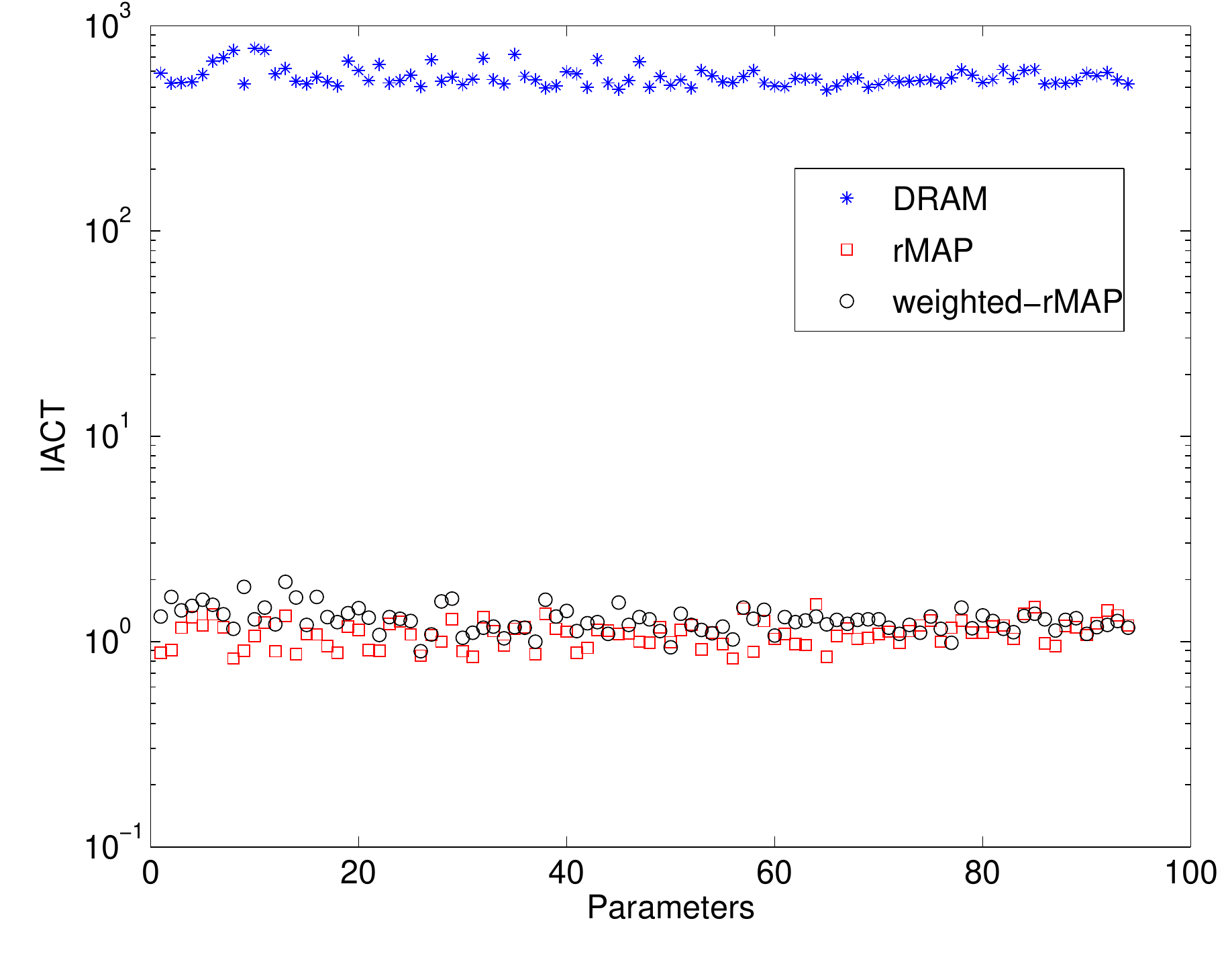}
  }

\caption{IACT for DRAM, rMAP and weighted-rMAP for all
  parameters. Samples from TRINCG with good initial guesses are used to
  compute these IACT's.}  \figlab{iactcompare}
\end{center}
\end{figure}

We compare costs of different sampling/optimization strategies in
Tables \tabref{rMAP1} and \tabref{rMAP2}. It is obvious that, compared
with the LM method, TRINCG improves efficiency both with and without a
warm-start---for example, when good initial guesses are adopted, LM is
about $60\%$ and about $320\%$ more expensive than TRINCG,
respectively. The importance of the warm-start strategy is also
salient in these tables. In particular, it speeds up the LM algorithm
significantly (at least five times) in the prior-dominated case such that
the rMAP sampler with LM performs better than DRAM with
statistically comparable number of samples. Nontheless, computational
costs of the LM method in Table \tabref{rMAP2} are more than the
corresponding DRAM sampler even with good initial
guesses, leaving TRINCG as the only tractable choice for rMAP sampling
this (``difficult'') likelihood-dominated problem.

\begin{table}[h!t!b!]
  \caption{Cost for the case $\alpha = 8.0$: the cost measured in the number
    of PDE solves in generating 1000 rMAP samples using four
    combinations: with either TRINCG or LM and with either warm-start
    strategy or not. As a comparison, the cost of DRAM sampler of
    getting 400,000 samples is shown in the last row.}
  \tablab{rMAP1}
  \begin{center}
    \begin{tabular}{|l||c|c|}
    \hline
      rMAP & good initial guess & random initial guess\\
      \hline
      \hline
      TRINCG & 208367 & 288254 \\
      \hline
      LM &  334618  & 1626279\\
      \hline
      \hline
      DRAM & 732917 & \\
      \hline
    \end{tabular}
  \end{center}
\end{table}

\begin{table}[h!t!b!]
  \caption{Cost for the case $\alpha = 3.0$: the cost measured in the number
    of PDE solves in generating 1000 rMAP samples using four
    combinations: with either TRINCG or LM and with either warm-start
    strategy or not. As a comparison, the cost of DRAM sampler of
    getting 400,000 samples is shown in the last row.}
  \tablab{rMAP2}
  \begin{center}
    \begin{tabular}{|l||c|c|}
    \hline
      rMAP & good initial guess & random initial guess\\
      \hline
      \hline
      TRINCG & 511956 & 568671 \\
      \hline
      LM &  1639601  & 2705973\\
      \hline
      \hline
      DRAM & 706017 &\\
      \hline
    \end{tabular}
  \end{center}
\end{table}

\section{Conclusions}
\seclab{conclusions} In this paper we present a randomized {\em
  maximum a posteriori} (rMAP) approach to approximately sample
posteriors of nonlinear Bayesian inverse problems in high dimensional
parameter spaces. The idea is to cast the standard MAP computation as
a stochastic optimization problem and use the sample average approach
to approximate the expectation.  We have shown that the randomized
maximum likelihood method is a special case of the proposed rMAP
method. The stochastic programming view point allows us to provide
additional theoretical results, in both finite and infinite dimensions
and for both linear and nonlinear inverse problems, leading to a
better understanding of rMAP. The appeal of the proposed approach is
that each rMAP sample requires solution of a PDE-constrained
optimization problem which can be carried out efficiently using a
trust region inexact Newton conjugate gradient method. To further
reduce the cost of each rMAP sample, we develop a warm start strategy
using sensitivity analysis via an efficient adjoint technique.
Viewing rMAP as an iterative stochastic Newton method reveals that
rMAP is in fact a move away from the inefficiencies of
random-walk/diffusion processes altogether, toward powerful
optimization methods that use derivative information to traverse the
posterior.

We have made a connection between the rMAP approach and a closely
related randomize-then-optimize method. We show that they are
identical for linear inverse problems but different for nonlinear
ones. Since rMAP samples are approximate samples of the posterior, we
present an approximate Metropolization to reduce the bias. FEM discretization of the infinite dimensional Bayesian inverse problem, solving optimization problems at each sampling step with the trust region inexact Newton
conjugate gradient method, as well as a sensitivity analysis based warm start strategy are also discussed. Analytical and numerical
experiments are presented to confirm various theoretical results and
demonstrate the potential of the rMAP approach for difficult nonlinear
Bayesian inverse problems.

\bibliography{ccgo,omar}

\begin{thebibliography}{10}

\bibitem{BardsleySolonenHaarioEtAl13}
{\sc Johnathan Bardsley, Antti Solonen, Heikki Haario, and Marko Laine}, {\em
  Randomize-then-optimize: A method for sampling from posterior distributions
  in nonlinear inverse problems}, submitted,  (2013).

\bibitem{BeskosPinskiSanz-SernaEtAl11}
{\sc A.~Beskos, F.~J. Pinski, J.~M. Sanz-Serna, and A.~M. Stuart}, {\em Hybrid
  {M}onte {C}arlo on {H}ilbert spaces}, Stochastic Processes and their
  Applications, 121 (2011), pp.~2201--2230.

\bibitem{BieglerGhattasHeinkenschlossEtAl03}
{\sc L.~T. Biegler, O.~Ghattas, M.~Heinkenschloss, and B.~van
  Bloemen~Waanders}, eds., {\em Large-Scale {PDE}-Constrained Optimization},
  Lecture Notes in Computational Science and Engineering, Vol.~30,
  Springer-Verlag, Heidelberg, 2003.

\bibitem{BorziSchulz12}
{\sc Alfio Borz{\`{\i}} and Volker Schulz}, {\em Computational Optimization of
  Systems Governed by Partial Differential Equations}, SIAM, 2012.

\bibitem{BranchColemanLi99}
{\sc Mary~Ann Branch, Thomas~F. Coleman, and Yuying Li}, {\em A subspace,
  interior, and conjugate gradient method for large-scale bound-constrained
  minimization problems}, SIAM Journal on Scientific Computing, 21 (1999),
  pp.~1--23 (electronic).

\bibitem{Brockwell06}
{\sc AE~Brockwell}, {\em Parallel markov chain monte carlo simulation by
  pre-fetching}, Journal of Computational and Graphical Statistics, 15 (2006),
  pp.~246--261.

\bibitem{Bui-Thanh07}
{\sc Tan Bui-Thanh}, {\em Model-Constrained Optimization Methods for Reduction
  of Parameterized Large-Scale Systems}, PhD thesis, Department of Aeronautics
  and Astronautics, MIT, 2007.

\bibitem{Bui-Thanh2015}
\leavevmode\vrule height 2pt depth -1.6pt width 23pt, {\em
  Discretization-invariant {MCMC} methods for {PDE}-constrained {B}ayesian
  inverse problems in infinite dimensional parameter spaces}, Submitted,
  (2015).

\bibitem{Bui-ThanhBursteddeGhattasEtAl12}
{\sc Tan Bui-Thanh, Carsten Burstedde, Omar Ghattas, James Martin, Georg
  Stadler, and Lucas~C. Wilcox}, {\em {Extreme-scale UQ for Bayesian inverse
  problems governed by PDEs}}, in SC12: Proceedings of the International
  Conference for High Performance Computing, Networking, Storage and Analysis,
  2012.

\bibitem{Bui-ThanhGhattas12a}
{\sc Tan Bui-Thanh and Omar Ghattas}, {\em Analysis of the {H}essian for
  inverse scattering problems. {P}art {I}: Inverse shape scattering of acoustic
  waves}, Inverse Problems, 28 (2012), p.~055001.

\bibitem{Bui-ThanhGhattas12}
\leavevmode\vrule height 2pt depth -1.6pt width 23pt, {\em Analysis of the
  {H}essian for inverse scattering problems. {P}art {II}: Inverse medium
  scattering of acoustic waves}, Inverse Problems, 28 (2012), p.~055002.

\bibitem{Bui-ThanhGhattas12d}
\leavevmode\vrule height 2pt depth -1.6pt width 23pt, {\em A scaled stochastic
  {N}ewton algorithm for {M}arkov chain {M}onte {C}arlo simulations}, Submitted
  to SIAM Journal of Uncertainty Quantification,  (2012).

\bibitem{Bui-ThanhGhattas12f}
\leavevmode\vrule height 2pt depth -1.6pt width 23pt, {\em Analysis of the
  {H}essian for inverse scattering problems. {P}art {III}: Inverse medium
  scattering of electromagnetic waves}.
\newblock Inverse Problems and Imaging, 2013.

\bibitem{Bui-ThanhGhattasMartinEtAl13}
{\sc Tan Bui-Thanh, Omar Ghattas, James Martin, and Georg Stadler}, {\em A
  computational framework for infinite-dimensional {B}ayesian inverse problems
  {P}art {I}: {T}he linearized case, with application to global seismic
  inversion}, SIAM Journal on Scientific Computing, 35 (2013),
  pp.~A2494--A2523.

\bibitem{Bui-ThanhGirolami14}
{\sc Tan Bui-Thanh and Mark~Andrew Girolami}, {\em Solving large-scale
  {PDE}-constrained {B}ayesian inverse problems with {R}iemann manifold
  {H}amiltonian {M}onte {C}arlo}, Inverse Problems, Special Issue (2014),
  p.~114014.

\bibitem{Byrd10}
{\sc Jonathan Byrd}, {\em Parallel Markov Chain Monte Carlo}, PhD thesis,
  University of Warwick, 2010.

\bibitem{CalvettiSomersalo07}
{\sc D.~Calvetti and E.~Somersalo}, {\em Introduction to Bayesian Scientific
  Computing: Ten Lectures on Subjective Computing}, Springer, New York, 2007.

\bibitem{Ciarlet78}
{\sc P.~G. Ciarlet}, {\em The Finite Element Method for Elliptic Problems},
  North--Holland, Amsterdam, New York, 1978.

\bibitem{ColemanLi96}
{\sc T.~F. Coleman and Y.~Li}, {\em An interior trust region approach for
  nonlinear minimization subject to bounds}, SIAM Journal on Optimization, 6
  (1996), pp.~418--445.

\bibitem{CotterRobertsStuartEtAl13}
{\sc S.~L. Cotter, G.~O. Roberts, A.~M. Stuart, and D.~White}, {\em {MCMC}
  methods for functions: modifying old algorithms to make them faster},
  Statistical Science, 28 (2013), pp.~424--446.

\bibitem{CuiLawMarzouk16}
{\sc Tiangang Cui, Kody~JH Law, and Youssef~M Marzouk}, {\em
  Dimension-independent likelihood-informed mcmc}, Journal of Computational
  Physics, 304 (2016), pp.~109--137.

\bibitem{CuiMartinMarzoukEtAl14}
{\sc Tiangang Cui, James Martin, Youssef~M Marzouk, Antti Solonen, and Alessio
  Spantini}, {\em Likelihood-informed dimension reduction for nonlinear inverse
  problems}, arXiv preprint arXiv:1403.4680,  (2014).

\bibitem{Dashti2013}
{\sc M.~Dashti, K.J.H. Law, A.M. Stuart, and J.~Voss}, {\em {MAP} estimators
  and their consistency in {B}ayesian nonparametric inverse problems}, Inverse
  Problems, 29 (2013), p.~095017.

\bibitem{Reyes15}
{\sc Juan~Carlos De~los Reyes}, {\em Numerical PDE-constrained optimization},
  Springer, 2015.

\bibitem{DuaneKennedyPendletonEtAl87}
{\sc S.~Duane, A.~D. Kennedy, B.~Pendleton, and D.~Roweth}, {\em Hybrid {M}onte
  {C}arlo}, Phys. Lett. B, 195 (1987), pp.~216--222.

\bibitem{Franklin70}
{\sc J.~N. Franklin}, {\em Well-posed stochastic extensions of ill--posed
  linear problems}, Journal of Mathematical Analysis and Applications, 31
  (1970), pp.~682--716.

\bibitem{GirolamiCalderhead11}
{\sc Mark Girolami and Ben Calderhead}, {\em Riemann manifold {L}angevin and
  {H}amiltonian {M}onte {C}arlo methods}, Journal of the Royal Statistical
  Society: Series B (Statistical Methodology), 73 (2011), pp.~123--214.

\bibitem{GolubVan96}
{\sc Gene~H. Golub and Charles~F. Van~Loan}, {\em Matrix Computations}, Johns
  Hopkins Studies in the Mathematical Sciences, Johns Hopkins University Press,
  Baltimore, MD, third~ed., 1996.

\bibitem{HaarioLaineMiraveteEtAl06}
{\sc Heikki Haario, Marko Laine, Antonietta Miravete, and Eero Saksman}, {\em
  {DRAM: Efficient adaptive MCMC}}, Statistics and Computing, 16 (2006),
  pp.~339--354.

\bibitem{Hairer09}
{\sc Martin Hairer}, {\em Introduction to {S}tochastic {PDE}s}.
\newblock Lecture Notes, 2009.

\bibitem{HalkoMartinssonTropp11}
{\sc Nathan Halko, Per-Gunnar Martinsson, and Joel~A. Tropp}, {\em Finding
  structure with randomness: {P}robabilistic algorithms for constructing
  approximate matrix decompositions}, SIAM Review, 53 (2011), pp.~217--288.

\bibitem{Hastings70}
{\sc W.~Keith Hastings}, {\em Monte {C}arlo sampling methods using {M}arkov
  chains and their applications}, Biometrika, 57 (1970), pp.~97--109.

\bibitem{HinzePinnauUlbrichEtAl09}
{\sc Michael Hinze, Rene Pinnau, Michael Ulbrich, and Stefan Ulbrich}, {\em
  Optimization with PDE Constraints}, Springer, 2009.

\bibitem{IglesiasLawStuart12}
{\sc Marco~A. Iglesias, Kody J.~H. Law, and Andrew~M. Stuart}, {\em Evaluation
  of gaussian approximations for data assimilation in reservoir models},
  Submitted,  (2012).

\bibitem{IlicLiuTurnerEtAl05}
{\sc M.~Ili\'c, F.~Liu, I.~Turner, and V.~Anh}, {\em Numerical approximation of
  a fractional-in-space diffusion equation}, Frac. Calc. and App. Anal., 8
  (2005), pp.~323--341.

\bibitem{IsaacPetraStadlerEtAl15}
{\sc Tobin Isaac, Noemi Petra, Georg Stadler, and Omar Ghattas}, {\em Scalable
  and efficient algorithms for the propagation of uncertainty from data through
  inference to prediction for large-scale problems, with application to flow of
  the antarctic ice sheet}, Journal of Computational Physics, 296 (2015),
  pp.~348--368.

\bibitem{KaipioSomersalo05}
{\sc Jari Kaipio and Erkki Somersalo}, {\em Statistical and Computational
  Inverse Problems}, vol.~160 of Applied Mathematical Sciences,
  Springer-Verlag, New York, 2005.

\bibitem{Kitanidis95}
{\sc P.~K. Kitanidis}, {\em Quasi-linear geostatistical theory for inversing},
  Water Resour. Res., 31 (1995), pp.~2411--2419.

\bibitem{Kitanidis96}
{\sc P.~K. Kitanidis}, {\em On the geostatistical approach to the inverse
  problem}, Advances in Water Resources, 19 (1996), pp.~333--342.

\bibitem{Lasanen02}
{\sc S.~Lasanen}, {\em Discretizations of generalized random variables with
  applications to inverse problems}, PhD thesis, University of Oulu, 2002.

\bibitem{LehtinenPaivarintaSomersalo89}
{\sc Markku~S. Lehtinen, Lassi P\"{a}iv\"{a}rinta, and Erkki Somersalo}, {\em
  Linear inverse problems for generalized random variables}, Inverse Problems,
  5 (1989), pp.~599--612.

\bibitem{MartinWilcoxBursteddeEtAl12}
{\sc James Martin, Lucas~C. Wilcox, Carsten Burstedde, and Omar Ghattas}, {\em
  A stochastic {Newton MCMC} method for large-scale statistical inverse
  problems with application to seismic inversion}, SIAM Journal on Scientific
  Computing, 34 (2012), pp.~A1460--A1487.

\bibitem{MetropolisRosenbluthRosenbluthEtAl53}
{\sc Nicholas Metropolis, Arianna~W. Rosenbluth, Marshall~N. Rosenbluth,
  Augusta~H. Teller, and Edward Teller}, {\em Equation of state calculations by
  fast computing machines}, The Journal of Chemical Physics, 21 (1953),
  pp.~1087--1092.

\bibitem{Neal10}
{\sc R.~M. Neal}, {\em Handbook of Markov Chain Monte Carlo}, Chapman \& Hall /
  CRC Press, 2010, ch.~MCMC using Hamiltonian dynamics.

\bibitem{NocedalWright06}
{\sc Jorge Nocedal and Stephen~J. Wright}, {\em Numerical Optimization},
  Springer Verlag, Berlin, Heidelberg, New York, second~ed., 2006.

\bibitem{Oliver14}
{\sc Dean~S. Oliver}, {\em Minimization for conditional simulation:
  relationship to optimal transport}, Journal of Computational Physics, 265
  (2014), pp.~1--15.

\bibitem{OliverHeReynolds96}
{\sc D.~S. Oliver, H.~He, and A.~C. Reynolds}, {\em Conditioning permeability
  fields to pressure data}, in European Conference for the Mathematics of Oil
  Recovery, 1996, pp.~1--11.

\bibitem{OliverReynoldsLiu08}
{\sc Dean~S. Oliver, Albert~C. Reynolds, and Ning Liu}, {\em Inverse theory for
  petroleum reservoir characterization and history matching}, Cambidge
  University Press, 2008.

\bibitem{PetraMartinStadlerEtAl14}
{\sc Noemi Petra, James Martin, Georg Stadler, and Omar Ghattas}, {\em A
  computational framework for infinite-dimensional {B}ayesian inverse problems:
  {P}art {II}. {S}tochastic {N}ewton {MCMC} with application to ice sheet
  inverse problems}, SIAM Journal on Scientific Computing (to appear),  (2014).
\newblock 

\bibitem{Piiroinen05}
{\sc P.~Piiroinen}, {\em Statistical measurements, experiments, and
  applications}, PhD thesis, Department of Mathematics and Statistics,
  University of Helsinki, 2005.

\bibitem{PratoZabczyk92}
{\sc Giuseppe~Da Prato and Jerzy Zabczyk}, {\em Stochastic Equations in
  Infinite Dimensions}, Cambidge University Press, 1992.

\bibitem{RobertCasella05}
{\sc Christian~P. Robert and George Casella}, {\em Monte Carlo Statistical
  Methods (Springer Texts in Statistics)}, Springer-Verlag New York, Inc.,
  Secaucus, NJ, USA, 2005.

\bibitem{RobertsRosenthal97}
{\sc Gareth~O. Roberts and Jeffrey~S. Rosenthal}, {\em Optimal scaling of
  discrete approximations to {L}angevin diffusions}, J. R. Statist. Soc. B, 60
  (1997), pp.~255--268.

\bibitem{RobertsTweedie96}
{\sc Gareth~O. Roberts and Richard~L. Tweedie}, {\em Exponential convergence of
  {L}angevin distributions and their discrete approximations}, Bernoulli, 2
  (1996), pp.~341--363.

\bibitem{RockafellarWetts98}
{\sc R.~T. Rockafellar and R.~J.-B. Wetts}, {\em Variational Analysis},
  Springer Verlag, Berlin, Heidelberg, New York, 1998.

\bibitem{ShapiroDentchevaRuszczynski09}
{\sc A.~Shapiro, D.~Dentcheva, and A.~Ruszczynski}, {\em Lectures on Stochastic
  Programming: Modeling and Theory}, Society for Industrial and Applied
  Mathematics, 2009.

\bibitem{Strid10}
{\sc Ingvar Strid}, {\em Efficient parallelisation of metropolis--hastings
  algorithms using a prefetching approach}, Computational Statistics \& Data
  Analysis, 54 (2010), pp.~2814--2835.

\bibitem{Stuart10}
{\sc Andrew~M. Stuart}, {\em Inverse problems: {A B}ayesian perspective}, Acta
  Numerica, 19 (2010), pp.~451--559.

\bibitem{Wang14}
{\sc Kainan Wang}, {\em Parallel Markov Chain Monte Carlo Methods for Large
  Scale Statistical Inverse Problems}, PhD thesis, Texas A\&M University, 2014.

\bibitem{Wilkinson05}
{\sc D.~J. Wilkinson}, {\em Handbook of Parallel Computing and Statistics},
  Marcel Dekker/CRC Press, 2005, ch.~Parallel Bayesian Computation,
  pp.~481--512.

\end{thebibliography}

\end{document}